\documentclass[pra,twocolumn,superscriptaddress]{revtex4-1}%
\usepackage{amsfonts}
\usepackage{mathtools}
\usepackage{amsmath}
\usepackage{amssymb}
\usepackage{dsfont}
\usepackage{graphicx}
\usepackage{xcolor}
\usepackage{pgfplots}
\usepackage{physics}
\usepackage{tikz}
\usepackage[colorlinks=true,linkcolor=blue,citecolor=red,plainpages=false,pdfpagelabels]
 {hyperref}
\usepackage{cleveref}
\usepackage{color}
\definecolor{dgreen}{HTML}{006600}

\providecommand{\U}[1]{\protect\rule{.1in}{.1in}}
\newtheorem{theorem}{Theorem}

\newtheorem{corollary}{Corollary}

\newtheorem{definition}{Definition}
\newtheorem{example}{Example}

\newtheorem{lemma}{Lemma}

\newtheorem{proposition}{Proposition}
\newtheorem{remark}{Remark}

\allowdisplaybreaks

\def\Tr{\operatorname{Tr}}
\def\d{\operatorname{d}}
\def\kex{\operatorname{EXT_{k}}}

\def\SEP{\operatorname{SEP}}
\def\EXT{\operatorname{EXT}}
\def\Ent{\operatorname{Ent}}

\def\supp{\operatorname{supp}}

\def\T{\operatorname{T}}

\def\({\left(}
\def\){\right)}
\def\[{\left[}
\def\]{\right]}
\def\V{\Vert}
\def\id{\operatorname{id}}

\let\emptyset\varnothing

\newcommand{\bb}[1]{\mathbb{#1}}

\newcommand{\mc}[1]{\mathcal{#1}}
\newcommand{\wt}[1]{\widetilde{#1}}

\newcommand{\tf}[1]{\mathbf{#1}}
\newenvironment{proof}[1][Proof]{\noindent\textbf{#1.} }{\ \rule{0.5em}{0.5em}}
\begin{document}
\preprint{ }
\title[]{Resource theory of unextendibility and non-asymptotic quantum capacity}
%\title[ ]{Limitations on small quantum processors from a resource theory of unextendibility}

\author{Eneet Kaur}
\email{ekaur1@lsu.edu}
\affiliation{Hearne Institute for Theoretical Physics, Department of Physics and Astronomy, and Center for Computation and Technology,
Louisiana State University, Baton Rouge, Louisiana 70803, USA}
\affiliation{Institute for Quantum Computing and Department of Physics and Astronomy, University of Waterloo, Waterloo, Ontario N2L 3G1, Canada}
\author{Siddhartha Das}
\email{sidddas@ulb.ac.be}
\affiliation{Hearne Institute for Theoretical Physics, Department of Physics and Astronomy, and Center for Computation and Technology,
Louisiana State University, Baton Rouge, Louisiana 70803, USA}
\affiliation{Centre for Quantum Information \& Communication (QuIC), \'{E}cole polytechnique de Bruxelles,   Universit\'{e} libre de Bruxelles, Brussels, B-1050, Belgium}
\author{Mark M. Wilde}
\email{mwilde@lsu.edu}
\affiliation{Hearne Institute for Theoretical Physics, Department of Physics and Astronomy, and Center for Computation and Technology,
Louisiana State University, Baton Rouge, Louisiana 70803, USA}
\author{Andreas Winter}
\email{andreas.winter@uab.cat}
\affiliation{ICREA \&{} F\'{\i}sica 
 Te\`{o}rica: Informaci\'{o} i Fen\`{o}mens Qu\`{a}ntics, 
 Departament de F\'{\i}sica, Universitat Aut\`{o}noma de Barcelona, 
 ES-08193 Bellaterra (Barcelona), Spain}
\keywords{one two three}
\pacs{PACS number}

\begin{abstract}
In this paper, we introduce the resource theory of unextendibility as a relaxation of the resource theory of entanglement. The free states in this resource theory are the $k$-extendible states, associated with the inability to extend quantum entanglement in a given quantum state to multiple parties. The free channels are $k$-extendible channels, which preserve the class of $k$-extendible states. We define several quantifiers of unextendibility by means of generalized divergences and establish their properties. By utilizing this resource theory, we obtain non-asymptotic upper bounds on the rate at which quantum communication or entanglement preservation is possible over a finite number of uses of an arbitrary quantum channel assisted by $k$-extendible channels at no cost. These bounds are significantly tighter than previously known bounds for both the depolarizing and erasure channels. Finally, we revisit the pretty strong converse for the quantum capacity of antidegradable channels and establish an upper bound on the non-asymptotic quantum capacity of these channels.

\end{abstract}
\volumeyear{year}
\volumenumber{number}
\issuenumber{number}
\eid{identifier}
\date{\today}
\startpage{1}
\endpage{10}

\maketitle
\tableofcontents 

\pagebreak

\section{Introduction}

In quantum information theory, an important task is to quantify the amount of entanglement that a sender Alice and a receiver Bob can share after using a quantum channel $\mathcal{N}$ a large number of times. That is, if Alice sends one share of a bipartite state $\rho_{A^nA'^n}$ over $n$ uses of a quantum channel, then what is the amount of entanglement that can be transmitted from Alice to Bob? One then considers three variations of the above task depending on the classical communication that can be employed by Alice and Bob to assist their task. In the first one, Alice and Bob are not allowed to employ classical communication (the unassisted case). This is referred to as unassisted entanglement transmission. In the second case, Alice is allowed to communicate classically with Bob for free. In the third variation, Alice and Bob are allowed two-way classical communication for free. In the asymptotic regime of many channel uses, the entanglement transmission capacity of a channel assisted by one-way classical communication is equal to its unassisted entanglement transmission capacity \cite{BDSW96,BKN98}.

Since obtaining the exact capacities for these tasks can be challenging, one important goal is to obtain tight upper bounds on the rates for these tasks, in order to understand the basic limitations of quantum communication. In this context, \cite{WFD17,TBR15} have obtained upper bounds for finite $n$; however, these hold for entanglement transmission assisted by two-way classical communication. Therefore, we do not expect them to be tight for the unassisted entanglement transmission or entanglement transmission assisted by one-way classical communication (1W-LOCC).

In this paper, we develop the details of the resource theory of unextendibility, which was proposed in our earlier companion paper \cite{KDWW19}). As mentioned previously, this resource theory is a semi-definite relaxation of the resource theory of entanglement and thus is connected to fundamental aspects of quantum mechanics. Furthermore, we put the resource theory of unextendibility to use by obtaining bounds on the rates at which entanglement can be transmitted over a quantum channel  assisted by 1W-LOCC. We obtain these upper bounds by defining and employing monotones in the resource theory of unextendibility. What we find here is that these bounds are significantly tighter than bounds previously obtained in \cite{WFD17,TBR15}, primarily because they are tailored to hold for entanglement transmission with the assistance of 1W-LOCC.

For every integer $k \geq 2$, there is a resource theory of $k$-unextendibility, and each of these can be understood as a relaxation of the resource theory of entanglement \cite{BDSW96,HHHH09}. The free states in the resource theory of $k$-unextendibility are the $k$-extendible states \cite{W89a,DPS02,DPS04}, and the free channels are the $k$-extendible channels, which we define in Section~\ref{sec:resource-theory}. These $k$-unextendible resource theories have a hierarchical structure, with the $k$-unextendible resource theory being contained in the $(k-1)$-unextendible resource theory. By ``contained in the resource theory,'' we mean that the free states in the $k$-unextendible resource theory are free states in the $(k-1)$-unextendible resource theory. This implies that the separable states are free states for all $k$-unextendible resource theories. A similar structure is observed for the free channels. The resource theories of $k$-unextendibility are relaxations alternative to the resource theory of negative partial transpose states from \cite{Rai99,Rai01}, in which the free states are the positive partial transpose (PPT) states and the free channels are the PPT-preserving channels. 

The main application of the resource theory of unextendibility reported here is that we obtain tighter upper bounds on the non-asymptotic quantum capacity of a quantum channel. We can get a sense of this by considering the following example: if we send one share of the maximally entangled state $\Phi_{AB} \coloneqq  \frac{1}{2} \sum_{i,j \in \{0,1\}} |i\rangle \! \langle j |_A \otimes |i\rangle \! \langle j |_B $ through a 50\% erasure channel with erasure symbol $\vert e\rangle \! \langle e \vert_B$, then the resulting state $\frac{1}{2}(\Phi_{AB} + I_A/2 \otimes \vert e\rangle \! \langle e \vert_B)$ is a two-extendible state, and is thus free in the resource theory of unextendibility for $k=2$. However, this state has distillable entanglement via two-way LOCC \cite{PhysRevLett.78.3217}, and so it is not free in the resource theory of entanglement. Thus, by relaxing the resource theory of entanglement, and as a consequence expanding the set of free states, we show in what follows how to obtain tighter, non-asymptotic upper bounds on the entanglement transmission rates of a quantum channel. 

The paper is organized as follows. In Section~\ref{sec:review}, we establish some notation and some definitions required for the proofs of our results. In Section~\ref{sec:resource-theory}, we introduce the resource theory of $k$-unextendibility. We also define quantifiers of unextendibility based on generalized divergences, and we establish their properties. In Section~\ref{sec:communications}, we obtain upper bounds on the non-asymptotic quantum capacity and one-way distillable entanglement. In Section~\ref{sec:example}, we showcase our bounds for depolarizing channels and erasure channels. In Section~\ref{sec:pretty-strong}, we revisit the pretty strong converse for the quantum capacity of antidegradable channels, and we employ the resource theory of unextendibility to obtain tighter bounds on their non-asymptotic quantum capacity. We finally conclude with some open questions in Section~\ref{sec:conclusion}.

\textit{Note on related work}: The relation of this paper to our previous one \cite{KDWW19} is that, in this paper, we go into far more detail on the resource theory and many of the proofs of the claims in \cite{KDWW19} are presented here.
There is also another paper \cite{RSTEDW18} that uses $k$-extendibility to place bounds on entanglement distillation protocols, but the kinds of protocols they consider and the particular way that they use $k$-extendibility are different from our approach in \cite{KDWW19} and in the present paper. Another paper \cite{BBFS18} employed $k$-extendibility in the context of placing bounds on the error in quantum communication protocols. They also introduced a  definition of $k$-extendible channels that is slightly different from that  given in \cite{KDWW19}.

\section{Preliminaries}

\label{sec:review}

\subsection{States, channels, isometries, and k-extendibility}

The Hilbert space of a quantum system $A$ is denoted by $\mc{H}_A$. The state of system $A$ is represented by a density operator $\rho_A$, which is a positive semi-definite operator with unit trace. The set of density operators is denoted by $\mc{D}(\mc{H}_A)$. The density operator of a composite system $RA$ is defined as $\rho_{RA}\in \mc{D}(\mc{H}_{RA})$, where $\mc{H}_{RA}=\mc{H}_R\otimes\mc{H}_A$. The notation $A^n\coloneqq  A_1A_2\cdots A_n$ indicates a composite system consisting of $n$ subsystems, each of which is isomorphic to Hilbert space $\mc{H}_A$. The fidelity of $\tau,\sigma\in\mc{D}(\mc{H}_A)$ is defined as $F(\tau,\sigma)=\norm{\sqrt{\tau}\sqrt{\sigma}}_1^2$ \cite{U76}, where $\norm{\cdot}_1$ denotes the trace norm.

A quantum channel is a completely positive trace preserving map (CPTP) map. Let $\mc{M}_{A\to B}$ be a quantum channel, and let $|\Gamma\rangle_{RA}$ denote the following maximally entangled vector:
\begin{equation}\label{eq:basis}
|\Gamma\rangle _{RA}\coloneqq \sum_{i}|i\rangle _R|i\rangle _A ,
\end{equation}
where $\dim(\mc{H}_R)=\dim(\mc{H}_A)$ and $\{|i\rangle _R\}_i$ and $\{|i\rangle _A\}_i$ are fixed orthonormal bases. We extend this notation to multiple parties with a given bipartite cut as
\begin{equation}
|\Gamma\rangle _{R_AR_B:AB}\coloneqq |\Gamma\rangle _{R_A:A}\otimes |\Gamma\rangle _{R_B:B}.
\end{equation}
The maximally entangled state $\Phi_{RA}$ is denoted as
\begin{equation}
\Phi_{RA}=\frac{1}{|A|}| \Gamma \rangle\!\langle \Gamma |_{RA},
\end{equation}
where $|A|=\dim(\mc{H}_A)$.
The Choi operator for a channel $\mc{M}_{A\to B}$ is defined as
\begin{equation}
\Gamma^\mc{M}_{RA}=(\id_R\otimes\mc{M}_{A\to B})\(|\Gamma\rangle \! \langle \Gamma|_{RA}\),
 \end{equation}
where $\id_R$ denotes the identity map on $R$. 
\iffalse
For $A'\simeq A$, the following identity holds
\begin{equation}\label{eq:choi-sim}
\langle \Gamma|_{A':R}\left[\rho_{SA'}\otimes \Gamma^\mc{M}_{RB}\right]|\Gamma\rangle _{A':R}=\mc{M}_{A\to B}(\rho_{SA}),
\end{equation}
where $A'\simeq A$. The above identity can be understood in terms of a post-selected variant \cite{HM04} of the quantum teleportation protocol \cite{BBC+93}. Another identity that holds is
\begin{equation}
\langle \Gamma|_{R:A} [Q_{SR}\otimes I_A]
|\Gamma\rangle _{R:A}=\Tr_R\{Q_{SR}\},
\end{equation}
for an operator $Q_{SR}\in \mc{B}(\mc{H}_S\otimes\mc{H}_R)$. 
\fi

Let $\SEP(A\!:\!B)$ denote the set of all separable states $\sigma_{AB}\in\mc{D}(\mc{H}_A\otimes\mc{H}_B)$, which are states that can be written as
\begin{equation}
\sigma_{AB}=\sum_{x}p(x)\omega^x_A\otimes\tau^x_B,
\end{equation}
where $p(x)$ is a probability distribution, $\omega^x_A \in \mc{D}(\mc{H}_A)$, and $\tau^x_B\in\mc{D}(\mc{H}_B)$ for all $x$. These are the free states in the resource theory of entanglement \cite{HHHH09,CG19}.

A local operations and classical communication (LOCC) channel $\mathcal{L}_{AB\rightarrow A^{\prime
}B^{\prime}}$ can be written as%
\begin{equation}
\mathcal{L}_{AB\rightarrow A^{\prime}B^{\prime}}=\sum_{y}\mathcal{E}%
_{A\rightarrow A^{\prime}}^{y}\otimes\mathcal{F}_{B\rightarrow B^{\prime}}%
^{y},\label{eq-sup:LOCC-channel}%
\end{equation}
where $\{\mathcal{E}_{A\rightarrow A^{\prime}}^{y}\}_{y}$ and $\{\mathcal{F}%
_{B\rightarrow B^{\prime}}^{y}\}_{y}$ are sets of completely positive
maps such that $\mathcal{L}_{AB\rightarrow A^{\prime}B^{\prime}}$ is trace
preserving. However, note that there exist separable channels that can be written in the form in \eqref{eq-sup:LOCC-channel} but are not realizable by LOCC \cite{BDFMRSSW99,CLM+14}. 

A special kind of LOCC channel is a one-way (1W-) LOCC channel from $A$
to $B$, in which Alice performs a quantum instrument, sends the classical
outcome to Bob, who then performs a quantum channel conditioned on the
classical outcome received from Alice. As such, any 1W-LOCC channel takes the
form in \eqref{eq-sup:LOCC-channel}, except that $\{\mathcal{E}_{A\rightarrow
A^{\prime}}^{y}\}_{y}$ is a set of CP\ maps such that the sum map $\sum
_{y}\mathcal{E}_{A\rightarrow A^{\prime}}^{y}$ is trace preserving, while
$\{\mathcal{F}_{B\rightarrow B^{\prime}}^{y}\}_{y}$ is a set of quantum channels.

%Consider a finite group $G$. For every $g\in G$, let $g\to U_A(g)$ and $g\to V_B(g)$ be projective unitary representations of $g$ acting on the input space $\mc{H}_A$ and the output space $\mc{H}_B$ of a quantum channel $\mc{N}_{A\to B}$, respectively. A quantum channel $\mc{N}_{A\to B}$ is covariant with respect to these representations if the following relation is satisfied for all $\rho_A\in\mc{D}(\mc{H}_A)$ and for all $g\in G$ \cite{Hol02,Hol07,H13book}:
%\begin{equation}
%\label{eq:cov-condition}
%\mc{N}_{A\to B}\!(U_A(g)\rho_A U_A^\dagger(g)) = V_B(g)\mc{N}_{A\to B}\(\rho_A\)V_B^\dagger(g) .
%\end{equation}

\subsection{Entropies and information}

The quantum entropy of a density operator $\rho_A$ is defined as \cite{Neu32}
\begin{equation}
S(A)_\rho\coloneqq  S(\rho_A)= -\Tr[\rho_A\log_2\rho_A].
\end{equation}

The quantum relative entropy of two quantum states is a measure of their distinguishability. For $\rho\in\mc{D}(\mc{H})$ and $\sigma\in\mc{B}_+(\mc{H})$, where $\mc{B}_{+}(\mc{H})$ is the set of positive semi-definite operators on $\mc{H}$, it is defined as~\cite{Ume62} 
\begin{equation}
D(\rho\V \sigma)\coloneqq  \left\{ 
\begin{tabular}{c c}
$\Tr\{\rho[\log_2\rho-\log_2\sigma]\}$, & $\supp(\rho)\subseteq\supp(\sigma)$\\
$+\infty$, &  otherwise.
\end{tabular} 
\right.
\end{equation}
The quantum relative entropy is non-increasing under the action of positive trace-preserving maps \cite{MR15}, that is $D(\rho\V\sigma)\geq D(\mc{M}(\rho)\V\mc{M}{(\sigma)})$ for any two density operators $\rho$ and $\sigma$ and a positive trace-preserving map $\mc{M}$.

\subsection{Generalized divergence and  relative entropies}

\label{sec:gen-div}

Let $\mathbf{D}$ be a function from $\mathcal{D}(\mathcal{H}) \times
\mc{B}_{+}(\mathcal{H})$ to $\mathbb{R}$.
Then $\mathbf{D}$ is called a generalized divergence \cite{PV10,SW12} if it satisfies the following data-processing inequality:
\begin{equation}\label{eq:gen-div-mono}
\mathbf{D}(\rho\Vert \sigma)\geq \mathbf{D}(\mathcal{N}(\rho)\Vert \mc{N}(\sigma)),
\end{equation}
 where $\rho \in \mathcal{D}(\mathcal{H})$ and $\sigma \in \mc{B}_{+}(\mathcal{H})$  and $\mc{N}$ is a quantum channel. 
Specific generalized divergences of relevance to this work are the sandwiched R\'enyi relative entropy \cite{MDSFT13,WWY14}, quantum relative entropy \cite{Ume62}, and $\varepsilon$-hypothesis testing relative entropy \cite{BD10,WR12}. 

The sandwiched R\'enyi relative entropy \cite{MDSFT13,WWY14} is denoted as $\wt{D}_\alpha(\rho\V\sigma)$  and defined for
$\rho\in\mc{D}(\mc{H})$, $\sigma\in\mc{B}_+(\mc{H})$ and  $ \alpha\in (0,1)\cup(1,\infty)$ as
\begin{equation}\label{eq:def_sre}
\wt{D}_\alpha(\rho\V \sigma)\coloneqq  \frac{1}{\alpha-1}\log_2 \Tr\left\{\left(\sigma^{\frac{1-\alpha}{2\alpha}}\rho\sigma^{\frac{1-\alpha}{2\alpha}}\right)^\alpha \right\}.
\end{equation}
It is set to $+\infty$ for $\alpha\in(1,\infty)$ if $\supp(\rho)\nsubseteq \supp(\sigma)$.
%Let $\widetilde{Q}_\alpha(\rho\V\sigma)\coloneqq  \Tr\left\{\left(\sigma^{\frac{1-\alpha}{2\alpha}}\rho\sigma^{\frac{1-\alpha}{2\alpha}}\right)^\alpha \right\}$ for short hand notation. %when $\supp(\rho)\subseteq\supp(\sigma)$.
The sandwiched R\'enyi relative entropy is monotone non-decreasing in $\alpha$ \cite{MDSFT13}:
\begin{equation}\label{eq:mono_sre}
\wt{D}_\alpha(\rho\V\sigma)\leq \wt{D}_\beta(\rho\V\sigma) ,
\end{equation}
 if   $\alpha\leq \beta$,  for  $\alpha,\beta\in(0,1)\cup(1,\infty)$.
For certain values of $\alpha$, the sandwiched R\'enyi relative entropy $\wt{D}_\alpha(\rho\V\sigma)$ is a particular kind of generalized divergence:
\begin{lemma}[\cite{FL13,Bei13}]
Let $\mc{N}:\mc{B}_+(\mc{H}_A)\to \mc{B}_+(\mc{H}_B)$ be a quantum channel   and let $\rho_A\in\mc{D}(\mc{H}_A)$ and $\sigma_A\in \mc{B}_+(\mc{H}_A)$. Then, for all $ \alpha\in \[1/2,1\)\cup (1,\infty)$,
\begin{equation}
\wt{D}_\alpha(\rho\V\sigma)\geq \wt{D}_\alpha(\mc{N}(\rho)\V\mc{N}(\sigma)) ,
\end{equation} 
\end{lemma}

In the limit $\alpha\to 1$, the sandwiched R\'enyi relative entropy $\wt{D}_\alpha(\rho\V\sigma)$ converges to the quantum relative entropy \cite{MDSFT13,WWY14}.
In the limit $\alpha\to \infty$, the sandwiched R\'enyi relative entropy $\wt{D}_\alpha(\rho\V\sigma)$ converges to the max-relative entropy \cite{MDSFT13}, which is defined as \cite{D09,Dat09}
\begin{equation}\label{eq:max-rel}
D_{\max}(\rho\V\sigma) \coloneqq \inf\{\lambda:\ \rho \leq 2^\lambda\sigma\},
\end{equation}
with $D_{\max}(\rho\V\sigma)=\infty$ if $\supp(\rho)\nsubseteq\supp(\sigma)$.
Another generalized divergence of interest is the $\varepsilon$-hypothesis-testing divergence \cite{BD10,WR12},  defined as
\begin{multline}
D^\varepsilon_h\!\(\rho\Vert\sigma\)\coloneqq \\-\log_2\inf_{\Lambda}\{\Tr\{\Lambda\sigma\}:  0\leq\Lambda\leq I \wedge\Tr\{\Lambda\rho\}\geq 1-\varepsilon\},\label{eq:hypo-test-div}
\end{multline}
for $\varepsilon\in[0,1]$, $\rho\in\mc{D}(\mc{H})$, and $\sigma\in\mc{B}_+(\mc{H})$.

\subsection{Channels with symmetry}

\label{sec:symmetry}
Consider a finite group $G$. For every $g\in G$, let $g\to U_A(g)$ and $g\to V_B(g)$ be projective unitary representations of $g$ acting on the input space $\mc{H}_A$ and the output space $\mc{H}_B$ of a quantum channel $\mc{N}_{A\to B}$, respectively. A quantum channel $\mc{N}_{A\to B}$ is covariant with respect to these representations if the following relation is satisfied \cite{Hol02,Hol07,H13book}:
\begin{equation}
\label{eq:cov-condition}
\mc{N}_{A\to B}\!(U_A(g)\rho_A U_A^\dagger(g)) = V_B(g)\mc{N}_{A\to B}\(\rho_A\)V_B^\dagger(g) .
\end{equation}

In our paper, we define covariant channels in the following way:
\begin{definition}[Covariant channel]\label{def:covariant}
A quantum channel is covariant if it is covariant with respect to a group $G$ for which each $g\in G$ has a unitary representation $U(g)$ acting on $\mc{H}_A$, such that $\{ U(g)\}_{g \in G}$ is a unitary one-design; i.e., the map  $(\cdot)\to \frac{1}{|G|}\sum_{g\in G}U(g)(\cdot)U^\dagger(g)$ always outputs the maximally mixed state for all input states. 
\end{definition}

The notion of teleportation simulation of a quantum channel first appeared in \cite{BDSW96}, and it was subsequently generalized in \cite[Eq.~(11)]{HHH99} to include general LOCC channels in the simulation. It was developed in more detail in \cite{Mul12} and used in the context of private communication in \cite{PLOB17} and \cite{WTB17,PhysRevLett.119.150501}.

\begin{definition}[Teleportation-simulable channel]\label{def:tel-sim}
A channel $\mc{N}_{A\to B}$ is teleportation-simulable if  there exists a resource state $\omega_{RB}\in\mc{D}\(\mc{H}_{RB}\)$ such that
for all $\rho_{A}\in\mc{D}\(\mc{H}_{A}\)$
\begin{equation}
\mc{N}_{A\to B}\(\rho_A\)=\mc{L}_{RA B\to B}\(\rho_{A}\otimes\omega_{RB}\),
\label{eq:TP-simul}
\end{equation}
where $\mc{L}_{RAB\to B}$ is an LOCC channel
(a particular example of an LOCC channel could be  a generalized teleportation protocol \cite{Wer01}). 
\end{definition}

\begin{lemma}[\cite{CDP09}]\label{thm:cov-tel-sim-channel}
All covariant channels (Definition~\ref{def:covariant}) are teleportation-simulable with respect to the resource state $\mathcal{N}_{A\to B}(\Phi_{RA})$.
\end{lemma}

\section{Framework for the resource theory of \texorpdfstring{$k$}{k}-unextendibility}\label{sec:resource-theory}

Any quantum resource theory consists of three ingredients~\cite{BG15,CG19}: the resourceful states, the free states, and the restricted set of free channels. The resource states by definition are those that are not free; they are useful and needed to carry out a given task. These states cannot be obtained by the action of the free channels on the free states. Also, free channels are incapable of increasing the amount of resourcefulness of a given state, whereas free states can be generated for free. 

\subsection{$k$-extendible states}
To develop a framework for the quantum resource theory of $k$-unextendibility, specified with respect to a fixed subsystem ($B$) of a bipartite system ($AB$), let us first recall the definition of a $k$-extendible state \cite{W89a,DPS02,DPS04}:
\begin{definition}[$k$-extendible state]\label{def:kex-state}
For integer $k\geq 2$, a state $\rho_{AB}\in\mc{D}(\mc{H}_{AB})$ is  $k$-extendible if there exists a state $\sigma_{AB^k}\coloneqq \sigma_{AB_1B_2\cdots B_k}\in\mc{D}(\mc{H}_{AB_1B_2\cdots B_k})$ that satisfies the following two criteria:
\begin{enumerate}
\item 
The state $\sigma_{AB_1B_2\cdots B_k}$ is permutation invariant with respect to the $B$ systems, in the sense that
for all $\pi\in S_{k}$,
\begin{equation}
\sigma_{AB_1B_2\cdots B_k} = 
\mathcal{W}_{B_{1}\cdots B_{k}}^{\pi}(\sigma_{AB_1B_2\cdots B_k}),
\end{equation}
where $\mc{W}^\pi$ is the unitary permutation channel associated with $\pi$ and $S_k$ is the symmetric group defined over a finite set of $k$ symbols. 
\item 
The state $\rho_{AB}$ is the marginal of
$\sigma_{AB_1 \cdots B_k}$, i.e.,
\begin{equation}
\rho_{AB} = \Tr_{B_2 \cdots  B_k}\{\sigma_{AB_1\ldots B_k}\}.
\end{equation}
\end{enumerate}
\end{definition}

Determining whether a bipartite state is separable or not is a computationally hard task \cite{G03,Gharibian10}. The $k$-extendible states, introduced in \cite{W89a,DPS04}, provide a systematic way of testing the entanglement of a state. If a state is entangled, it is not $k$-extendible for at least some $k$; furthermore, it is not $k'$-extendible for all $k' \geq k$. However, if the state is separable, then it is  $k$-extendible for all $k$. Then the question regarding the separability of the state can be reformulated as the verification of $k$-extendibility of a state, which is a semidefinite program (SDP). The size of the SDP increases with increase in $k$, because the number of constraints that need to be specified increases. Nevertheless, checking for $k$-extendibility of a state provides a hierarchy of SDPs in the sense discussed above, which can be insightful in understanding the entanglement of a bipartite state. 

To give some physical context to the definition of a $k$-extendible state, suppose that Alice and Bob share a bipartite state and that Bob subsequently mixes his system and the vacuum state at a 50:50 beamsplitter. Then the resulting state of Alice's system and one of the outputs of the beamsplitter is a two-extendible state by construction. As a generalization of this, suppose that Bob sends his system through the $N$-splitter of \cite[Eq.~(10)]{vLB00}, with the other input ports set to the vacuum state. Then the state of Alice's system and one of the outputs of the $N$-splitter is $N$-extendible by construction. One could also physically realize $k$-extendible states in a similar way by means of quantum cloning machines \cite{RevModPhys.77.1225}.

Although the following definition might be obvious, we nevertheless state it explicitly for clarity:

\begin{definition}[Unextendible state]
A state that is not $k$-extendible according to Definition~\ref{def:kex-state} is  called $k$-unextendible. 
\end{definition}

%We say that a state is $k$-unextendible if it is not $k$-extendible.
For simplicity and throughout this work, if we mention  ``extendibility,'' ``extendible,'' ``unextendibility,'' or ``extendible,'' then these terms should be understood as $k$-extendibility, $k$-extendible, $k$-unextendibility, or $k$-unextendible, respectively, with an implicit dependence on $k$.

Let $\kex(A\!:\!B)$ denote the set of all states $\sigma_{AB}\in\mc{D}(\mc{H}_{AB})$ that are $k$-extendible with respect to system~$B$. A $k$-extendible state is also $\ell$-extendible, where $\ell\leq k$. This follows trivially from the definition.

\subsection{$k$-extendible channels}

\label{sec:ext-channel}

In order to define $k$-extendible channels, we need to generalize the notions of permutation invariance and marginals of quantum states to quantum channels. First, permutation invariance of a state gets generalized to permutation covariance of a channel. Next, the marginal of a state gets generalized to the marginal of a channel, which includes a no-signaling constraint, in the following sense:

\begin{definition}[$k$-extendible channel]\label{def:ext-channel}
A bipartite channel $\mathcal{N}_{AB\rightarrow
A^{\prime}B^{\prime}}$ is  $k$-extendible if there exists a quantum channel
$\mathcal{M}_{AB_{1}\cdots B_{k}\rightarrow A^{\prime}B_{1}^{\prime}\cdots
B_{k}^{\prime}}$ that satisfies the following two criteria:
\begin{enumerate}
\item The channel
$\mathcal{M}_{AB_{1}\cdots B_{k}\rightarrow A^{\prime}B_{1}^{\prime}\cdots
B_{k}^{\prime}}$ 
 is permutation covariant with respect to the $B$
systems. That is, for all $\pi\in S_{k}$ and for all states $\rho_{AB_{1}\cdots B_{k}}
$, the following equality holds%
\begin{multline}
  \mathcal{M}_{AB_{1}\cdots B_{k}\rightarrow A^{\prime}B_{1}^{\prime}\cdots
B_{k}^{\prime}}(\mathcal{W}_{B_{1}\cdots B_{k}}^{\pi}(\rho_{AB_{1}\cdots
B_{k}}))\\=\mathcal{W}_{B_{1}^{\prime}\cdots B_{k}^{\prime}}^{\pi}%
(\mathcal{M}_{AB_{1}\cdots B_{k}\rightarrow A^{\prime}B_{1}^{\prime}\cdots
B_{k}^{\prime}}(\rho_{AB_{1}\cdots B_{k}})),\label{eq:perm-cov-channel}%  
\end{multline}
where $\mc{W}^\pi$ is the unitary permutation channel associated with the permutation~$\pi$.

\item The channel $\mathcal{N}_{AB\rightarrow A^{\prime}B^{\prime}}$ is 
the marginal of $\mathcal{M}_{AB_{1}\cdots B_{k}\rightarrow A^{\prime}%
B_{1}^{\prime}\cdots B_{k}^{\prime}}$ in the following sense: for every state $\rho_{AB_{1}\cdots B_{k}}$,
\begin{multline}
  \mathcal{N}_{AB\rightarrow A^{\prime
}B^{\prime}}(\rho_{AB_1})\\=\operatorname{Tr}_{B_{2}^{\prime}\cdots B_{k}^{\prime
}}\{\mathcal{M}_{AB_{1}\cdots B_{k}\rightarrow A^{\prime}B_{1}^{\prime}\cdots
B_{k}^{\prime}}(\rho_{AB_{1}\cdots B_{k}})\}.\label{eq:marginal-channel}%  
\end{multline}
We can alternatively write \eqref{eq:marginal-channel} as
\begin{multline}
    \operatorname{Tr}_{B_{2}^{\prime}\cdots B_{k}^{\prime
}} \circ \mathcal{M}_{AB_{1}\cdots B_{k}\rightarrow A^{\prime}B_{1}^{\prime}\cdots
B_{k}^{\prime}} \\
= \mathcal{N}_{AB\rightarrow A^{\prime
}B^{\prime}} \circ \operatorname{Tr}_{B_{2}\cdots B_{k}}
\end{multline}

\end{enumerate}
A channel $\mathcal{M}_{AB_{1}\cdots B_{k}\rightarrow A^{\prime}%
B_{1}^{\prime}\cdots B_{k}^{\prime}}$ satisfying the above conditions is called a $k$-extension of $\mathcal{N}%
_{AB\rightarrow A^{\prime}B^{\prime}}$. 

\end{definition}

 \begin{figure}
 \begin{centering}
 \includegraphics[width=1\columnwidth]{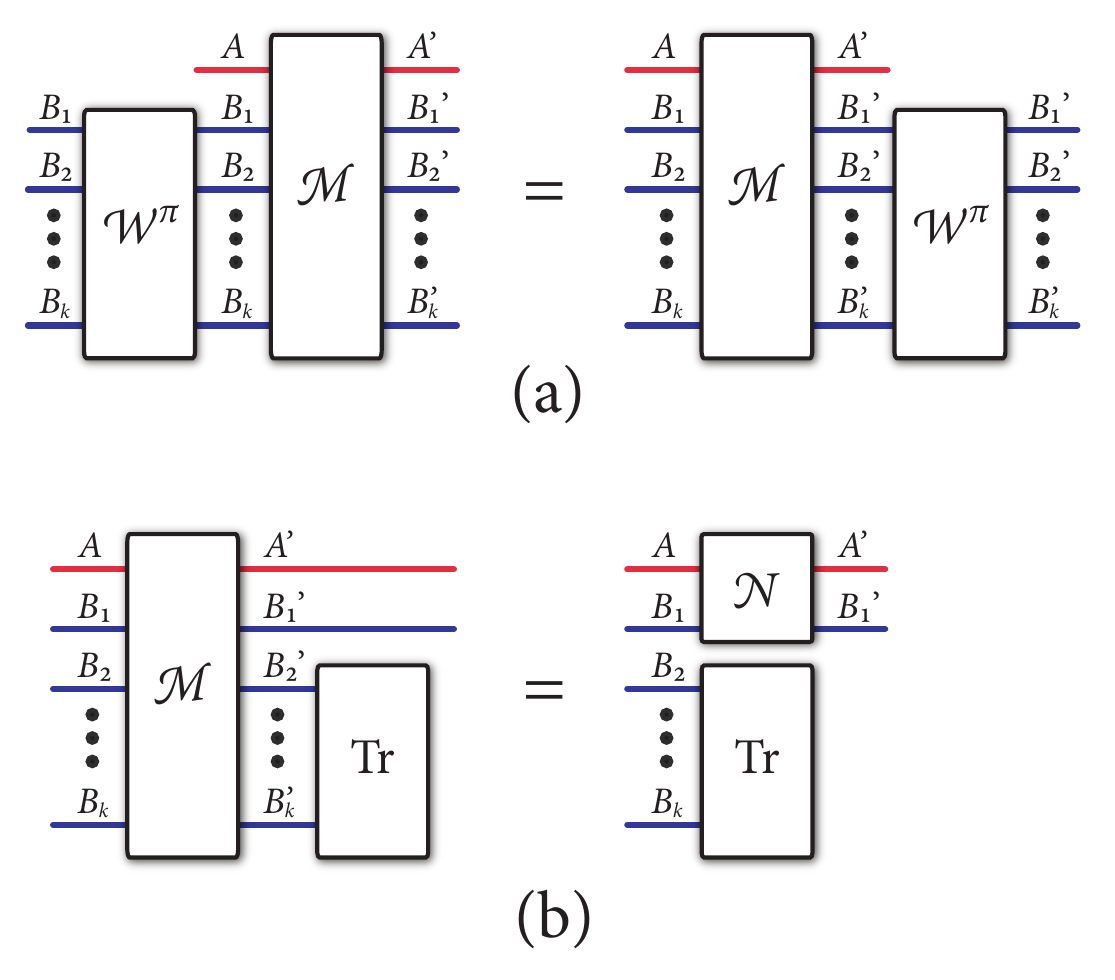}
 \caption{A visual depiction of the conditions for the channel $\mathcal{M}_{AB_{1}\cdots B_{k}\rightarrow A^{\prime}B_{1}^{\prime}\cdots
B_{k}^{\prime}}$ to be a $k$-extension of $\mathcal{N}_{AB\rightarrow A^{\prime
}B^{\prime}}$. (a) The extension channel $\mathcal{M}_{AB_{1}\cdots B_{k}\rightarrow A^{\prime}B_{1}^{\prime}\cdots
B_{k}^{\prime}}$ should be permutation covariant with respect to Bob's systems. (b) The extension channel $\mathcal{M}_{AB_{1}\cdots B_{k}\rightarrow A^{\prime}B_{1}^{\prime}\cdots
B_{k}^{\prime}}$ should reduce to the original channel $\mathcal{N}_{AB\rightarrow A^{\prime
}B^{\prime}}$ when tracing out the output systems $B'_2 \cdots B'_k $ of $\mathcal{M}_{AB_{1}\cdots B_{k}\rightarrow A^{\prime}B_{1}^{\prime}\cdots
B_{k}^{\prime}}$.}
 \label{fig:k-ext-conditions}
 \end{centering}
 \end{figure}

The conditions in Definition~\ref{def:ext-channel} are depicted in Figure~\ref{fig:k-ext-conditions}.
The condition in \eqref{eq:marginal-channel} %invokes
corresponds to a one-way no-signaling (semi-causal) constraint on the extended $(k-1)$ subsystems $B^{k-1}\coloneqq  B^k\setminus B_i$ to $A'B'_i$ for all $i\in[k]$ (cf., \cite[Proposition~7]{DW16}). This condition can be reformulated as \cite{PHH2006}
\begin{multline}
    \operatorname{Tr}_{B_{2}^{\prime}\cdots B_{k}^{\prime
}}\{\mathcal{M}_{AB_{1}\cdots B_{k}\rightarrow A^{\prime}B_{1}^{\prime}\cdots
B_{k}^{\prime}}(\rho_{AB_1\cdots B_K})\}=\\ \operatorname{Tr}_{B_{2}^{\prime}\cdots B_{k}^{\prime
}}\{\mathcal{M}_{AB_{1}\cdots B_{k}\rightarrow A^{\prime}B_{1}^{\prime}\cdots
B_{k}^{\prime}}(\mathcal{R}^{\pi}_{B_2\cdots B_k}(\rho_{AB_1\cdots B_K}))\},
\end{multline}
where $\mathcal{R}^{\pi}_{B_2\cdots B_k}$ is a channel that replaces the state in systems $B_2\cdots B_k$ with a mixed state $\pi_{B_2\cdots B_k}$ (or any other arbitrary state).  
Equivalently, the condition
in \eqref{eq:marginal-channel}
can also be expressed as
\begin{align}
\operatorname{Tr}_{B_{2}^{\prime}\cdots B_{k}^{\prime
}}\{\mathcal{M}_{AB_{1}\cdots B_{k}\rightarrow A^{\prime}B_{1}^{\prime}\cdots
B_{k}^{\prime}}(X_{AB_1}\otimes Y_{B_2\cdots B_k})\} =0 
\end{align}
for all $X_{AB_1}, Y_{B_2\cdots B_k}$ such that $\operatorname{Tr}\{Y_{B_2\cdots B_k}\}=0$ \cite{DW16}.

Classical $k$-extendible channels were defined in a somewhat similar way in
\cite{BH17}, and so our definition above represents a quantum generalization
of the classical notion. We also note here that $k$-extendible channels were
defined in a different way in \cite{PBHS13}, but our definitions
reduce to the same class of channels in the case that the input systems
$B_{1}$ through $B_{k}$ and the output systems $A^{\prime}$ are trivial.

%Let $\tf{F}$ be the set of all free states (in all possible finite dimensions), and let $\tf{F}_{d}$ be the set of free states in $\mc{D}(\mc{H}^{d})$ for $d\coloneqq \dim(\mc{H}_{AB})$.
%The basic assumptions for the framework of our unextendible resource theory are the same as the framework of quantum resource theory stated as the five postulates in . Consistent with these assumptions, we have all $\rho_{AB}\in\kex(A\!:\!B)$ and all 

%In this section, we discuss the structure of a class of extendibility-preserving channels called $k$-extendible channels as stated in Definition~\ref{def:ext-channel}. 
%
%The $k$-extendible channel represents a fully quantum generalization of
%the $k$-extendible non-signaling classical channels from \cite{BH17}.

% The condition in \eqref{eq:perm-cov-channel}\ is equivalent to the following
% constraint on the Choi matrix of $\mathcal{N}_{AB_{1}\cdots B_{k}\rightarrow
% A^{\prime}B_{1}^{\prime}\cdots B_{k}^{\prime}}$:%
% \begin{equation}
% \forall\pi\in S_{k}:\left[  W_{\hat{B}_{1}\cdots\hat{B}_{k}}^{\pi T}\otimes
% W_{B_{1}^{\prime}\cdots B_{k}^{\prime}}^{\pi},\Gamma_{\hat{A}A^{\prime}\hat
% {B}_{1}\cdots\hat{B}_{k}B_{1}^{\prime}\cdots B_{k}^{\prime}}^{\mathcal{N}%
% }\right]  =0.
% \end{equation}

% By picking a matrix basis and exploiting linearity, the infinite number of
% conditions in \eqref{eq:marginal-channel} reduces to a finite number. 

We can reformulate the constraints on the $k$-extendible channels in terms of the Choi operator $\Gamma^{\mc{M}}_{\hat{A}A'\hat{B}^kB'^k}$ of the extension channel $\mathcal{M}_{AB_{1}\cdots B_{k}\rightarrow A^{\prime}B_{1}^{\prime}\cdots
B_{k}^{\prime}}$ of $\mc{N}_{AB\to A'B'}$ as follows:
\begin{align}
&\Gamma^{\mc{M}}_{\hat{A}A'\hat{B}^kB'^k} \geq 0,\\
&\Tr_{A'B'^k}\{\Gamma^{\mc{M}}_{\hat{A}A'\hat{B}^kB'^k}\} =I_{\hat{A}\hat{B}^k},\\
&\left[  W_{\hat{B}_{1}\cdots\hat{B}_{k}}^{\pi }\otimes
W_{B_{1}^{\prime}\cdots B_{k}^{\prime}}^{\pi},\Gamma_{\hat{A}A^{\prime}\hat
{B}_{1}\cdots\hat{B}_{k}B_{1}^{\prime}\cdots B_{k}^{\prime}}^{\mathcal{M}%
}\right]  =0,\ \forall\pi\in S_{k}\\
&\Gamma^{\mc{M}}_{\hat{A}A'\hat{B}_1B'_1\hat{B}_2\cdots \hat{B}_k} = \Gamma^{\mc{N}}_{\hat{A}A'\hat{B}_1\hat{B}'_1}\otimes \pi_{\hat{B}_2\cdots \hat{B}_k}
\end{align}  
The first constraint corresponds to complete positivity of the $k$-extendible channel, while the second constraint corresponds to trace preservation of the channel. The third constraint reflects the permutation covariance property of the channel with respect to the permutation group, and the last constraint corresponds to the no-signaling condition.

%We can also use the Choi state to understand the condition in
%\eqref{eq:marginal-channel}%
%\begin{multline}
%\langle\Gamma|_{\bar{A}\hat{A}\bar{B}\hat{B}}\rho_{\bar{A}\bar{B}}%
%\otimes\Gamma_{\hat{A}A^{\prime}\hat{B}B^{\prime}}^{\mathcal{N}}|\Gamma
%\rangle_{\bar{A}\hat{A}\bar{B}\hat{B}}=\\
%\langle\Gamma|_{\bar{A}\hat{A}\bar
%{B}_{1}\cdots\bar{B}_{k}\hat{B}_{1}\cdots\hat{B}_{k}}\rho_{\bar{A}\bar{B}%
%_{1}\cdots\bar{B}_{k}}\otimes\Gamma_{\hat{A}A^{\prime}\hat{B}_{1}\cdots\hat
%{B}_{k}B_{1}^{\prime}\cdots B_{k}^{\prime}}^{\mathcal{N}}|\Gamma\rangle
%_{\bar{A}\hat{A}\bar{B}_{1}\cdots\bar{B}_{k}\hat{B}_{1}\cdots\hat{B}_{k}}.
%\end{multline}

The following theorem is the key statement that makes the resource theory of unextendibility, as presented above, a consistent resource theory:

\begin{theorem}\label{thm:main-resource-theory-k-ext}
For a bipartite $k$-extendible channel $\mathcal{N}_{AB\rightarrow A^{\prime}B^{\prime}%
}$ and a $k$-extendible state $\rho_{AB}$, the output state $\mathcal{N}%
_{AB\rightarrow A^{\prime}B^{\prime}}(\rho_{AB})$ is $k$-extendible. 
\end{theorem}

\begin{proof}
Let $\rho_{AB_{1}\cdots B_{k}}$ be a $k$-extension of $\rho_{AB}$. Let
$\mathcal{M}_{AB_{1}\cdots B_{k}\rightarrow A^{\prime}B_{1}^{\prime}\cdots
B_{k}^{\prime}}$ be a channel that extends $\mathcal{N}_{AB\rightarrow
A^{\prime}B^{\prime}}$. Then the following state is a $k$-extension of
$\mathcal{N}_{AB\rightarrow A^{\prime}B^{\prime}}(\rho_{AB})$:%
\begin{equation}
\mathcal{M}_{AB_{1}\cdots B_{k}\rightarrow A^{\prime}B_{1}^{\prime}\cdots
B_{k}^{\prime}}(\rho_{AB_{1}\cdots B_{k}}).
\end{equation}
To verify this statement, consider that for all $\pi\in S_{k}$, the following
holds by applying \eqref{eq:perm-cov-channel}\ and the fact that $\rho
_{AB_{1}\cdots B_{k}}$ is a $k$-extension of $\rho_{AB}$:%
\begin{align}
& \mathcal{W}_{B_{1}^{\prime}\cdots B_{k}^{\prime}}^{\pi}(\mathcal{M}%
_{AB_{1}\cdots B_{k}\rightarrow A^{\prime}B_{1}^{\prime}\cdots B_{k}^{\prime}%
}(\rho_{AB_{1}\cdots B_{k}}))\nonumber\\
& =\mathcal{M}_{AB_{1}\cdots B_{k}\rightarrow A^{\prime}B_{1}^{\prime}\cdots
B_{k}^{\prime}}(\mathcal{W}_{B_{1}\cdots B_{k}}^{\pi}(\rho_{AB_{1}\cdots
B_{k}}))\\
& =\mathcal{M}_{AB_{1}\cdots B_{k}\rightarrow A^{\prime}B_{1}^{\prime}\cdots
B_{k}^{\prime}}(\rho_{AB_{1}\cdots B_{k}}).
\end{align}
Due to \eqref{eq:marginal-channel}, it follows that
$\mathcal{N}_{AB\rightarrow A^{\prime}B^{\prime}}(\rho_{AB})$ is a marginal of $\mathcal{M}_{AB_{1}\cdots B_{k}\rightarrow A^{\prime}B_{1}^{\prime}\cdots
B_{k}^{\prime}}(\rho_{AB_{1}\cdots B_{k}})$.
\end{proof}

\bigskip
With the above framework in place, we note here that postulates I--V of
\cite{BG15} apply to the resource theory of unextendibility.
The $k$-extendible channels are the free channels, and the $k$-extendible states are the free states.
%preserving operations as the free states and the free operations, respectively, for the unextendible resource theory. 

\begin{example}[1W-LOCC]
An example of a $k$-extendible channel is a one-way local operations and classical communication ($1$W-LOCC) channel.  Consider that a $1$W-LOCC channel $\mathcal{N}_{AB\rightarrow A^{\prime}B^{\prime}}$ can be written as%
\begin{equation}
\mathcal{N}_{AB\rightarrow A^{\prime}B^{\prime}}=\sum_{x}\mathcal{E}%
_{A\rightarrow A^{\prime}}^{x}\otimes\mathcal{F}_{B\rightarrow B^{\prime}}%
^{x},
\end{equation}
where $\{\mathcal{E}_{A\rightarrow A^{\prime}}^{x}\}_{x}$ is a collection of
completely positive maps such that $\sum_{x}\mathcal{E}_{A\rightarrow
A^{\prime}}^{x}$ is a quantum channel and $\{\mathcal{F}_{B\rightarrow
B^{\prime}}^{x}\}_{x}$ is a collection of quantum channels. A $k$-extension
$\mathcal{M}_{AB_{1}\cdots B_{k}\rightarrow A^{\prime}B_{1}^{\prime}\cdots
B_{k}^{\prime}}$\ of the channel $\mathcal{N}_{AB\rightarrow A^{\prime
}B^{\prime}}$ can be taken as follows:%
\begin{multline}
 \mathcal{M}_{AB_{1}\cdots B_{k}\rightarrow A^{\prime}B_{1}^{\prime}\cdots
B_{k}^{\prime}}=\\\sum_{x}\mathcal{E}_{A\rightarrow A^{\prime}}^{x}%
\otimes\mathcal{F}_{B_{1}\rightarrow B_{1}^{\prime}}^{x}\otimes\mathcal{F}%
_{B_{2}\rightarrow B_{2}^{\prime}}^{x}\otimes\cdots\otimes\mathcal{F}%
_{B_{k}\rightarrow B_{k}^{\prime}}^{x}.
\end{multline}

It is then clear that the condition in \eqref{eq:perm-cov-channel} holds for
$\mathcal{M}_{AB_{1}\cdots B_{k}\rightarrow A^{\prime}B_{1}^{\prime}\cdots
B_{k}^{\prime}}$ as chosen above. Furthermore, the condition in
\eqref{eq:marginal-channel} holds because each $\mathcal{F}_{B_{i}\rightarrow
B_{i}^{\prime}}^{x}$ is a channel\ for $i\in\{1,\ldots,k\}$.
\end{example}

We now define a subclass of $k$-extendible channels. These channels are realized as follows: Alice performs a quantum channel $\mc{E}_{A\to A'C}$ on her system $A$ and obtains systems $A'C$. Then, Alice sends $C$ to Bob over a $k$-extendible channel $\mc{A}^k_{C\to C'}$.  The channel $\mc{A}^k_{C\to C'}$ is a special case of the bipartite $k$-extendible channel $\mc{N}_{AB\to A'B'}$ considered in Definition~\ref{def:ext-channel}, in which we identify the input $C$ with $A$ of $\mc{N}_{AB\to A'B'}$, the output $C'$ with $B'$ of $\mc{N}_{AB\to A'B'}$ and the systems $B$ and $A'$ are trivial. Finally, Bob applies the channel $\mc{D}_{C'B\to B'}$ on system $C'$ and his local system $B$ to get $B'$. Denoting the overall channel by $\mathcal{K}_{AB\rightarrow A'B'}^k$, it is realized as follows:
\begin{equation}
\label{eq:subclass_def}
\mathcal{K}_{AB\rightarrow A'B'}^k(\cdot)\coloneqq   \mc{D}_{C'B\to B'}  \circ\mc{A}^k_{C\to C'} \circ\mc{E}_{A\to A'C}(\cdot).
\end{equation}

Due to their structure, we can place an upper bound on the distinguishability of a channel in the subclass described above and the set of 1W-LOCC channels, as quantified by the diamond norm \cite{Kit97}. See  Appendix~\ref{sec:subclass-channels} for the precise statement and for details of the proof.

\subsection{Quantifying \texorpdfstring{$k$}{k}-unextendibility}

\label{sec:unext-gen-div}

In any resource theory, it is pertinent to quantify the resourcefulness of the resource states and the resourceful channels. Based on the resource theory of unextendibility, any measure of the $k$-unextendibility of a state should possess the following two desirable properties:
\begin{enumerate}
\item data processing: non-increasing under the action of $k$-extendible channels,
\item attains minimum value if the state is $k$-extendible.
\end{enumerate}
%Without loss of generality, we can define measures to be such that the minimum value they attain is zero. 

Here we present a measure of unextendibility that is based on generalized divergence and  satisfies both criteria discussed above:

\begin{definition}[Unextendible generalized divergence]\label{def:unext-gen-div-state}
The $k$-unextendible generalized divergence of a bipartite state $\rho_{AB}$ is defined as
\begin{equation}
\tf{E}_{k}(A;B)_{\rho}=\inf_{\sigma_{AB}\in\kex(A:B)}\tf{D}(\rho_{AB}\Vert\sigma_{AB}),
\end{equation}
where $\tf{D}(\rho\Vert \sigma)$ denotes the generalized divergence from~\eqref{eq:gen-div-mono}.
\end{definition}

We can extend the definition above to obtain an unextendible generalized divergence of a channel, in order to quantify how well a quantum channel can preserve unextendibility. 

\begin{definition}
\label{def:unext-gen-div-channel}
The $k$-unextendible generalized divergence of a quantum channel $\mc{N}_{A\to B}$ is defined as
\begin{multline}
\label{eq:unex-channel-def}
\tf{E}_{k}(\mc{N})\coloneqq \\ \sup_{\psi_{RA}\in \mathcal{D}(\mathcal{H}_{RA})}\inf_{\sigma_{RB}\in\kex(R:B)}\tf{D}\(\mc{N}_{A\to B}(\psi_{RA})\Vert\sigma_{RB}\),
\end{multline}
where $\tf{D}(\cdot\Vert \cdot)$ is a generalized divergence and the optimization is over all pure states $\psi_{RA}\in\mc{D}(\mc{H}_{RA})$ with $\dim(\mc{H}_R) = \dim(\mc{H}_R)$.
\end{definition}

In the  definition above, we could have taken an optimization over all mixed-state inputs with the reference system $R$ arbitrarily large. However, due to purification, data processing, and the Schmidt decomposition theorem, doing so does not result in a larger value of the quantity, so that it suffices to restrict the optimization as we have done above.

In Definitions~\ref{def:unext-gen-div-state} and \ref{def:unext-gen-div-channel}, we can take the generalized divergence to be the quantum relative entropy~$D$, the $\varepsilon$-hypothesis-testing divergence $D^\varepsilon_h$,  the $\alpha$-sandwiched-R\'{e}nyi divergence $\wt{D}_\alpha$, the traditional R\'{e}nyi divergence, the trace distance, etc., in order to have various $k$-unextendible measures of states and channels (see Section~\ref{sec:gen-div} for definitions). 

\iffalse
\begin{definition}[$\alpha$-sandwiched-R\'{e}nyi information of $k$-extendibility of a quantum channel]
The $\alpha$-sandwiched-R\'{e}nyi information of $k$-extendibility of a quantum channel $\mc{N}_{A\to B}$ is defined as
\begin{equation}\label{def:sandwich-inf}
\wt{E}^\alpha_{k}(\mc{N})=\sup_{\rho_{RA}}\inf_{\sigma_{AB}\in\kex(A:B)}\wt{D}_\alpha\(\mc{N}_{A\to B}\rho_{RA}\Vert\sigma_{RB}\).
\end{equation}
\end{definition}
\fi 

\subsubsection{\texorpdfstring{$k$}{k}-unextendible divergences for isotropic and  Werner states}

\label{sec:isotropic-werner-states}

In this section, we evaluate some unextendible divergences for two specific classes of states: isotropic and Werner states. In particular, we obtain an analytic form for the $k$-unextendible generalized divergence (Proposition~\ref{prop:generalized-divergence-isotropic}) for isotropic states \cite{HH99} and Werner states \cite{Wer89}, and, as a consequence, we calculate its $k$-unextendible relative entropy and R\'{e}nyi divergence (Proposition~\ref{prop:relative-entropy-isotropic}).

\begin{definition}[Isotropic state
\cite{HH99}]
An isotropic state $\rho^{(t,d)}_{AB}$
is $U\otimes {U}^\ast$-invariant for an arbitrary unitary $U$, where $\dim(\mc{H}_A)=d=\dim(\mc{H}_B)$. Such a state can be written in the following form for $t\in [0,1]$:
\begin{equation}\label{eq:iso-state1}
\rho^{(t,d)}_{AB}=t\Phi^{d}_{AB}+(1-t)\frac{I_{AB}-\Phi^{d}_{AB}}{d^2-1},
\end{equation}
where $\Phi^{d}_{AB}$ denotes a maximally entangled state of Schmidt rank $d$.
\end{definition}

\begin{lemma}[\cite{JV13}]\label{thm:isotropicstates}
An isotropic state $\rho^{(t,d)}_{AB}$ written as in \eqref{eq:iso-state1} is $k$-extendible if and only if $t\in\left[0, \frac{1}{d}\left(1+\frac{d-1}{k}\right)\right]$.
\end{lemma}

\begin{proof}
Isotropic states are parametrized in \cite{JV13} for $y\in\left[  0,d\right]  $
as%
\begin{equation}
\frac{d}{d^{2}-1}\left[  \left(  d-y\right)  \frac{I_{AB}}{d^{2}}+\left(
y-\frac{1}{d}\right)  \Phi^d_{AB}\right]  .
\end{equation}
There, as shown in \cite[Theorem~III.8]{JV13}, an isotropic state is $k$-extendible if and only if
\begin{equation}
y\leq1+\left(  d-1\right)  /k.\label{eq:isotropic-k-ext}%
\end{equation}
Translating this to the parametrization in \eqref{eq:iso-state1}, we find
that%
\begin{align}
&  \frac{d}{d^{2}-1}\left[  \left(  d-y\right)  \frac{I_{AB}}{d^{2}}+\left(
y-\frac{1}{d}\right)  \Phi^d_{AB}\right]  \notag \\
&  =\frac{d}{d^{2}-1}\left[  \frac{d-y}{d^{2}}\left(  I_{AB}-\Phi^d_{AB}\right)
+\left(  \frac{d-y}{d^{2}}+y-\frac{1}{d}\right)  \Phi^d_{AB}\right]  \\
&  =\frac{d-y}{d}\frac{I_{AB}-\Phi^d_{AB}}{d^{2}-1}+\frac{y}{d}\Phi^d_{AB}.
\end{align}
Using the fact that $t=y/d$ to translate between the two different
parametrizations of isotropic states, the condition in
\eqref{eq:isotropic-k-ext} translates to%
\begin{equation}
t\leq\frac{1}{d}\left(  \frac{d-1}{k}+1\right)  .
\end{equation}
This concludes the proof.
\end{proof}

\begin{definition}
[Werner state \cite{Wer89}]Let $A$ and $B$ be quantum systems, each of dimension~$d$. A
Werner state is defined for $p\in\lbrack0,1]$ as%
\begin{equation}
W_{AB}^{(p,d)}\coloneqq \left(  1-p\right)  \frac{2}{d\left(  d+1\right)  }\Pi
_{AB}^{+}+p\frac{2}{d\left(  d-1\right)  }\Pi_{AB}^{-},
\label{eq:werner-param}
\end{equation}
where $\Pi_{AB}^{\pm}\coloneqq \left(  I_{AB}\pm F_{AB}\right)  /2$ are the
projections onto the symmetric and antisymmetric subspaces of $A$ and $B$.
\end{definition}

\begin{lemma}[\cite{JV13}]\label{lem:werner-k-ext}
A Werner state $W_{AB}^{(p,d)}$ is $k$-extendible if and only if
$p\in\left[  0,\frac{1}{2}\left(  \frac{d-1}{k}+1\right)  \right]  $.
\end{lemma}

\begin{proof}
Werner states are parametrized in \cite{JV13} for $q\in\left[  -1,1\right]  $ as%
\begin{equation}
\frac{d}{d^{2}-1}\left[  \left(  d-q\right)  \frac{I_{AB}}{d^{2}}+\left(
q-\frac{1}{d}\right)  \frac{F_{AB}}{d}\right]  .
\end{equation}
There, as shown in \cite[Theorem~III.7]{JV13}, a Werner state is $k$-extendible if and only if
\begin{equation}
q\geq-\left(  d-1\right)  /k.\label{eq:Werner-k-ext}%
\end{equation}
Translating this to the parametrization in \eqref{eq:werner-param}, and using that%
\begin{align}
I_{AB} &  =\Pi_{AB}^{+}+\Pi_{AB}^{-},\\
F_{AB} &  =\Pi_{AB}^{+}-\Pi_{AB}^{-},
\end{align}
we find that%
\begin{align}
{}&  \frac{d}{d^{2}-1}\left[  \left(  d-q\right)  \frac{I_{AB}}{d^{2}}+\left(
q-\frac{1}{d}\right)  \frac{F_{AB}}{d}\right]  \notag \\
\begin{split}
{}&  =\frac{d}{d^{2}-1}\left[  \frac{d-q}{d^{2}}\left(  \Pi_{AB}^{+}+\Pi
_{AB}^{-}\right)\right.  \\{}&\qquad+ \left.\left(  \frac{q}{d}-\frac{1}{d^{2}}\right)  \left(
\Pi_{AB}^{+}-\Pi_{AB}^{-}\right)  \right]\end{split}  \\
\begin{split}
{}&  =\frac{d}{d^{2}-1}\left[ \left(  \frac{d-q}{d^{2}}+
 \frac{q}{d}-\frac
{1}{d^{2}}\right)  \Pi_{AB}^{+}\right. \\&\qquad \left.+\left(  \frac{d-q}{d^{2}}-\left(  \frac{q}%
{d}-\frac{1}{d^{2}}\right)  \right)  \Pi_{AB}^{-}\right] 
\end{split}  
\\
{}&  =\frac{1+q}{2}\frac{2}{d\left(  d+1\right)  }\Pi_{AB}^{+}+\frac{1-q}%
{2}\frac{2}{d\left(  d-1\right)  }\Pi_{AB}^{-}.
\end{align}
Using the fact that $p=\left(  1-q\right)  /2$ to translate between the two
different parametrizations of Werner states, the condition in
\eqref{eq:Werner-k-ext} translates to%
\begin{equation}
p\leq\frac{1}{2}\left(  \frac{d-1}{k}+1\right)  .
\end{equation}
This concludes the proof.
\end{proof}

For $p,q\in\left[  0,1\right]  $ and for any generalized divergence
$\mathbf{D}$, we make the following abbreviation:%
\begin{equation}
\mathbf{D}(p\Vert q)\coloneqq \mathbf{D}(\kappa(p)\Vert\kappa(q)),
\end{equation}
where%
\begin{equation}
\kappa(x)=x|0\rangle \! \langle0|+(1-x)|1\rangle \! \langle1|.
\end{equation}

We then have the following:

\begin{proposition}\label{prop:generalized-divergence-isotropic}
The $k$-unextendible generalized divergence of a Werner state $W_{AB}^{(p,d)}$ and an
isotropic state $\rho_{AB}^{(t,d)}$\ are respectively equal to%
\begin{align}
\mathbf{E}_{k}(A;B)_{W^{(p,d)}} &  =\inf_{q\in\left[  0,\frac{1}{2}\left(
\frac{d-1}{k}+1\right)  \right]  }\mathbf{D}(p\Vert q),\\
\mathbf{E}_{k}(A;B)_{\rho^{(t,d)}} &  =\inf_{q\in\left[  0,\frac{1}{d}\left(
\frac{d-1}{k}+1\right)  \right]  }\mathbf{D}(t\Vert q).
\end{align}

\end{proposition}

\begin{proof}
By definition, $\mathbf{E}_{k}(A;B)_{W^{p}}$ involves an infimum with respect
to all possible $k$-extendible states. It is monotone with respect to all
$1$W-LOCC\ channels, and one such choice is the full bilateral twirl:%
\begin{align}
\mathcal{T}_{AB}^{W}(\omega_{AB}) \coloneqq \int
d\mu(U)\left[  U_{A}\otimes U_{B}\right]  \omega_{AB}\left[  U_{A}\otimes
U_{B}\right]  ^{\dag}.
\end{align}
Note that this can be implemented by a unitary two-design \cite{DLT02}. The Werner state is
invariant with respect to this channel, whereas any other $k$-extendible state
$\sigma_{AB}$ becomes a Werner state under this channel. Let $\sigma_{AB}$
denote an arbitrary $k$-extendible state. We thus have%
\begin{align}
\mathbf{D}(W_{AB}^{(p,d)}\Vert\sigma_{AB}) &  \geq\mathbf{D}(\mathcal{T}_{AB}%
^{W}(W_{AB}^{(p,d)})\Vert\mathcal{T}_{AB}^{W}(\sigma_{AB}))\\
&  =\mathbf{D}(W_{AB}^{(p,d)}\Vert\mathcal{T}_{AB}^{W}(\sigma_{AB}))\\
&  =\mathbf{D}(W_{AB}^{(p,d)}\Vert W_{AB}^{(r,d)}),
\end{align}
where in the last line, we have noted that $\mathcal{T}_{AB}^{W}(\sigma_{AB})$
is a Werner state and can thus be written as $W_{AB}^{(r,d)}$ for some
$r\in\left[  0,1\right]  $. Furthermore, by Theorem~\ref{thm:main-resource-theory-k-ext},\ $W_{AB}^{(r,d)}$ is a
$k$-extendible state since $\sigma_{AB}$ is by assumption. Thus, it suffices
to consider only $k$-extendible Werner states in the optimization of
$\mathbf{E}_{k}(A;B)_{W^{(p,d)}}$. Next, the following equality holds%
\begin{equation}
\mathbf{D}(W_{AB}^{(p,d)}\Vert W_{AB}^{(r,d)})=\mathbf{D}(p
\Vert r  ),
\end{equation}
because the quantum-to-classical channel%
\begin{equation}
\omega_{AB}\rightarrow\operatorname{Tr}\{\Pi_{AB}^{+}\omega_{AB}%
\}|0\rangle \! \langle0|+\operatorname{Tr}\{\Pi_{AB}^{-}\omega_{AB}\}|1\rangle\!
\langle1|
\end{equation}
takes a Werner state $W_{AB}^{(p,d)}$ to $\left(  1-p\right)  |0\rangle\!
\langle0|+p|1\rangle \! \langle1|$ and the classical-to-quantum channel%
\begin{equation}
\tau\rightarrow\langle0|\tau|0\rangle\frac{2}{d\left(  d+1\right)  }\Pi
_{AB}^{+}+\langle1|\tau|1\rangle\frac{2}{d\left(  d-1\right)  }\Pi_{AB}^{-}%
\end{equation}
takes $\left(  1-p\right)  |0\rangle \! \langle0|+p|1\rangle \! \langle1|$ back to
$W_{AB}^{(p,d)}$. Finally, we can conclude the first equality in the statement of
the theorem.

The reasoning for the second equality is exactly the same, but we instead
employ the bilateral twirl%
\begin{equation}
\mathcal{T}_{AB}^{I}(\omega_{AB})\coloneqq \int d\mu(U)\left[  U_{A}\otimes
U_{B}^{\ast}\right]  \omega_{AB}\left[  U_{A}\otimes U_{B}^{\ast}\right]
^{\dag}.
\end{equation}
This is a $k$-extendible channel, the isotropic states are invariant under this twirl, and all other states are projected to isotropic states under this twirl. Also, the channel%
\begin{equation}
\omega_{AB}\rightarrow\operatorname{Tr}\{\Phi_{AB}\omega_{AB}\}|0\rangle
\langle0|+\operatorname{Tr}\{\left(  I_{AB}-\Phi_{AB}\right)  \omega
_{AB}\}|1\rangle \! \langle1|
\end{equation}
takes an isotropic state $\rho_{AB}^{(t,d)}$ to $t|0\rangle \! \langle0|+\left(
1-t\right)  |1\rangle \! \langle1|$ and the classical-to-quantum channel%
\begin{equation}
\tau\rightarrow\langle0|\tau|0\rangle\Phi_{AB}+\langle1|\tau|1\rangle
\frac{I_{AB}-\Phi_{AB}}{d^{2}-1}%
\end{equation}
allows for going back. These statements allow us to conclude the second inequality. 
\end{proof}

\begin{lemma}
Let $1>p>q>0$. Then the relative entropy $D(p\Vert q)$ is a monotone
decreasing function of $q$ for $p>q>0$. That is, for $1>p>q>r>0$, the
following inequality holds%
\begin{equation}
D(p\Vert r)>D(p\Vert q).
\end{equation}

\end{lemma}

\begin{proof}
To prove the statement, we show that the derivative of $D(p\Vert q)$ with
respect to $q$ is negative. The derivative of $D(p\Vert q)$ with respect to
$q$ is equal to%
\begin{equation}
\frac{d}{dq}D(p\Vert q)=\frac{1-p}{1-q}-\frac{p}{q}.
\end{equation}
The condition that $\frac{d}{dq}D(p\Vert q)<0$ is thus equivalent to the
condition%
\begin{equation}
\frac{q}{1-q}<\frac{p}{1-p}.
\end{equation}
This latter condition holds because the function $x/(1-x)$ is a monotone
increasing function on the interval $x\in(0,1)$.\ That this latter claim is
true follows because the derivative of $x/(1-x)$ with respect to $x$ is given
by%
\begin{equation}
\frac{d}{dx}\left(  \frac{x}{1-x}\right)  =\frac{1}{1-x}+\frac{x}{\left(
1-x\right)  ^{2}},
\end{equation}
which is positive for $x\in(0,1)$.
\end{proof}

\begin{lemma}
Let $1>p>q>0$ and let $\alpha\in(0,1)\cup(1,\infty)$. Then the R\'{e}nyi relative
entropy $D_{\alpha}(p\Vert q)$ is a monotone decreasing function of $q$ for
$p>q>0$. That is, for $1>p>q>r>0$, the following inequality holds%
\begin{equation}
D_{\alpha}(p\Vert r)>D_{\alpha}(p\Vert q).
\end{equation}

\end{lemma}

\begin{proof}
To prove the statement, we show that the derivative of $D_{\alpha}(p\Vert q)$
with respect to $q$ is negative. The derivative of $D_{\alpha}(p\Vert q)$ with
respect to $q$ is equal to%
\begin{align}
\frac{d}{dq}D_{\alpha}(p\Vert q) &  =\left[  1-q+\frac{1}{\left(  \frac
{q}{1-q}/\frac{p}{1-p}\right)  ^{\alpha}-1}\right]  ^{-1}\\
&  =\frac{\left(  \frac{q}{1-q}/\frac{p}{1-p}\right)  ^{\alpha}-1}{\left[
\left(  \frac{q}{1-q}/\frac{p}{1-p}\right)  ^{\alpha}-1\right]  \left[
1-q\right]  +1}.
\end{align}
Since $\frac{p}{1-p}>\frac{q}{1-q}$ for $1>p>q>0$ (as shown in the previous
proof), it follows that
\begin{equation}
\left(  \frac{q}{1-q}/\frac{p}{1-p}\right)  ^{\alpha}-1<0
\end{equation}
for all $\alpha\in(0,1)\cup(1,\infty)$. We would then like to prove that%
\begin{equation}
\left[  \left(  \frac{q}{1-q}/\frac{p}{1-p}\right)  ^{\alpha}-1\right]
\left[  1-q\right]  +1>0.
\end{equation}
Note that this is equivalent to%
\begin{equation}
\left[  1-\left(  \frac{q}{1-q}/\frac{p}{1-p}\right)  ^{\alpha}\right]
\left[  1-q\right]  <1,
\end{equation}
which follows because%
\begin{equation}
1-\left(  \frac{q}{1-q}/\frac{p}{1-p}\right)  ^{\alpha}\in(0,1)
\end{equation}
and $1-q\in(0,1)$. Thus, we can conclude that $\frac{d}{dq}D_{\alpha}(p\Vert
q)<0$ for $1>p>q>0$, and the statement of the lemma follows.
\end{proof}

With all of the above, we conclude the following:

\begin{proposition}\label{prop:relative-entropy-isotropic}
The $k$-unextendible relative entropy of a Werner state $W_{AB}^{(p,d)}$ and an isotropic state
$\rho_{AB}^{(t,d)}$\ are respectively equal to%
\begin{align}
& E_{k}(A;B)_{W^{(p,d)}} \notag \\
&  =\left\{
\begin{array}
[c]{cc}%
0 & \text{if }p\in\left[  0,\frac{1}{2}\left(  \frac{d-1}{k}+1\right)
\right]  \\
D(p\Vert\frac{1}{2}\left(  \frac{d-1}{k}+1\right)  ) & \text{else}%
\end{array}
\right.  ,\\
 & \notag \\
&E_{k}(A;B)_{\rho^{(t,d)}} \notag \\
&  =\left\{
\begin{array}
[c]{cc}%
0 & \text{if }p\in\left[  0,\frac{1}{d}\left(  \frac{d-1}{k}+1\right)
\right]  \\
D(t\Vert\frac{1}{d}\left(  \frac{d-1}{k}+1\right)  ) & \text{else}%
\end{array}
\right.  .
\end{align}
Similarly, the $k$-unextendible R\'{e}nyi divergences are given for $\alpha\in(0,1)\cup(1,\infty)$
by%
\begin{align}
& E_{k}^{\alpha}(A;B)_{W^{(p,d)}}   \notag \\
& =\left\{
\begin{array}
[c]{cc}%
0 & \text{if }p\in\left[  0,\frac{1}{2}\left(  \frac{d-1}{k}+1\right)
\right]  \\
D_{\alpha}(p\Vert\frac{1}{2}\left(  \frac{d-1}{k}+1\right)  ) & \text{else}%
\end{array}
\right.  ,\\
& \notag \\
& E_{k}^{\alpha}(A;B)_{\rho^{(t,d)}} \notag \\
&  =\left\{
\begin{array}
[c]{cc}%
0 & \text{if }p\in\left[  0,\frac{1}{d}\left(  \frac{d-1}{k}+1\right)
\right]  \\
D_{\alpha}(t\Vert\frac{1}{d}\left(  \frac{d-1}{k}+1\right)  ) & \text{else}%
\end{array}\right.
\end{align}
\end{proposition}

\subsubsection{Properties of \texorpdfstring{$k$}{k}-unextendible divergences of a bipartite state}\label{sec:properties-divergence}

In this section, we discuss some of the properties of an unextendible generalized divergence, focusing first on the quantity derived from quantum relative entropy. 
The $k$-unextendible relative entropy of a state $\rho_{AB}$ is given by Definition~\ref{def:unext-gen-div-state}, by replacing $\bf{D}$ with 
the quantum relative entropy $D$. 
In particular, we prove several properties of unextendible relative entropy, including uniform continuity (Lemma~\ref{prop:continuity}), faithfulness (Lemma~\ref{prop:faithfulness}), subadditivity,  additivity under tensor-product states (Lemma~\ref{prop:subadditivity}), and convexity (Lemma~\ref{prop:convexity}).

We begin by proving the uniform continuity of unextendible relative entropy. In order to do so, we use the following result \cite{Win16} concerning the relative entropy distance with respect to any closed, convex set $C$ of
states, or more generally positive semi-definite operators:
\begin{equation}
D_C(\rho)\coloneqq \inf_{\gamma\in C}D(\rho\Vert\gamma).
\end{equation}

\begin{lemma}[\cite{Win16}]\label{winter:continuity} 
For a closed, convex, and bounded set $C$ of
positive semi-definite operators, containing at least one of
full rank, let
\begin{equation}
\kappa \coloneqq  \sup_{\tau,\tau'}[ D_C(\tau)-D_C(\tau')]
\end{equation}
be the largest variation of $D_C$. Then, for any two states $\rho$ and $\sigma$ for which $\frac{1}{2}\|\rho-\sigma\|_1
\leq\varepsilon$, with $\varepsilon\in [0,1]$, we have that
\begin{equation}
|D_C(\rho)-D_C(\sigma)|\leq \varepsilon \kappa +g(\varepsilon),
\end{equation}
where $g(\varepsilon)\coloneqq  (\varepsilon+1)\log_2(\varepsilon+1)-\varepsilon \log_2 \varepsilon$.
\end{lemma}

\begin{lemma}[Uniform continuity]\label{prop:continuity}
For any two bipartite states $\rho_{AB}$ and $\sigma_{AB}$ acting on the composite Hilbert space $\mathcal{H}_A\otimes\mathcal{H}_B$, with $d = \min\{|A|,|B|\}$, and 
\begin{equation}
\frac{1}{2}\|\rho_{AB}-\sigma_{AB}\|\leq\varepsilon \in [0,1],
\end{equation}
we have that
\begin{equation}
|E_k(A ;B)_\rho-E_k(A ;B)_\sigma|\leq  \varepsilon\log_2 \min\{d,k\} +g(\varepsilon).
\end{equation}
\end{lemma}
\begin{proof}
This follows directly from  Lemma~\ref{winter:continuity}. To see this, observe that we have the following inequalities holding for any states $\tau_{AB}$ and $\tau'_{AB}$:
\begin{align}
E_k(A;B)_{\tau'} & \geq 0, \\
E_k(A;B)_{\tau} & \leq E_R(A;B)_{\tau}\\
&\leq \min\{S(A)_\tau, S(B)_\tau\}  \\
&\leq \log d,
\end{align}
where $E_R(A;B)_{\tau}$ denotes the relative entropy of entanglement \cite{VP98,HHHH09}. 

Finally, we obtain the $\log_2 k$ upper bound on $E_k(A;B)_{\tau}$ by picking the $k$-extendible state for $E_k(A;B)_{\tau} = \inf_{\sigma_{AB} \in \kex(A:B)} D(\tau_{AB}\Vert \sigma_{AB})$ as
\begin{equation}
\sigma_{AB} = \frac{1}{k} \tau_{AB} 
+ \(1 - \frac{1}{k}\)\tau_{A} \otimes \tau_{B}.  
\end{equation}
Such a state is $k$-extendible with a $k$-extension given by
\begin{equation}
\sigma_{AB_1 \cdots B_k} = \frac{1}{k}
\sum_{i=1}^k \tau_{B_1} \otimes \cdots 
\otimes \tau_{B_{i-1}} \otimes \tau_{AB_i}
\otimes \tau_{B_{i+1}} \otimes \cdots \otimes 
\tau_{B_{k}}.
\end{equation}
Then by using the facts that $D(\rho \Vert \sigma) \geq D(\rho \Vert \sigma')$ for $0 \leq \sigma \leq \sigma'$ and $D(\rho \Vert c \sigma) = D(\rho \Vert \sigma) - \log_2 c$ for $c > 0 $, we find that 
\begin{align}
E_k(A;B)_{\tau} & = \inf_{\sigma_{AB} \in \kex(A:B)} D(\tau_{AB}\Vert \sigma_{AB}) \label{eq:log-k-bnd-1}\\
& \leq D\!\(\tau_{AB}\middle \Vert \frac{1}{k} \tau_{AB} 
+ \(1 - \frac{1}{k}\)\tau_{A} \otimes \tau_{B}\) \\
& \leq D(\tau_{AB} \Vert \tau_{AB}) - \log_2 (1/k) = \log_2 k. \label{eq:log-k-bnd-last}
\end{align} 
This concludes the proof.
\end{proof}

\begin{lemma}[Faithfulness]\label{prop:faithfulness}
Fix $\varepsilon \in [0,1]$.
The $k$-unextendible relative entropy $E_k(A ;B)_\rho$ of any arbitrary state $\rho_{A B}$ is a faithful measure, in the following sense: If $E_k(A;B)_{\rho}\leq \varepsilon$, then
\begin{equation}
 \inf_{\sigma_{AB}\in \kex(A:B)}\|\rho_{AB}-\sigma_{AB}\|_1\leq \sqrt{\varepsilon \cdot 2
\ln 2 }
\end{equation}
 and if 
$\inf_{\sigma_{AB}\in \kex(A:B)}
\frac{1}{2}\|\rho_{AB}-\sigma_{AB}\|_1\leq \varepsilon$, then \begin{equation}
E_k(A;B)_{\rho}\leq \varepsilon\log_2 \min\{ d,k\} + g(\varepsilon).
\end{equation}
\end{lemma}
\begin{proof}
The proof of the first statement follows directly from the quantum Pinsker inequality \cite[Theorem~1.15]{OP93}. The second statement follows directly from Lemma~\ref{prop:continuity}. 
\end{proof}
%---------------------------------------------------------

\begin{lemma}[Subadditivity and non-extensivity]\label{prop:subadditivity}
For a state $\rho_{A_1B_1A_2B_2 \cdots A_n B_n}\coloneqq\omega^{(1)}_{A_1 B_1}\otimes\omega^{(2)}_{A_2 B_2} \otimes \cdots \otimes \omega^{(n)}_{A_n B_n}$,
the $k$-unextendible relative entropy is sub-additive and non-extensive, in the sense that
\begin{multline}
E_k(A_1 A_2 \cdots A_n;B_1 B_2 \cdots B_n)_\rho\\
\leq \min\left\{\log_2 k,\  \sum_{i=1}^n E_k(A_i;B_i)_{\omega^{(i)}}\right\}.
\end{multline}
In fact, the non-extensivity bound 
\begin{equation}
E_k(A_1 A_2 \cdots A_n;B_1 B_2 \cdots B_n)_\rho\leq \log_2 k
\end{equation}
applies to an arbitrary state $\rho_{A_1B_1A_2B_2 \cdots A_n B_n}$.
\end{lemma}
\begin{proof}The subadditivity proof is straightforward. We show it for a tensor product of two states and note that the general statement follows from induction:
\begin{align}
& E_k(A_1 A_2 ;B_1 B_2)_\rho 
\notag \\
&=\inf_{\substack{\sigma_{A_1A_2B_1B_2}\in \\
\kex(A_1 A_2 :B_1 B_2)}}D(\omega_{A_1B_1}\otimes\tau_{A_2B_2}\Vert\sigma_{A_1A_2B_1B_2}) \notag \\
&\leq 
\inf_{\substack{\sigma_{A_1B_1} \otimes \sigma_{A_2 B_2}\in\\ \kex(A_1 A_2 :B_1 B_2)}}
D(\omega_{A_1B_1}\otimes\tau_{A_2B_2}\Vert\sigma_{A_1 B_1}\otimes\sigma_{A_2B_2}) \notag \\
&= \inf_{\sigma_{A_1B_1}\in \kex(A_1:B_1)}D(\omega_{A_1B_1}\Vert\sigma_{A_1B_1})\notag \\
& \qquad +\inf_{\sigma_{A_2B_2}\in \kex(A_2:B_2)}D(\tau_{A_2B_2}\Vert\sigma_{A_2B_2}) \notag  \\
& = E_k(A_1;B_1)_\omega+E_k(A_2;B_2)_\tau.
\end{align}
The first equality follows from the definition. The first inequality follows from a particular choice of $\sigma_{A_1A_2B_1B_2}$. The second inequality follows from additivity of relative entropy with respect to tensor-product states.

The proof of the non-extensivity upper bound of $\log_2 k$ follows from the same reasoning as in \eqref{eq:log-k-bnd-1}--\eqref{eq:log-k-bnd-last}.
\end{proof}

%-------------------------------------------------------------
\begin{lemma}[Convexity]\label{prop:convexity}
Let a bipartite state $\rho_{AB}=\sum_{x\in\mc{X}}p_X(x)\rho^x_{AB}$, where $p_X(x)$ is a probability distribution and $\{\rho_{AB}^x\}_x$ is a set of quantum states. Then, the
$k$-unextendible relative entropy is convex, in the sense that
\begin{equation}
E_k(A;B)_\rho\leq \sum_{x\in\mc{X}}p_X(x)E_k(A;B)_{\rho^x}.
\end{equation}
\end{lemma}

\begin{proof}
Let $\sigma^x_{AB}$ be the $k$-extendible state that achieves the minimum for $\rho_{AB}^x$ in $E_k(A;B)_{\rho^x}$. Then,
\begin{align}
E_k(A;B)_{\rho}&=\inf_{\sigma_{AB}\in \kex(A:B)} D(\rho_{AB}\Vert\sigma_{AB})\\
&\leq D\!\(\sum_x p_X(x)\rho_{AB}^x\middle \Vert\sum_x p_X(x)\sigma_{AB}^x\)\\
&\leq \sum_{x}p_X(x)D(\rho_{AB}^x\Vert\sigma_{AB}^x)\\
&=\sum_x p_X(x)E_k(A;B)_{\rho}.
\end{align}
The second inequality follows from the joint convexity of quantum relative entropy.
\end{proof}
\bigskip

The following lemmas have straightforward proofs, making use of the additivity of sandwiched R\'enyi relative entropy with respect to tensor-product states, as well as its joint quasi-convexity:

\begin{lemma}[Subadditivity and non-extensivity]
For a state $\rho_{A_1B_1A_2B_2 \cdots A_n B_n}\coloneqq\omega^{(1)}_{A_1 B_1}\otimes\omega^{(2)}_{A_2 B_2} \otimes \cdots \otimes \omega^{(n)}_{A_n B_n}$
and $\alpha\in (0,1)\cup(1,\infty)$,
the $k$-unextendible $\alpha$-sandwiched-R\'{e}nyi divergence is sub-additive and non-extensive, in the sense that
\begin{multline}
\wt{E}^\alpha_k(A_1 A_2 \cdots A_n;B_1 B_2 \cdots B_n)_\rho\\
\leq \min\left\{\log_2 k,\  \sum_{i=1}^n \wt{E}^\alpha_k(A_i;B_i)_{\omega^{(i)}}\right\}.
\end{multline}
In fact, the non-extensivity bound 
\begin{equation}
\wt{E}_k^{\alpha}(A_1 A_2 \cdots A_n;B_1 B_2 \cdots B_n)_\rho\leq \log_2 k    
\end{equation}
applies to an arbitrary state $\rho_{A_1B_1A_2B_2 \cdots A_n B_n}$.
\end{lemma}

%------------------------------
\begin{lemma}
The $k$-unextendible $\alpha$-sandwiched-R\'{e}nyi divergence is quasi-convex; i.e., if $\rho_{AB}\in\mc{D}(\mc{H}_{AB})$  decomposes as $\rho_{AB}=\sum_{x\in\mc{X}}p_X(x)\rho^x_{AB}$, where $\sum_{x\in\mc{X}}p_X(x)=1$ and each $\rho^x_{AB}\in\mc{D}(\mc{H}_{AB})$, then
\begin{equation}
\wt{E}^\alpha_{k}(A;B)_{\rho}\leq \sup_{x}\wt{E}^\alpha_{k}(A;B)_{\rho^x}.
\end{equation}
\end{lemma}
%----------------------------------------

\section{Unextendibility, non-asymptotic one-way distillable entanglement, and non-asymptotic quantum capacity}\label{sec:communications}

In this section, we use the resource theory of unextendibility to derive non-asymptotic converse bounds on the rate at which entanglement can be transmitted over a finite number of uses of a quantum channel. We do the same for the non-asymptotic one-way distillable entanglement of a bipartite state.

\subsection{Entanglement transmission codes and one-way entanglement distillation protocols}

An $(n,M,\varepsilon)$ entanglement transmission protocol accomplishes the task of entanglement transmission over $n$ independent uses of a quantum channel $\mc{N}_{A\to B}$. The case of $n=1$ is known as ``one-shot entanglement transmission,'' given that we are considering just a single use of a channel in this case. However, note that a given $(n,M,\varepsilon)$ entanglement transmission protocol for the channel 
$\mc{N}_{A\to B}$ can be considered as a 
$(1,M,\varepsilon)$ entanglement transmission protocol
for the channel $\mc{N}_{A\to B}^{\otimes n}$.

An entanglement transmission code for $\mc{N}$, is specified by a triplet $\{M,\mc{E},\mc{D}\}$, where $M=\dim(\mc{H}_R)$ is the Schmidt rank of a maximally entangled state $\Phi_{RA^\prime}$, one share of which is to be transmitted over $\mc{N}$. The quantum channels $\mc{E}_{A^\prime\to A^n}$ and $\mc{D}_{B^n\to \hat{A}}$ are encoding and decoding channels, respectively. An $(n,M,\varepsilon)$ code is such that
\begin{equation}
\label{eqn:protocoltest}
F(\Phi_{R\hat{A}},\omega_{R\hat{A}})\geq 1-\varepsilon, 
\end{equation}
where
\begin{equation}
\omega_{R\hat{A}} \coloneqq \(\mc{D}_{B^n\to \hat{A}}\circ\mc{N}^{\otimes n}_{A\to B}\circ\mc{E}_{A^\prime\to A^n}\)\(\Phi_{RA^\prime}\).
\end{equation}
We note that the criterion $F(\Phi_{R\hat{A}},\omega_{R\hat{A}})\geq 1-\varepsilon$ is equivalent to  
 \begin{equation} \label{eq:test}
 \Tr\{\Phi_{R\hat{A}}\omega_{R\hat{A}}\}\geq 1-\varepsilon. 
 \end{equation}
 
We can also consider a modification of the above protocol in which the final decoding is a $k$-extendible channel $\mc{D}_{RB^n\to R \hat{A}}$, acting on the input systems $R:B^n$ and outputting the systems $R:\hat{A}$. See Figure~\ref{fig:protocol-k-ext-post-proc} for a depiction of such a modified protocol. We call such a protocol entanglement transmission assisted by a $k$-extendible post-processing, and the resulting non-asymptotic quantum capacity is denoted by $Q_{\text{I}}^{(k)}(\mathcal{N}
_{A\rightarrow B},n,\varepsilon)$.

\begin{figure}
	\begin{center}		{\includegraphics[width=1\columnwidth]{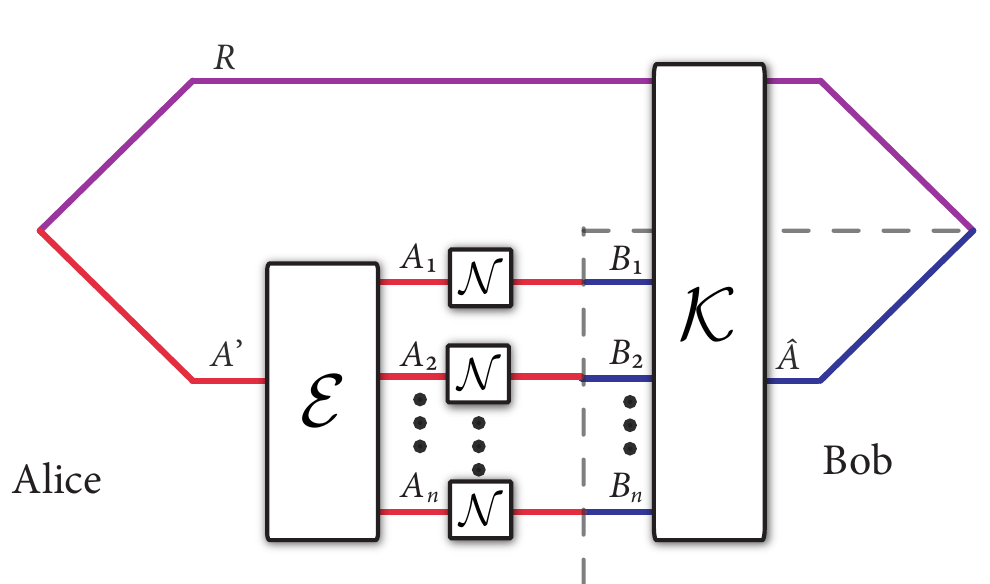}}
	\end{center}
	\caption{Depiction of an entanglement transmission protocol assisted by a $k$-extendible post-processing channel. The quantum channel $\mathcal{N}$ is used $n$ times, in conjunction with an encoding channel $\mathcal{E}_{A' \to A^n}$ and a $k$-extendible post-processing decoding channel
	$\mc{K}_{RB^n\to R\hat{A}}$, in order to establish entanglement shared between Alice and Bob.  }%
	\label{fig:protocol-k-ext-post-proc}%
\end{figure}

Another kind of protocol to consider is a one-way entanglement distillation protocol. An $(n,M,\varepsilon)$ one-way entanglement distillation protocol begins with Alice and Bob sharing $n$ copies of a bipartite state $\rho_{AB}$. They then act with a 1W-LOCC channel $\mc{L}_{A^n B^n \to M_A M_B}$ on $\rho_{AB}^{\otimes n}$, and the resulting state satisfies
\begin{equation}
F(\mc{L}_{A^n B^n \to M_A M_B}(\rho_{AB}^{\otimes n}), \Phi_{M_A M_B}) \geq 1 - \varepsilon,
\end{equation}
where $\Phi_{M_A M_B}$ is a maximally entangled state of Schmidt rank $M$. We can also modify this protocol to allow for a $k$-extendible channel instead of a 1W-LOCC channel, and the resulting protocol is an $(n,M,\varepsilon)$ entanglement distillation protocol assisted by a $k$-extendible channel.
Let $D^{(k)}(\rho_{AB},n,\varepsilon)$ denote the non-asymptotic
distillable entanglement with the assistance of $k$-extendible channels; i.e.,
$D^{(k)}(\rho_{AB},n,\varepsilon)$ is equal to the maximum value of $\frac
{1}{n}\log_{2}M$ such that there exists an $(n,M,\varepsilon)$ protocol for
$\rho_{AB}$ as described above.

% We can extend one-shot entanglement transmission protocol to an $k$-extendible-assisted $(n,M,\varepsilon)$ entanglement transmission protocol by considering $n$ channel $\mc{N}$ uses interleaved by $k$-extendible channels as discussed in \cite{KW17a}. 

\subsection{Bounds on non-asymptotic quantum capacity and one-way distillable entanglement in terms of $k$-extendible divergence}

We now establish an upper bound on the non-asymptotic quantum capacity in terms of the unextendible hypothesis testing divergence:
\begin{theorem} 
\label{thm:k-ext-q-cap-bnd}
The following bound holds for all $ k\in\bb{N}$ and for  every $(1,M,\varepsilon)$ entanglement transmission protocol over a quantum channel $\mathcal{N}$ and assisted by a $k$-extendible post-processing:
\begin{equation}
-\log_2\!\left[\frac{1}{M}+\frac{1}{k} - \frac{1}{M k}\right]\leq \sup_{\psi_{RA}} E^\varepsilon_{k}(R;B)_{\tau}, \label{eq:k-ext-q-cap-bnd}
\end{equation}
where 
\begin{equation}
\label{def:hypothesis}
E^\varepsilon_{k}(R;B)_{\tau}\coloneqq \inf_{\sigma_{RB}\in\kex(R;B)}D^\varepsilon_h\(\tau_{RB}\Vert\sigma_{RB}\)
\end{equation}
is the $k$-unextendible $\varepsilon$-hypothesis-testing divergence, $\tau_{RB}\coloneqq \mc{N}_{A\to B}(\psi_{RA})$, and the optimization in \eqref{eq:k-ext-q-cap-bnd} is with respect to pure states $\psi_{RA}$ such that $\dim(\mc{H}_R)=\dim(\mc{H}_A)$. Similarly, the following bound holds for any $(1,M,\varepsilon)$ entanglement distillation  protocol for a state $\rho_{AB}$, which is  assisted by a $k$-extendible post-processing:
\begin{equation}
-\log_2\!\left[\frac{1}{M}+\frac{1}{k} - \frac{1}{M k}\right]\leq  E^\varepsilon_{k}(A;B)_{\rho}, \label{eq:k-ext-end-dist-bnd}
\end{equation}
\end{theorem}

\begin{proof}
Suppose that there exists a $(1,M,\varepsilon)$ entanglement transmission protocol, assisted by a $k$-extendible post-processing, that satisfies the condition in \eqref{eqn:protocoltest}. 
Let $\sigma_{R\hat{A}}\in\kex(R;\hat{A})$, and let $\Phi_{R\hat{A}}$ denote a maximally entangled state. Then the following chain of inequalities holds
\begin{align}
&D^\varepsilon_{h}(\omega_{R\hat{A}}\Vert\sigma_{R\hat{A}})\notag \\
&\geq  -\log_2\Tr\{\Phi_{R\hat{A}}\sigma_{R\hat{A}}\}\\
&=-\log_2\Tr\left\{\int \d\!U\(U_R\otimes U^\ast_{\hat{A}}\)\Phi_{R\hat{A}}\(U_R\otimes U^\ast_{\hat{A}}\)^\dag\sigma_{R\hat{A}}\right\}\\
&=-\log_2\Tr\left\{\Phi_{R\hat{A}}\int \d\!U\(U_R\otimes U^\ast_{\hat{A}}\)^\dag \sigma_{R\hat{A}} \(U_R\otimes U^\ast_{\hat{A}}\)\right\}\label{eq:rel-ent-iso-sigma}.
\end{align}
The first inequality follows because the condition in
\eqref{eq:test} implies that we can  relax the measurement operator $\Lambda$ in
\eqref{eq:hypo-test-div} to be equal to
$\Phi_{R\hat{A}}$.  The first equality is due to the ``transpose trick'' property of the maximally entangled state, which leads to its $U \otimes U^*$ invariance. For the last equality, we use the cyclic property of the trace. 

Let
\begin{equation}
\overline{\sigma}_{R\hat{A}}\coloneqq \int \d\!U\(U_R\otimes U^\ast_{\hat{A}}\)^\dag \sigma_{R\hat{A}} \(U_R\otimes U^\ast_{\hat{A}}\) .   
\end{equation}
The state $\overline{\sigma}_{R\hat{A}}$ is $k$-extendible because $\sigma_{R\hat{A}}$ is and because the unitary twirl can be realized as a 1W-LOCC channel. The symmetrized state $\overline{\sigma}_{R\hat{A}}$ is furthermore isotropic because it
 is invariant under the action of a unitary of the form $U\otimes U^\ast$.  From Lemma~\ref{thm:isotropicstates}, we find that
\begin{equation}
\label{isotropic}
\overline{\sigma}_{R\hat{A}}=t\Phi_{R{\hat{A}}}+(1-t)\frac{I_{R\hat{A}}-\Phi_{R\hat{A}}}{M^2-1},
\end{equation}
for some $t\in \left[0,\frac{1}{M}+\frac{1}{k}-\frac{1}{Mk} \right]$.
Combining \eqref{isotropic} with \eqref{eq:rel-ent-iso-sigma} leads to
\begin{align}
D^\varepsilon_{h}(\omega_{R\hat{A}}\Vert\sigma_{R\hat{A}})
& \geq -\log_2 t \\
&\geq -\log_2\!\left[\frac{1}{M}+\frac{1}{k}-\frac{1}{Mk} \right].
\end{align}

 Since the above bound holds for an arbitrary state $\sigma_{R\hat{A}}\in \kex(R;\hat{A})$, we conclude that
\begin{align}
\label{eq:1step}
E^\varepsilon_{k}(R;\hat{A})_\omega & =\inf_{\sigma_{R\hat{A}}\in\kex(R;\hat{A})}D^\varepsilon_{h}(\omega_{R\hat{A}}\Vert\sigma_{R\hat{A}})\\
& \geq  -\log_2\!\left[\frac{1}{M}+\frac{1}{k}-\frac{1}{Mk} \right].
\end{align}
Let 
$\rho_{RB}\coloneqq \mc{N}_{A\rightarrow B}(\rho_{RA})$, where $\rho_{RA}\coloneqq  \mc{E}_{A^\prime\to A}\(\Phi_{RA^\prime}\)$, and let $\sigma_{RB}\in \kex(R;B)$. Then
for a $k$-extendible post-processing channel $\mc{D}_{RB\rightarrow R\hat{A}}$, we have that
\begin{align}
& \!\!\!\! D^\varepsilon_{h}(\rho_{RB}\Vert\sigma_{RB})\notag \\
&\geq D^\varepsilon_{h}(\mc{D}_{RB\rightarrow R\hat{A}}(\rho_{RB})\Vert \mc{D}_{RB\rightarrow R\hat{A}}(\sigma_{RB}))
\\
& = D^\varepsilon_{h}(\omega_{R\hat{A}}\Vert\sigma_{R\hat{A}}) \\
&\geq E^\varepsilon_{k}(R;\hat{A})_{\omega}.
\end{align}
The first inequality follows from the data-processing inequality for the hypothesis testing relative entropy. The channel $\mc{D}_{RB\rightarrow R \hat{A}}$ is a $k$-extendible channel, and given that $\sigma_{RB}\in \kex(R;B)$,  Theorem~\ref{thm:main-resource-theory-k-ext} implies that $\sigma_{R\hat{A}} \in \kex(R;\hat{A})$. The last inequality follows from the definition in \eqref{def:hypothesis}. Since this inequality holds for all $\sigma_{RB} \in \kex(R;B)$, we conclude that
\begin{equation}\label{eqn:monotone}
E^\varepsilon_{k}(R;B)_{\rho}\geq E^\varepsilon_{k}(R;\hat{A})_{\omega}.
\end{equation}
We now optimize $E^\varepsilon_{k}$ with respect to all inputs $\rho_{RA}$ to the channel $\mc{N}_{A \to B}$:
\begin{align}
\sup_{\rho_{RA}}E^\varepsilon_{k}(R;B)_{\mc{N}(\rho)}\geq E^\varepsilon_{k}(R;B)_{\mc{N}(\rho)}.
\end{align}
Using purification, the Schmidt decomposition theorem, and the data processing inequality of $E^\varepsilon_{k}(R;B)_{\rho}$, we find that
\begin{equation} \label{eq:f2}
\sup_{\rho_{RA}}E^\varepsilon_{k}(R;B)_{\mc{N}(\rho)}=\sup_{\psi_{RA}}E^\varepsilon_{k}(R;B)_{\mc{N}(\psi)}.
\end{equation}
for a pure state $\psi_{RA}$ with $|R| = |A|$.
Combining \eqref{eq:1step}, \eqref{eqn:monotone}, and \eqref{eq:f2}, we conclude the bound in \eqref{eq:k-ext-q-cap-bnd}.

By employing similar reasoning as above, we arrive at the bound in \eqref{eq:k-ext-end-dist-bnd}.
%find that
%\begin{equation}\label{eq:final_bound}
%\sup_{\psi_{RA}}E^\varepsilon_{k}(R;B)_{\mc{N}(\psi)}\geq -\log\left[\frac{1}{M}+\frac{1}{k}-\frac{1}{Mk} \right].
%\end{equation}
%Since, the bound on right hand side of \eqref{eq:final_bound} is monotone under the action of $k$-extendible channel, the proof concludes. 
\end{proof}

\begin{remark}
Note that Theorem~\ref{thm:k-ext-q-cap-bnd} applies in the case that the channel $\mathcal{N}$ is an infinite-dimensional channel, taking input density operators acting on a separable Hilbert space to output density operators acting on a separable Hilbert space. In claiming this statement, we are supposing that an entanglement transmission protocol begins with a finite-dimensional space, the encoding then maps to the infinite-dimensional space, the channel $\mathcal{N}$ acts, and then finally the decoding channel maps back to a finite-dimensional space. Furthermore, an entanglement distillation protocol acts on infinite-dimensional states and distills finite-dimensional maximally entangled states from them. We arrive at this conclusion because the $\varepsilon$-hypothesis testing relative entropy is well defined for infinite-dimensional states.
\end{remark}

\begin{remark}
Due to the facts that $D_h^{\varepsilon}(\rho\Vert\sigma) \geq D_h^{\varepsilon}(\rho\Vert\sigma')$ for $0 \leq \sigma \leq \sigma'$, $D_h^{\varepsilon}(\rho\Vert c \sigma) = D_h^{\varepsilon}(\rho\Vert\sigma) - \log_2 c$ for $c >0$ \cite[Lemma~7]{DTW14}, 
$D_h^{\varepsilon}(\rho\Vert \rho) = \log_2\!\(\frac{1}{1-\varepsilon}\)$,
and by applying the same reasoning as in 
\eqref{eq:log-k-bnd-1}--\eqref{eq:log-k-bnd-last}, we conclude that
\begin{equation}
\sup_{\psi_{RA}} E^\varepsilon_{k}(R;B)_{\tau} \leq \log_2\!\(\frac{1}{1-\varepsilon}\)+\log_2 k,
\end{equation}
which provides a limitation on the $(\varepsilon,k)$-unextendibility of any quantum channel.
%Applying \eqref{eq:k-ext-q-cap-bnd}, we find that the following bound holds $\forall k\in\bb{N}$ and for  any $(1,M,\varepsilon)$ entanglement transmission protocol conducted over a quantum channel $\mathcal{N}$ and assisted by a $k$-extendible post-processing:
%\begin{equation}
%-\log_2\!\left[\frac{1}{M}+\frac{1}{k}-\frac{1}{M k}\right]\leq \log_2\!\(\frac{1}{1-\varepsilon}\)+\log_2 k.
%\end{equation}
\end{remark}

By turning around the bound in \eqref{eq:k-ext-q-cap-bnd}, we find the following alternative way of expressing it:
\begin{remark}
\label{rem:alt-bound-formula}
The number of ebits ($\log_2 M$) transmitted  by a $(1,M,\varepsilon)$ entanglement transmission protocol over a quantum channel $\mc{N}$ and assisted by a $k$-extendible post processing is bounded from above as \begin{equation}
\log_2 M \leq \log_2\!\(\frac{k-1}{k}\) - \log_2\!\(2^{-\sup_{\psi_{RA}}E^{\varepsilon}_k(R;B)_{\tau}}-\frac{1}{k}\). 
\end{equation}
where $E^\varepsilon_{k}(R;B)_{\tau}$ is defined in \eqref{def:hypothesis}.
\end{remark}

\subsubsection{On the size of the extendibility parameter $k$ versus the error $\varepsilon$}

\label{sec:size-k}

By observing the form of the bound in Remark~\ref{rem:alt-bound-formula}, we see that it is critical for the inequality
\begin{equation}
2^{-\sup_{\psi_{RA}}E^{\varepsilon}_k(R;B)_{\tau}}-\frac{1}{k} > 0
\label{eq:ineq-k-for-solution}
\end{equation}
to hold in order for the bound to be non-trivial. Related, we see that this inequality always holds in the limit $k \to \infty$, and in this limit, we recover the $\varepsilon$-relative entropy of entanglement bound from \cite{TBR15,WTB17}. Here, we address the question of how large $k$ should be in order to ensure that the inequality in \eqref{eq:ineq-k-for-solution} holds.

\begin{proposition}
\label{prop:size-of-k} For a fixed $\varepsilon\in(0,1)$, the following
inequality holds%
\begin{equation}
2^{-E_{k}^{\varepsilon}(\mathcal{N})}-\frac{1}{k}>0,
\end{equation}
or equivalently, that%
\begin{equation}
E_{k}^{\varepsilon}(\mathcal{N})<\log_{2}k.\label{eq:k-cond}%
\end{equation}
as long as
\begin{equation}
k>2^{I_{h}^{\varepsilon}(\mathcal{N})}\varepsilon+1,
\end{equation}
where%
\begin{equation}
I_{h}^{\varepsilon}(\mathcal{N})\coloneqq \sup_{\psi_{RA}}D_{h}^{\varepsilon
}(\mathcal{N}_{A\rightarrow B}(\psi_{RA})\Vert\psi_{R}\otimes\mathcal{N}%
_{A\rightarrow B}(\psi_{A}))
\end{equation}
is the channel's $\varepsilon$-mutual information.
\end{proposition}

\begin{proof}
This follows because the condition in \eqref{eq:k-cond} is equivalent to%
\begin{align}
E_{k}^{\varepsilon}(\mathcal{N})
& = \sup_{\psi_{RA}}\inf_{\sigma_{RB}\in
{\EXT}_{k}(R;B)}D_{h}^{\varepsilon}(\mathcal{N}_{A\rightarrow B}(\psi
_{RA})\Vert\sigma_{RB})\notag \\
& <\log_{2}k.
\end{align}
We can pick the $k$-extendible state $\sigma_{RB}^{\psi}$, for a fixed
$\psi_{RA}$, as follows:%
\begin{equation}
\sigma_{RB}^{\psi}=\frac{1}{k}\mathcal{N}_{A\rightarrow B}(\psi_{RA})+\left(
1-\frac{1}{k}\right)  \psi_{R}\otimes\mathcal{N}_{A\rightarrow B}(\psi_{A}),
\end{equation}
implying that%
\begin{equation}
E_{k}^{\varepsilon}(\mathcal{N})\leq\sup_{\psi_{RA}}D_{h}^{\varepsilon
}(\mathcal{N}_{A\rightarrow B}(\psi_{RA})\Vert\sigma_{RB}^{\psi}).
\end{equation}
The choice $\sigma_{RB}^{\psi}$ is $k$-extendible because the following state
constitutes its $k$-extension:
\begin{multline}
\sigma_{RB_{1}\cdots B_{k}}^{\psi}=\frac{1}{k}\sum_{i=1}^{k}\mathcal{N}%
_{A\rightarrow B_{1}}(\psi_{A})\otimes\cdots\otimes\mathcal{N}_{A\rightarrow
B_{i-1}}(\psi_{A})\\
\otimes\mathcal{N}_{A\rightarrow B_{i}}(\psi_{RA}%
)\otimes\mathcal{N}_{A\rightarrow B_{i+1}}(\psi_{A})\otimes\cdots
\\
\otimes\mathcal{N}_{A\rightarrow B_{k}}(\psi_{A}).
\end{multline}
The optimal measurement operator $\Lambda^{\ast}$ for $D_{h}^{\varepsilon
}(\mathcal{N}_{A\rightarrow B}(\psi_{RA})\Vert\sigma_{RB}^{\psi})$ satisfies%
\begin{equation}
\operatorname{Tr}\{\Lambda^{\ast}\mathcal{N}_{A\rightarrow B}(\psi_{RA}%
)\}\geq1-\varepsilon,
\end{equation}
which means that%
\begin{align}
\operatorname{Tr}\{\Lambda^{\ast}\sigma_{RB}^{\psi}\} &  =\frac{1}%
{k}\operatorname{Tr}\{\Lambda^{\ast}\mathcal{N}_{A\rightarrow B}(\psi
_{RA})\}\notag \\
& \qquad +\left(  1-\frac{1}{k}\right)  \operatorname{Tr}\{\Lambda^{\ast}%
(\psi_{R}\otimes\mathcal{N}_{A\rightarrow B}(\psi_{A}))\}\notag \\
&  \geq\frac{1}{k}\left[  1-\varepsilon\right]  +\left(  1-\frac{1}{k}\right)
2^{-I_{h}^{\varepsilon}(\mathcal{N})},
\end{align}
and in turn that%
\begin{multline}
D_{h}^{\varepsilon}(\mathcal{N}_{A\rightarrow B}(\psi_{RA})\Vert\sigma
_{RB}^{\psi})\\
\leq-\log_{2}\!\left(    \frac{1}{k}\left[  1-\varepsilon
\right]  +\left(  1-\frac{1}{k}\right)  2^{-I_{h}^{\varepsilon}(\mathcal{N}%
)}  \right)  .
\end{multline}
The goal is to have the right-hand side above less than $\log_{2}k$ for all
$\psi_{RA}$, and this condition is equivalent to
\begin{equation}
-\log_{2}\!\left(    \frac{1}{k}\left[  1-\varepsilon\right]  +\left(
1-\frac{1}{k}\right)  2^{-I_{h}^{\varepsilon}(\mathcal{N})}  \right)
<\log_{2}k.
\end{equation}
Rewriting this, it is the same as%
\begin{equation}
\frac{1}{k}\left[  1-\varepsilon\right]  +\left(  1-\frac{1}{k}\right)
2^{-I_{h}^{\varepsilon}(\mathcal{N})}>\frac{1}{k},
\end{equation}
which is in turn the same as%
\begin{align}
-\frac{\varepsilon}{k}+\left(  1-\frac{1}{k}\right)  2^{-I_{h}^{\varepsilon
}(\mathcal{N})} &  >0\\
\Leftrightarrow\left(  k-1\right)  2^{-I_{h}^{\varepsilon}(\mathcal{N})} &
>\varepsilon\\
\Leftrightarrow k &  >2^{I_{h}^{\varepsilon}(\mathcal{N})}\varepsilon+1. 
\end{align}
This concludes the proof.
\end{proof}

\begin{remark}
\label{rem:lower-bnd-k-loose}
We note that the lower bound on $k$ from Proposition~\ref{prop:size-of-k} is not necessarily optimal and certainly could be improved. For example, when $\varepsilon < 1/2$ and the channel $\mc{N}$ is a two-extendible channel, $k=2$ suffices in order for the bound from Theorem~\ref{thm:k-ext-q-cap-bnd} to apply, and thus the bound in Proposition~\ref{prop:size-of-k} can be very loose. The value of Proposition~\ref{prop:size-of-k} is simply in knowing that a finite lower bound on $k$ exists for every channel, such that one can always find a finite $k$ for and beyond which our bound on entanglement transmission rates applies.
\end{remark}

\subsection{Non-asymptotic quantum capacity assisted by $k$-extendible channels}

In this subsection, we define another kind of non-asymptotic quantum capacity, in which a quantum channel is used $n$ times, and between every channel use, a $k$-extendible channel is employed for free to assist in the goal of entanglement transmission. Such a protocol is similar to those that have been discussed in the literature previously \cite{TGW14,KW17a,BW17}, but we review the details here for completeness.

In such a protocol (see Figure~\ref{fig:protocol-k-ext-two-way} for a depiction of an example), a sender Alice and a receiver Bob are
spatially separated and connected by a quantum channel $\mathcal{N}%
_{A\rightarrow B}$. They begin by performing a $k$-extendible channel $\mathcal{K}%
_{\emptyset\rightarrow A_{1}^{\prime}A_{1}B_{1}^{\prime}}^{(1)}$, which leads
to a $k$-extendible state $\rho_{A_{1}^{\prime}A_{1}B_{1}^{\prime}}^{(1)}$, where
$A_{1}^{\prime}$ and $B_{1}^{\prime}$ are systems that are finite-dimensional
but arbitrarily large. The system $A_{1}$ is such that it can be fed into the
first channel use. Alice sends system $A_{1}$ through the first channel use,
leading to a state $\sigma_{A_{1}^{\prime}B_{1}B_{1}^{\prime}}^{(1)}%
\coloneqq \mathcal{N}_{A_{1}\rightarrow B_{1}}(\rho_{A_{1}^{\prime}A_{1}%
B_{1}^{\prime}}^{(1)})$. Alice and Bob then perform the $k$-extendible channel
$\mathcal{K}_{A_{1}^{\prime}B_{1}B_{1}^{\prime}\rightarrow A_{2}^{\prime}%
A_{2}B_{2}^{\prime}}^{(2)}$, which leads to the state%
\begin{equation}
\rho_{A_{2}^{\prime}A_{2}B_{2}^{\prime}}^{(2)}\coloneqq \mathcal{K}_{A_{1}%
^{\prime}B_{1}B_{1}^{\prime}\rightarrow A_{2}^{\prime}A_{2}B_{2}^{\prime}%
}^{(2)}(\sigma_{A_{1}^{\prime}B_{1}B_{1}^{\prime}}^{(1)}).
\end{equation}
Alice sends system $A_{2}$ through the second channel use $\mathcal{N}%
_{A_{2}\rightarrow B_{2}}$, leading to the state $\sigma_{A_{2}^{\prime}%
B_{2}B_{2}^{\prime}}^{(2)}\coloneqq \mathcal{N}_{A_{2}\rightarrow B_{2}}%
(\rho_{A_{2}^{\prime}A_{2}B_{2}^{\prime}}^{(1)})$. This process iterates:\ the
protocol uses the channel $n$ times. In general, we have the following states
for all $i\in\{2,\ldots,n\}$:%
\begin{align}
\rho_{A_{i}^{\prime}A_{i}B_{i}^{\prime}}^{(i)}  &  \coloneqq \mathcal{K}%
_{A_{i-1}^{\prime}B_{i-1}B_{i-1}^{\prime}\rightarrow A_{i}^{\prime}A_{i}%
B_{i}^{\prime}}^{(i)}(\sigma_{A_{i-1}^{\prime}B_{i-1}B_{i-1}^{\prime}}%
^{(i-1)}),\\
\sigma_{A_{i}^{\prime}B_{i}B_{i}^{\prime}}^{(i)}  &  \coloneqq \mathcal{N}%
_{A_{i}\rightarrow B_{i}}(\rho_{A_{i}^{\prime}A_{i}B_{i}^{\prime}}^{(i)}),
\end{align}
where $\mathcal{K}_{A_{i-1}^{\prime}B_{i-1}B_{i-1}^{\prime}\rightarrow
A_{i}^{\prime}A_{i}B_{i}^{\prime}}^{(i)}$ is a $k$-extendible channel. The final step of
the protocol consists of a $k$-extendible channel $\mathcal{K}_{A_{n}^{\prime}%
B_{n}B_{n}^{\prime}\rightarrow M_{A}M_{B}}^{(n+1)}$, which generates the
systems $M_{A}$ and $M_{B}$ for Alice and Bob, respectively. The protocol's
final state is as follows:%
\begin{equation}
\omega_{M_{A}M_{B}}\coloneqq \mathcal{K}_{A_{n}^{\prime}B_{n}B_{n}^{\prime
}\rightarrow M_{A}M_{B}}^{(n+1)}(\sigma_{A_{n}^{\prime}B_{n}B_{n}^{\prime}%
}^{(n)}).
\end{equation}
%Figure~\ref{fig:private-code}\ depicts such a protocol.

The goal of the protocol is that the final state $\omega_{M_{A}M_{B}}$ is
close to a maximally entangled state. Fix $n,M\in\mathbb{N}$ and
$\varepsilon\in\lbrack0,1]$. The original protocol is an $(n,M,\varepsilon)$
protocol if the channel is used $n$ times as discussed above, $\left\vert
M_{A}\right\vert =\left\vert M_{B}\right\vert =M$, and if
\begin{align}
& F(\omega_{M_{A}M_{B}},\Phi_{M_{A}M_{B}})  \notag \\
&  =\langle\Phi|_{M_{A}M_{B}}%
\omega_{M_{A}M_{B}}|\Phi\rangle_{M_{A}M_{B}}\\
&  \geq1-\varepsilon \label{eq:fidelity-assump}.
\end{align}

\begin{figure*}
	\begin{center}		{\includegraphics[width=0.85\linewidth]{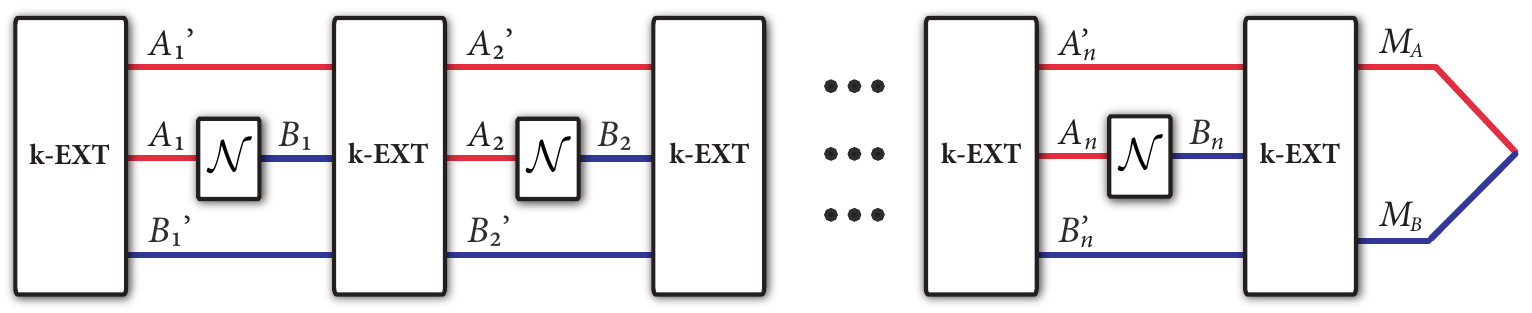}}
	\end{center}
	\caption{Depiction of a quantum communication protocol using a quantum channel $\mathcal{N}$ assisted by $k$-extendible channels before and after every channel use. The quantum channel $\mathcal{N}$ is used $n$ times, in conjunction with the assisting $k$-extendible channels, in order to establish entanglement shared between Alice and Bob.  }%
	\label{fig:protocol-k-ext-two-way}%
\end{figure*}

Let $Q_{\text{II}}^{(k)}%
(\mathcal{N}_{A\rightarrow B},n,\varepsilon)$ denote the non-asymptotic
quantum capacity assisted by $k$-extendible channels; i.e.,
$Q_{\text{II}}^{(k)}(\mathcal{N}_{A\rightarrow B},n,\varepsilon)$ is the
maximum value of $\frac{1}{n}\log_{2}M$ such that there exists an
$(n,M,\varepsilon)$ protocol for $\mathcal{N}_{A\rightarrow B}$ as described above.

A rate $R$ is achievable for $k$-extendible-assisted quantum communication if for all
$\varepsilon\in(0,1]$, $\delta>0$, and sufficiently large$~n$, there exists an
$(n,2^{n\left(  R-\delta\right)  },\varepsilon)$ protocol. The $k$-extendible-assisted
quantum capacity of a channel~$\mathcal{N}$, denoted as
$Q_{\text{II}}^{(k)}(\mathcal{N})$, is equal to the
supremum of all achievable rates.

\begin{theorem} \label{theorem:adaptive-protocols}
The following converse bound holds for every integer $ k\geq 2$ and for  every $(n,M,\varepsilon)$ $k$-extendible assisted quantum communication protocol over $n$ uses of a quantum channel $\mathcal{N}$:
\begin{equation}
-\frac{1}{n}\log_2\!\left[\frac{1}{M}+\frac{1}{k}-\frac{1}{M k}\right] 
 \leq E^{\max}_{k}(\mc{N})+ \frac{1}{n}\log_2\!\(\frac{1}{1-\varepsilon}\), 
  \label{eq:k-ext-q-cap-bnd-max}
\end{equation}
where $E^{\max}_{k}(\mc{N})$ is the $k$-unextendible max-relative entropy of the channel $\mc{N}$, defined as
\begin{equation}
E^{\max}_{k}(R;B)_{\rho}\coloneqq \inf_{\sigma_{RB}\in\kex(R:B)}D_{\max}\(\rho_{RB}\Vert\sigma_{RB}\),
\end{equation}
 $\tau_{RB}\coloneqq \mc{N}_{A\to B}(\psi_{RA})$, and the optimization is with respect to pure states $\rho_{RA}$ with $|R|=|A|$.
\end{theorem}
\begin{proof}
The above bound can be derived by invoking Proposition~\ref{prop:amort} and following arguments similar to those given in the proof of  \cite[Theorem~3]{BW17}. We also require the amortization collapse of $E^{\max}_{k}(\mc{N})$, as given in Appendix~\ref{app:amort-collapse}.
\end{proof}

\bigskip 
Similar to the observation in Remark~\ref{rem:alt-bound-formula},
by turning around the bound in \eqref{eq:k-ext-q-cap-bnd-max}, we find the following alternative way of expressing it:
\begin{remark}
\label{rem:alt-bound-formula-max}
The number of qubits ($\log_2 M$) transmitted  by an $(n,M,\varepsilon)$ $k$-extendible assisted quantum communication protocol conducted  over a quantum channel $\mc{N}$  is bounded from above as 
\begin{equation}
\log_2 M \leq \log_2\!\(\frac{k-1}{k}\) - \log_2\!\(2^{-n E^{\max}_k(\mathcal{N})}[1-\varepsilon]-\frac{1}{k}\).
\label{eq:alt-bound-max-k} 
\end{equation}
where $E^{\max}_{k}(\mc{N})$ is the $k$-unextendible max-relative entropy of the channel $\mc{N}$, as defined in \eqref{equation-67}.
\end{remark}

Related to the discussion in Section~\ref{sec:size-k}, it is necessary for the inequality $2^{-n E^{\max}_k(\mathcal{N})}[1-\varepsilon]-\frac{1}{k} > 0 $ to hold in order for the bound in \eqref{eq:alt-bound-max-k} to be non-trivial. The following proposition gives a sufficient condition on the size of $k$ in order for the inequality in \eqref{eq:alt-bound-max-k} to hold. This condition can be checked numerically.

\begin{proposition}
\label{prop:lower-bnd-k-adaptive}
Fix $\varepsilon \in (0,1)$, a channel $\mathcal{N}$, and $n \geq 1$.
The following inequality holds%
\begin{equation}
2^{-nE_{k}^{\max}(\mathcal{N})}\left[  1-\varepsilon\right]  -\frac{1}{k}>0,
\end{equation}
or equivalently,%
\begin{equation}
nE_{k}^{\max}(\mathcal{N})+\log_{2}\!
\left(  \frac{1}{1-\varepsilon}\right)
<\log_{2}k,\label{eq:k-cond-max}%
\end{equation}
as long as%
\begin{equation}
k>2^{I_{\max}(\mathcal{N})}\left[  \frac{k^{1-1/n}}{\left[  1-\varepsilon
\right]  ^{1/n}}-\left(  1-2^{-I_{\max}(\mathcal{N})}\right)  \right]  ,
\end{equation}
where%
\begin{equation}
I_{\max}(\mathcal{N})\coloneqq \sup_{\psi_{RA}}D_{\max}(\mathcal{N}_{A\rightarrow
B}(\psi_{RA})\Vert\psi_{R}\otimes\mathcal{N}_{A\rightarrow B}(\psi_{A}))
\end{equation}
is the channel's max-mutual information.
\end{proposition}

\begin{proof}
The condition in \eqref{eq:k-cond-max} is equivalent to%
\begin{align}
E_{k}^{\max}(\mathcal{N}) & =\sup_{\psi_{RA}}\inf_{\sigma_{RB}\in \EXT_k(R:B)}D_{\max
}(\mathcal{N}_{A\rightarrow B}(\psi_{RA})\Vert\sigma_{RB})\notag \\
& < \log_{2}k.
\end{align}
We can pick the $k$-extendible state $\sigma_{RB}^{\psi}$, for a fixed
$\psi_{RA}$, as follows:%
\begin{equation}
\sigma_{RB}^{\psi}=\frac{1}{k}\mathcal{N}_{A\rightarrow B}(\psi_{RA})+\left(
1-\frac{1}{k}\right)  \psi_{R}\otimes\mathcal{N}_{A\rightarrow B}(\psi_{A}),
\end{equation}
implying that%
\begin{equation}
E_{k}^{\max}(\mathcal{N})\leq\sup_{\psi_{RA}}D_{\max}(\mathcal{N}%
_{A\rightarrow B}(\psi_{RA})\Vert\sigma_{RB}^{\psi}).
\end{equation}
Now defining, for a fixed $\psi_{RA}$,%
\begin{align}
\lambda(\psi) & \coloneqq  I_{\max}(R;B)_{\mathcal{N}(\psi)}\\
& \coloneqq  D_{\max
}(\mathcal{N}_{A\rightarrow B}(\psi_{RA})\Vert\psi_{R}\otimes\mathcal{N}%
_{A\rightarrow B}(\psi_{A})),
\end{align}
we find that%
\begin{align}
& \sigma_{RB}^{\psi} \notag \\
&  =\frac{1}{k}\mathcal{N}_{A\rightarrow B}(\psi
_{RA})+\left(  1-\frac{1}{k}\right)  \psi_{R}\otimes\mathcal{N}_{A\rightarrow
B}(\psi_{A})\\
&  \geq\frac{1}{k}\mathcal{N}_{A\rightarrow B}(\psi_{RA})+\left(  1-\frac
{1}{k}\right)  2^{-\lambda(\psi)}\mathcal{N}_{A\rightarrow B}(\psi_{RA})\\
&  =\left[  \frac{1}{k}+\left(  1-\frac{1}{k}\right)  2^{-\lambda(\psi
)}\right]  \mathcal{N}_{A\rightarrow B}(\psi_{RA}).
\end{align}
Now exploiting the fact that $D_{\max}(\rho\Vert\sigma)\leq D_{\max}(\rho
\Vert\sigma^{\prime})$ for $\sigma\geq\sigma^{\prime} \geq 0$, as well as $D_{\max
}(\rho\Vert c\sigma)=D_{\max}(\rho\Vert\sigma)-\log_{2}c$ for $c>0$, we find
that%
\begin{align}
&  \!\!\!\!\!\sup_{\psi_{RA}}D_{\max}(\mathcal{N}_{A\rightarrow B}(\psi_{RA})\Vert
\sigma_{RB}^{\psi})\nonumber\\
&  \leq\sup_{\psi_{RA}}\Big[  D_{\max}(\mathcal{N}_{A\rightarrow B}(\psi
_{RA})\Vert\mathcal{N}_{A\rightarrow B}(\psi_{RA})) \notag \\
& \qquad -\log_{2}\!\left(  \frac
{1}{k}+\left(  1-\frac{1}{k}\right)  2^{-\lambda(\psi)}\right)  \Big]  \\
&  =\sup_{\psi_{RA}}\left[  -\log_{2}\!\left(  \frac{1}{k}+\left(  1-\frac{1}%
{k}\right)  2^{-\lambda(\psi)}\right)  \right]  \\
&  =-\log_{2}\!\left(  \frac{1}{k}+\left(  1-\frac{1}{k}\right)  2^{-I_{\max
}(\mathcal{N})}\right)  \\
&  =-\log_{2}\!\left(  2^{-I_{\max}(\mathcal{N})}+\frac{1}{k}\left(
1-2^{-I_{\max}(\mathcal{N})}\right)  \right)  .
\end{align}
The goal is to have the inequality in \eqref{eq:k-cond-max} holding, and, by
the above analysis, this results if the following inequality holds
\begin{multline}
-n\log_{2}\!\left(  \left[  2^{-I_{\max}(\mathcal{N})}+\frac{1}{k}\left(
1-2^{-I_{\max}(\mathcal{N})}\right)  \right]  \right)  \\
+\log_{2}\!\left(
\frac{1}{1-\varepsilon}\right)  <\log_{2}k.
\end{multline}
Rewriting this, it is the same as%
\begin{align}
\left[  2^{-I_{\max}(\mathcal{N})}+\frac{1}{k}\left(  1-2^{-I_{\max
}(\mathcal{N})}\right)  \right]  ^{n}\left[  1-\varepsilon\right]    &
>\frac{1}{k} \notag \\
\Leftrightarrow\left[  2^{-I_{\max}(\mathcal{N})}+\frac{1}{k}\left(
1-2^{-I_{\max}(\mathcal{N})}\right)  \right]  \left[  1-\varepsilon\right]
^{1/n}  & >\frac{1}{k^{1/n}} \notag \\
\Leftrightarrow\left[  k2^{-I_{\max}(\mathcal{N})}+\left(  1-2^{-I_{\max
}(\mathcal{N})}\right)  \right]  \left[  1-\varepsilon\right]  ^{1/n}  &
>k^{1-1/n} 
\end{align}
\begin{equation}
    \Leftrightarrow k2^{-I_{\max}(\mathcal{N})}+\left(  1-2^{-I_{\max}%
(\mathcal{N})}\right)     >\frac{k^{1-1/n}}{\left[  1-\varepsilon\right]
^{1/n}},
\end{equation}
\begin{equation}
    \Leftrightarrow k   >2^{I_{\max}(\mathcal{N})}\left[  \frac{k^{1-1/n}%
}{\left[  1-\varepsilon\right]  ^{1/n}}-\left(  1-2^{-I_{\max}(\mathcal{N}%
)}\right)  \right]  .
\end{equation}
This concludes the proof.
\end{proof}

\bigskip
A similar comment as in Remark~\ref{rem:lower-bnd-k-loose} applies to Proposition~\ref{prop:lower-bnd-k-adaptive}.

We now define $k$-simulable channels and observe how the upper bounds on non-asymptotic quantum capacity simplify for these channels.

\begin{definition}[$k$-simulable channels]
\label{def:k-sim-ch}
A channel $\mathcal{N}_{A \rightarrow B}$ is $k$-simulable with associated resource state $\omega_{R\hat{B}}$ if the following holds for every input state $\rho_A \in \mathcal{D}(\mathcal{H}_A)$:
\begin{equation}
\mathcal{N}_{A\rightarrow B}(\rho_A) = \mathcal{K}_{RA\hat{B}\rightarrow B} (\rho_A \otimes \omega_{R\hat{B}}),
\end{equation}
where $\mathcal{K}_{RAB\rightarrow B}$ is a $k$-extendible channel. 
\end{definition}

Note that a teleportation-simulable channel, as given in Definition~\ref{def:tel-sim}, is a particular example of a $k$-simulable channel, whenever the LOCC channel in \eqref{eq:TP-simul} is a 1W-LOCC channel.

For a $k$-simulable channel, an $(n,M,\varepsilon)$ quantum communication protocol assisted by $k$-extendible channels simplifies in such a way that it is equivalent to an $(n,M,\varepsilon)$ entanglement distillation protocol starting from the resource state $\omega_{R\hat{B}}^{\otimes n}$ and assisted by a $k$-extendible post-processing channel. This kind of observation was made in \cite{BDSW96,Mul12} and extended to any resource theory in \cite{KW17a}. See Figure~5 of \cite{KW17a} for a summary of the reduction that applies to our case of interest here.
We then have the following:  

\begin{corollary}
\label{cor:k-ext-k-sim-bnd}
Let $\mathcal{N}$ be a $k$-simulable channel as in Definition~\ref{def:k-sim-ch}.
The following bound holds for all $ k\in\bb{N}$ and for  every $(n,M,\varepsilon)$ quantum communication protocol conducted over the quantum channel $\mathcal{N}$ and assisted by $k$-extendible channels:
\begin{equation}
-\log_2\!\left[\frac{1}{M}+\frac{1}{k}-\frac{1}{M k}\right]\leq  E^\varepsilon_{k}(R^n;\hat{B}^n)_{\omega^{\otimes n}} \label{eq:k-ext-k-sim-bnd},
\end{equation}
where $\omega_{R\hat B}$ is the resource state in 
Definition~\ref{def:k-sim-ch}.
\end{corollary}

\section{Examples}

\label{sec:example}
We now showcase the above bounds for depolarizing and erasure channels. 

\subsection{Depolarizing Channel}

The action of a qubit depolarizing channel $\mathcal{D}^p_{A\to B}$ on an input state $\rho$ is as follows:
\begin{equation}
\mathcal{D}^p_{A\to B}(\rho) \coloneqq (1-p)\rho + \frac{p}{3}(X\rho X+Y\rho Y + Z \rho Z),
\end{equation}
where $p\in [0,1]$ is the depolarizing parameter and $X$, $Y$, and $Z$ are the Pauli operators. A depolarizing channel is a covariant channel for all $p\in[0,1]$, which is a fact that is easy to see after expressing its action as $\mathcal{D}^p_{A\to B}(\rho) = (1-q)\rho + q I/2$, for $q = 4p/3 $. This property is crucial to  obtain an upper bound on the unextendible $\varepsilon$-hypothesis-testing divergence of the depolarizing channel. 

To this end, we first argue that the optimal input state for $n$ independent uses of the depolarizing channel is an $n$-fold tensor product of the maximally entangled state $\Phi_{RA} = \frac{1}{2} \sum_{i,j\in \{0,1\}} |i\rangle \! \langle j |_R\otimes |i\rangle \! \langle j |_A$. For tensor-product channels, we can restrict the input state to be invariant under permutations of the input systems, due to Lemma~\ref{lemma:covariant_channels} in Appendix~\ref{appendix:exploiting-symmetries}. Also, for covariant channels, the input states that optimize the $k$-extendible relative entropy are of the form given in Lemma~\ref{lemma:covariant_channels} in Appendix~\ref{appendix:exploiting-symmetries}.  Therefore, it suffices to restrict the input state to be a tensor-power maximally entangled state; i.e., we conclude that
\begin{multline}
E_k^{\varepsilon}([\mathcal{D}^p]^{\otimes n}) = \\ \inf_{\sigma_{R^n B^n} \in \kex(R^n:B^n )}D_h^{\varepsilon}([\mathcal{D}^p_{A\to B}(\Phi_{RA})]^{\otimes n} \Vert \sigma_{R^n B^n}).
\end{multline} 

We make a particular choice of the $k$-extendible state $\sigma_{R^n B^n}$ above (which is not necessarily optimal) to be a tensor product of the isotropic states $\sigma_{AB}^{(t,2)}$, defined as
\begin{equation}\label{eq:iso-state}
\rho^{(t,d)}_{AB}=t\Phi^{d}_{AB}+(1-t)\frac{I_{AB}-\Phi^{d}_{AB}}{d^2-1},
\end{equation}
where $\Phi^{d}_{AB}$ denotes a maximally entangled state of Schmidt rank $d$, and $t\in [0,1]$. Note that the action of $\mc{D}^p$ on a maximally entangled state results in an isotropic state $\sigma_{AB}^{(p,2)}$ parametrized by $p$. 
% Then, the upper bound on \eqref{def:hypothesis} reduces to
% \begin{equation}\label{upper_bnd_KDWW}
% \frac{1}{n}D^\varepsilon_h\left.\(\(\sigma_{AB}^p\)^{\otimes n}\right\Vert \(\sigma_{AB}^t\)^{\otimes n}\)\geq -\frac{1}{n}\log \(\frac{1}{M}+\frac{M-1}{kM}\)
% \end{equation}
Since the states $\(\sigma^{(p,2)}_{AB}\)^{\otimes n}$ and $\(\sigma^{(t,2)}_{AB}\)^{\otimes n}$ are diagonal in the same basis, the $\varepsilon$-hypothesis testing relative entropy between the two states is equal to the $\varepsilon$-hypothesis testing relative entropy between the product Bernoulli probability distributions $\{1-p,p\}^{\times n}$ and $\{t,1-t\}^{\times n}$. We therefore obtain the following bound on the number of ebits transmitted by $n$ channel uses of the depolarizing channel:
\begin{multline}\label{eq:bound}
\frac{1}{n}\log_2 M \leq \frac{1}{n}\log_2\!\(\frac{k-1}{k}\) -\\ \frac{1}{n}\log_2\!\(2^{-D^\varepsilon_h\(\{1-p,p\}^{\times n}\left\Vert \{t,1-t\}^{\times n}\right)\right.}-\frac{1}{k}\).
\end{multline}
The resulting classical hypothesis testing relative entropy between the product Bernoulli  distributions can be distinguished exactly by the optimal Neyman-Pearson test \cite{PPV10}. 

Note that \eqref{eq:bound} converges to the upper bound given in \cite{TBR15} in the limit as $k \rightarrow \infty$. Refer to Figures~\ref{fig:plots_dep1} and \ref{fig:plots} for a comparison of various upper bounds on the non-asymptotic quantum capacity of the depolarizing channel. For tensor products of the isotropic states $\sigma_{AB}^{(t,2)}$, the numerics suggest that the minimizing state is either a $k=2$ extendible state or a separable state. If the minimizing state is a separable state, then the bound in \eqref{eq:bound} is equal to the TBR bound from \cite{TBR15}. 

\begin{figure}
\centering
\begin{tikzpicture} \label{fig-dep1}
	\begin{axis}[
	scale = 1,
    ylabel=Number of qubits transmitted,
	xlabel=Number $n$ of channel uses,
    xtick={2,5,10,15,20,25},
	xmin = 2,
	xmax = 28,
	ymin = 0,
	ymax = 8,
	tick label style={/pgf/number format/fixed, /pgf/number format/precision=3},
	legend style = {at = {(0.05,0.8)},anchor = north west}, 
	legend cell align = left,]
	\addplot[thick,dashed,red] table[x=number,y=upper_bnd,col sep=comma] {p15.txt};
	\addplot[thick,color=dgreen, dashdotted] table[x=number,y=upper_bnd_TBR,col sep=comma] {p15.txt};
	\addplot[thick,color=blue,dotted] table[x=number,y=WFD,col sep=comma] {p15.txt};
	\legend{\small{KDWW},\small{TBR}, \small{WFD}};
	\end{axis}
	\end{tikzpicture}
    %\captionsetup[figure]{font=small,labelfont=small}
\caption{Upper bounds on the number of qubits that can be reliably transmitted over a depolarizing channel with $p=0.1$, and $\varepsilon= 0.05$. The red dashed line is the bound from Theorem \ref{thm:k-ext-q-cap-bnd}. The green dash-dotted and blue dotted lines are upper bounds from \cite{TBR15} and \cite{WFD17}, respectively.}
\label{fig:plots_dep1}%
\end{figure}

\begin{figure}
\begin{tikzpicture}
\label{fig-dep2}
	\begin{axis}[
	scale = 1,
    ylabel=Number of qubits transmitted,
	xlabel=Number $n$ of channel uses,
    xtick={6,10,15,20,25},
	xmin = 6,
	xmax = 28,
	ymin = 0,
	ymax = 1,
	tick label style={/pgf/number format/fixed, /pgf/number format/precision=3},
	legend style = {at = {(0.05,0.80)},anchor = north west}, 
	legend cell align = left,]
	\addplot[thick,dashed,red] table[x=number1,y=upper_bnd,col sep=comma] {p25dep.txt};
	\addplot[thick,color=dgreen, dashdotted] table[x=number1,y=upper_bnd_TBR,col sep=comma] {p25dep.txt};
	\addplot[thick,color=blue,dotted] table[x=number1,y=rate1,col sep=comma] {p25dep.txt};
	\legend{\small{KDWW},\small{TBR}, \small{WFD}};
	\end{axis}
	\end{tikzpicture}
    %\captionsetup[figure]{font=small,labelfont=small}
\caption{Upper bounds on the number of qubits that can be reliably transmitted over a depolarizing channel with $p=0.25$, and $\varepsilon= 5\times 10^{-5}$. The red dashed line is the bound from Theorem~\ref{thm:k-ext-q-cap-bnd}. The green dash-dotted and blue dotted lines are upper bounds from \cite{TBR15} and \cite{WFD17}, respectively.}%
\label{fig:plots}%
\end{figure}

\subsection{Erasure channel}

The action of a qubit erasure channel \cite{GBP97} on an input density operator $\rho$ is as follows:
\begin{equation}
\mathcal{E}^p_{A\rightarrow B}(\rho_A) \coloneqq  (1-p)\rho_{B}+ p \op{e}_B,
\end{equation}
where $p\in[0,1]$ is the erasure parameter  and $\op{e}$ is a pure state, orthogonal to any input state. The optimal input state for $n$ uses of the erasure channel, when considering its unextendible generalized divergence, is the $n$-fold tensor product maximally entangled state $\Phi_{A'A}^{\otimes n}$. This follows also from the covariance of the erasure channel and Lemma~\ref{lemma:covariant_channels}. 
% The output of maximally entangled state after transmission over the erasure channel is
% \begin{equation}
% \mathcal{E}^p_{A'\rightarrow B}(\Phi_{AA'}) = p\Phi_{AB}+(1-p) \pi_A\otimes\op{e}_B.
% \end{equation}

Our goal is to obtain upper bounds on the entanglement transmission rate when using  the erasure channel $n$ times. Consider sending $n$ shares of the maximally entangled state $\Phi_{AA'}$ over $n$ uses of the erasure channel $\mc{E}^p_{A'\rightarrow B}$. The output state $\rho_{A_1B_1A_2B_2\cdots A_nB_n}$ has the form
\begin{equation}
\label{eq:erasure_min}
\rho_{A_1B_1A_2B_2\cdots A_nB_n}=\sum_{x^n \in\{0,1\}^n} p(x^n)\(\bigotimes_{j=1}^n\tau^{x_j}_{A_jB_j}\),
\end{equation}
where for all $j\in[n]$,
\begin{equation}
\tau^{x_j}_{A_jB_j}\in \left\{\Phi_{A_jB_j},\pi_{A_j}\otimes\op{e}_{B_j}\right\},
\end{equation}
and for all $x^n\in\{0,1\}^n$, $p(x^n) \in [0,1]$ is a product distribution such that $\sum_{x\in\{0,1\}^n}p(x^n)=1$. Due to an i.i.d.~application of the channels, we find that the probabilities $p(x^n)$ corresponding to a state $\tau_{A_1B_1A_2B_2\cdots A_nB_n}^{x^n}$ with the same number of  erasure symbols are equal. The total probability for having $\ell$ erasure symbols in the state $\rho_{A_1B_1A_2B_2\cdots A_nB_n}$ is equal to $\binom{n}{\ell}(1-p)^{n-\ell} p^{\ell}$, where $\ell \in \{0,\ldots, n\}$.

Without loss of generality, the block-diagonal form of the output state of $n$ uses of an erasure channel, when inputting a tensor-power maximally entangled state, allows us to restrict the class of $k$-extendible states $\sigma\in\kex(A^n;B^n)$, over which we optimize the unextendible $\varepsilon$-hypothesis testing relative entropy, to be of the form in \eqref{eq:erasure_min}, except with $p(x^n)$ a probability distribution that is not necessarily product and chosen such that the state is $k$-extendible. This follows because the state $\rho_{A_1B_1A_2B_2\cdots A_nB_n}$ is invariant under $n$ independent bilateral twirls, along with $n$ independent and incomplete measurements of the form $\{ \vert 0 \rangle \! \langle 0 \vert + \vert 1 \rangle \! \langle 1 \vert, \vert e \rangle \! \langle e \vert\}$ by Bob, while such a 1W-LOCC channel symmetrizes the $k$-extendible state to have the aforementioned form. We let $\sigma_{A_1B_1A_2B_2\cdots A_nB_n}$ be of the form in \eqref{eq:erasure_min} with coefficients (probabilities) set to $q(x^n)$. Furthermore, we note that $\rho_{A_1B_1A_2B_2\cdots A_nB_n}$ is permutation invariant after Alice and Bob perform a coordinated random permutation channel on their composite systems locally. This allows us to restrict the form of $\sigma_{A_1B_1A_2B_2\cdots A_nB_n}$ to be permutation invariant under such a symmetrizing permutation channel because it is a $k$-extendible channel. 

From the argument above, we find that the minimizing state has the block structure given in \eqref{eq:erasure_min}, and the coefficients for states in the sum with the same number of erasure symbols are equal. We now want to obtain conditions on the probabilities $q(x^n)$, where $x^n \in \{0,1\}^n$, from the $k$-extendibility of the state $\sigma_{A_1B_1A_2B_2\cdots A_nB_n}$. The constraints that we impose on $q(x^n)$ are not unique. That is, there could exist other constraints such that the state $\sigma_{A_1B_1A_2B_2\cdots A_nB_n}$ is still $k$-extendible.

Let us first consider $n=2$ channel uses. By what we discussed above, the minimizing $k$-extendible state $\sigma_{A_1B_1A_2B_2}$ then has the following form
\begin{multline}\label{k-ext-erasure}
\sigma_{A_1B_1A_2B_2} \coloneqq \\ c_0\Phi_{A_1B_1} \otimes \Phi_{A_2B_2}+
c_1\left(\Phi_{A_1B_1}\otimes \pi_{A_2}\otimes\op{e}_{B_2}\right. \\+ \left.\Phi_{A_2B_2} \otimes \pi_{A_1}\otimes\op{e}_{B_1}\right)\\
+  c_2\pi_{A_1}\otimes\op{e}_{B_1}\otimes \pi_{A_2}\otimes\op{e}_{B_2},
\end{multline}
where $\{c_i\}_i$ for $i\in\{0,1,2\}$ is a probability distribution such that $c_0 + 2 c_1 +  c_2 =1$. 
Focusing on the special case $k=2$, we now want to obtain constraints on each $c_i$ such that $\sigma_{A_1B_1A_2B_2}$ is a two-extendible state. To this end, we replace all the terms $\Phi_{A_iB_i}$ in the above state with the two-extendible state $\frac{1}{2}\Phi_{A_iB_i}+\left(1-\frac{1}{2}\right)\pi_{A_i}\otimes \op{e}_{B_i}$. We obtain the following state, which is guaranteed to be two-extendible by construction: 
\begin{multline}
\frac{c_0}{4} \Phi_{A_1B_1}\otimes \Phi_{A_1B_1} \\+  \left(\frac{c_0}{4}+\frac{c_1}{2}\right)
\left(\Phi_{A_1B_1}\otimes\pi_{A_2}\otimes\op{e}_{B_2}\right. +\\ \left.\pi_{A_1}\otimes\op{e}_{B_1}\otimes \Phi_{A_2B_2}\right)\\ +\left(\frac{c_0}{4}+c_1+c_2\right)\left(\pi_{A_1}\otimes\op{e}_{B_1} \otimes\pi_{A_2}\otimes\op{e}_{B_2} \right).
\end{multline}
Abbreviating the new coefficients as $b_0$, $b_1$, and $b_2$, the above approach leads to the following constraint on them such that the state $\sigma_{A_1B_1A_2B_2}$ is two-extendible:
\begin{align}
\begin{bmatrix}
b_0\\
b_1\\
b_2
\end{bmatrix}&=
\begin{bmatrix}
\frac{1}{4} &0 &0\\
\frac{1}{4}& \frac{1}{2}&0\\
\frac{1}{4} & 2\cdot\frac{1}{2}& 1
\end{bmatrix}
\begin{bmatrix}
c_0\\
c_1\\
c_2
\end{bmatrix}.
\end{align}

We now generalize the above procedure of obtaining two-extendible states for two channel uses to obtaining $k$-extendible states for $n$ channel uses. We obtain the following condition on the coefficients $b_i$:
\begin{align}
    \begin{bmatrix}
    b_0\\
    b_1\\
    b_2\\
    \vdots\\
    b_n
    \end{bmatrix} &=       
%           \begin{bmatrix}
%            \binom{n}{0}\left(\frac{1}{k}\right)^n& 0 &0 &\cdots &0\\           
%            \binom{n}{1}\left(\frac{1}{k}\right)^{(n-1)}\left(1-\frac{1}{k}\right)& \binom{2}{0} \left(\frac{1}{k}\right)^2&0 &\cdots &0 \\
%           \binom{n}{2}\left(\frac{1}{k}\right)^{(n-2)}\left(1-\frac{1}{k}\right)^{2}& \binom{2}{1}\left(\frac{1}{k}\right)\left(1-\frac{1}{k}\right) &\binom{1}{0} \frac{1}{k} &\cdots &0 \\
%            \binom{n}{3}\left(\frac{1}{k}\right)^{(n-3)}\left(1-\frac{1}{k}\right)^{3}& \binom{2}{0}\left(1-\frac{1}{k}\right)^2 &\binom{1}{1} (1-\frac{1}{k}) &\cdots &\binom{n}{0}      
%            \end{bmatrix} 
\mathbf{M}
          \begin{bmatrix}
           \binom{n}{0}c_{0} \\
           \binom{n}{1}c_{1} \\
           \binom{n}{2}c_{2} \\
           \vdots \\
           \binom{n}{n}c_{n} \\
         \end{bmatrix},
      \end{align}
where the general form of the matrix $\mathbf{M}_{(n+1)\times (n+1)}=[m_{u,v}]$ is given as 
\begin{align}
m_{u,v}=\binom{n-v}{u-v}\left(1-\frac{1}{k}\right) ^{u-v} \left(\frac{1}{k}\right)^{n-u}
\end{align}
if $u\geq v$,
and otherwise, $m_{u,v} = 0$,
where $n$ is the number of channel uses and $u,v \in \{0,\ldots, n\}$. The coefficients are such that $c_0, c_1, \ldots, c_n\in [0,1]$ and $\sum_{j=0}^n \binom{n}{j} c_j=1$. 
We then have that
\begin{multline}
\inf_{\sigma'_{A_1B_1\cdots A_nB_n}\in \kex}D^{\varepsilon}_h(\rho_{A_1B_2\cdots A_nB_n}\Vert\sigma'_{A_1B_1\cdots A_nB_n})\\
\leq \min_{b_0,b_1,\ldots, b_n}D^{\varepsilon}_h(\left\{a_0,a_1,\ldots, a_n\right\}\Vert \left\{b_0,b_1,\ldots, b_n\right\}),\label{SDP_erasure}
\end{multline}
where the distribution $\{a_0,a_1,\ldots, a_n\}$ is induced by measuring the number of erasures in $\rho_{A_1B_2\cdots A_nB_n}$ and the coefficients  $\{b_0,b_1,\ldots, b_n\}$ are chosen as discussed above.
The inequality follows from restricting the form of the minimizing state. By exploiting the dual formulation of the hypothesis testing relative entropy \cite{DKFRR13}, we can now write the expression in \eqref{SDP_erasure} as the following linear program:
\begin{multline}
\min_{c_0,c_1,\ldots, c_n}D^{\varepsilon}_h\left(\left\{a_0,a_1,\ldots, a_n\right\}\Vert \left\{b_0,b_1,\ldots, b_n\right\}\right)
= \\-\log_2\!\left(\max_{\left\{c_0,c_1,\ldots, c_n\right\},\{\alpha_i\}_i,y} y (1-\varepsilon) -\sum_{i=0}^n\alpha_i \right),
\end{multline}
such that
\begin{align}
\forall i \in [0,n], &\quad \alpha_i -ya_i+b_i \geq 0,\\
& \quad b_i = \sum_{j=0}^n m_{i,j} c_j ,\\
&\quad 0 \leq c_i \leq 1,\\
&\quad y \geq 0, \ \alpha_i \geq 0,\\
&\quad\sum_{j=0}^n \binom{n}{j} c_j=1.
\end{align}
For the plots in Figures~\ref{fig:plots_erasure1} and \ref{fig:plots_erasure}, we have taken $\sigma_{A_1B_1A_2B_2\cdots A_nB_n}$ to be in a particular set of extendible states as defined above. Within this set, we have optimized over at most $k=10$ extendible states. 
\begin{figure}
\centering
\begin{tikzpicture}\label{fig-erasure1}
	\begin{axis}[
	scale = 1,
    ylabel=Number of qubits transmitted,
	xlabel=Number $n$ of channel uses,
    xtick={5,10,15,20,25},
	xmin = 2,
	xmax = 20,
	ymin = 0,
	ymax = 8,
	tick label style={/pgf/number format/fixed, /pgf/number format/precision=3},
	legend style = {at = {(0.05,0.8)},anchor = north west}, 
	legend cell align = left,]
	\addplot[thick,dashed,red] table[x=number1,y=raten,col sep=comma] {p.35_erasure_channel_final.txt};
	\addplot[thick,color=dgreen, dashdotted] table[x=number1,y=tbr_raten,col sep=comma] {p.35_erasure_channel_final.txt};
	
	\legend{\small{KDWW},\small{TBR}, \small{WFD}};
	\end{axis}
	\end{tikzpicture}
    %\captionsetup[figure]{font=small,labelfont=small}
\caption{Upper bounds on the number of qubits that can be reliably transmitted over an erasure channel with $p=0.35$, and $\varepsilon= 0.05$. The red dashed line is the bound from Theorem \ref{thm:k-ext-q-cap-bnd}. The green dash-dotted line is an upper bound from \cite{TBR15}. }
\label{fig:plots_erasure1}
\end{figure}

\begin{figure}
\begin{tikzpicture}
	\begin{axis}[
	scale = 1,
    ylabel=Number of qubits transmitted,
	xlabel=Number $n$ of channel uses,
    xtick={6,10,15,20,25},
	xmin = 3,
	xmax = 23,
	ymin = 0,
	ymax = 8,
	tick label style={/pgf/number format/fixed, /pgf/number format/precision=3},
	legend style = {at = {(0.05,0.80)},anchor = north west}, 
	legend cell align = left,]
	\addplot[thick,dashed,red] table[x=number1,y=raten,col sep=comma] {p.5_erasure_channel_final.txt};
	\addplot[thick,color=dgreen, dashdotted] table[x=number1,y=tbr_raten,col sep=comma] {p.5_erasure_channel_final.txt};

	\legend{\small{KDWW},\small{TBR}, \small{WFD}};
	\end{axis}
	\end{tikzpicture}
    %\captionsetup[figure]{font=small,labelfont=small}
\caption{Upper bounds on the number of qubits that can be reliably transmitted over an erasure channel with $p=0.49$, and $\varepsilon= 0.05$. The red dashed line is the bound from Theorem \ref{thm:k-ext-q-cap-bnd}. The green dash-dotted line is an upper bound from \cite{TBR15}.}%
\label{fig:plots_erasure}%
\end{figure}

\section{Pretty strong converse for antidegradable channels}

\label{sec:pretty-strong}

As a direct application of Theorem~\ref{theorem:adaptive-protocols}, we revisit the ``pretty strong converse'' of \cite{MW13} for antidegradable channels. A channel $\mathcal{N}_{A \to B}$ is antidegradable \cite{CG06,M10} if the output state $\mathcal{N}_{A \to B}(\rho_{RA})$ is two-extendible for every input state $\rho_{RA}$. Due to this property, antidegradable channels have zero asymptotic quantum capacity \cite{PhysRevLett.78.3217,Holevo2008}. Theorem~\ref{theorem:adaptive-protocols} implies the following bound for the non-asymptotic case:
\begin{corollary}
Fix $\varepsilon \in [0,1/2)$. The following bound holds for every $(n,M,\varepsilon)$  quantum communication protocol employing $n$ uses of an antidegradable channel $\mathcal{N}$ interleaved by two-extendible channels:
\begin{equation}
\frac{1}{n}\log_2 M
 \leq \frac{1}{n}  \log_2\!\left(\frac{1}{1-2\varepsilon}\right). %\label{eq:pretty-strong}    
\end{equation}
\end{corollary}

\begin{proof}
Let $\mathcal{N}_{A\rightarrow B}$ be an antidegradable channel, and suppose
that $\rho_{RA}$ is a state input to the channel. Then the output state
$\mathcal{N}_{A\rightarrow B}(\rho_{RA})$ is always a two-extendible state
(due to anti-degradability) \cite{M10}. As a direct consequence of Theorem~\ref{theorem:adaptive-protocols}, the following bound applies to every $(n,M,\varepsilon)$  quantum communication protocol employing $n$ uses of an antidegradable channel $\mathcal{N}$ interleaved by two-extendible channels:
\begin{equation}
-\frac{1}{n}\log_{2}\!\left[  \frac{1}{M}+\frac{1}{2}-\frac{1}{2M}\right]  \leq\frac{1}%
{n}\log_{2}\!\left(  \frac{1}{1-\varepsilon}\right)  .
\end{equation}
This follows by setting $k=2$ and noticing that $\sup_{\psi_{RA}} E^{\max}_{k}(R;B)_{\tau}= 0 $, where 
 $\tau_{RB}\coloneqq \mc{N}_{A\to B}(\psi_{RA})$, for such antidegradable channels.
 After some basic algebraic steps, for $\varepsilon < 1/2$, we
can rewrite this bound as%
\begin{equation}
\frac{1}{n}\log_{2}M\leq\frac{1}{n}\log_{2}\!\left[  \frac{1}{2\left(
1-\varepsilon\right)  -1}\right]  .
\end{equation}
These steps are as follows:%
\begin{align}
-\frac{1}{n}\log_{2}\!\left[  \frac{1}{M}+\frac{1}{2}-\frac{1}{2M}\right]    & \leq
\frac{1}{n}\log_{2}\!\left(  \frac{1}{1-\varepsilon}\right)  \nonumber\\
\Leftrightarrow\log_{2}\!\left[  \frac{2M}{M+1}\right]    & \leq\log_{2}\!\left(
\frac{1}{1-\varepsilon}\right)  \\
\Leftrightarrow\frac{2}{1+1/M}  & \leq\frac{1}{1-\varepsilon}\\
\Leftrightarrow2\left(  1-\varepsilon\right)    & \leq1+1/M\\
\Leftrightarrow2\left(  1-\varepsilon\right)  -1  & \leq1/M \\
\Leftrightarrow  1- 2 \varepsilon  & \leq1/M.
\end{align}
This concludes the proof.
\end{proof}

We conclude from the above inequality  that, for an antidegradable channel, there is a strong limitation on its ability to generate entanglement whenever the error parameter $\varepsilon < \tfrac12$, as is usually desired for applications in quantum computation. We also remark that the bound above is tighter than related bounds given in \cite{MW13}, and furthermore, the bound applies to quantum communication protocols assisted by interleaved two-extendible channels, which were not considered in~\cite{MW13}.

More generally, if the output of the channel is always a $k$-extendible state, then we have the following bound:
\begin{corollary}
\label{cor:pretty-strong-k}
Fix $\varepsilon \in [0,1-1/k)$. Let
$\mathcal{N}_{A \to B}$ be a $k$-extendible channel, in the sense that $\mathcal{N}_{A \to B}(\rho_{RA})$ is $k$-extendible for every input state $\rho_{RA}$.
Then the following bound holds for  every $(n,M,\varepsilon)$  quantum communication protocol employing $n$ uses of the channel $\mathcal{N}$ interleaved by $k$-extendible channels:
\begin{equation}
\frac{1}{n}\log_{2}M\leq\frac{1}{n}\log_{2}\!\left(\frac{1}{  1-\frac{k}{k-1}\varepsilon
}\right).
\end{equation}
\end{corollary}

\begin{proof}
This follows by the same reasoning as in the previous proof. If the output of the channel is $k$-extendible, then employing Theorem~\ref{theorem:adaptive-protocols} gives that%
\begin{equation}
-\frac{1}{n}\log_{2}\!\left[  \frac{1}{M}+\frac{1}{k}-\frac{1}{M k}\right]  \leq\frac{1}%
{n}\log_{2}\!\left(  \frac{1}{1-\varepsilon}\right)  .
\end{equation}
We then employ the following algebraic steps:
\begin{align}
-\frac{1}{n}\log_{2}\!\left[  \frac{1}{M}+\frac{1}{k}-\frac{1}{M k}\right]    & \leq
\frac{1}{n}\log_{2}\!\left(  \frac{1}{1-\varepsilon}\right)  \\
-\frac{1}{n}\log_{2}\!\left[  \frac{k-1+M}{kM}\right]    & \leq\frac{1}{n}%
\log_{2}\!\left(  \frac{1}{1-\varepsilon}\right)  \\
\frac{kM}{k-1+M}  & \leq\frac{1}{1-\varepsilon}\\
\frac{k}{\left(  k-1\right)  /M+1}  & \leq\frac{1}{1-\varepsilon}\\
k\left(  1-\varepsilon\right)    & \leq\left(  k-1\right)  /M+1\\
\left[  \frac{k\left(  1-\varepsilon\right)  -1}{k-1}\right]    & \leq1/M\\
  1 - \frac{k}{k-1}\varepsilon    & \leq1/M .
\end{align}
We then get that%
\begin{equation}
\frac{1}{n}\log_{2}M\leq
\frac{1}{n}\log_{2}\!\left(\frac{1}{  1-\frac{k}{k-1}\varepsilon
}\right).
\end{equation}
This concludes the proof.
\end{proof}

\bigskip
Thus, for a fixed $\varepsilon\in\left[  0,1-1/k\right)  $, we conclude that the
rate of quantum communication for a single-sender single-receiver $k$-extendible channel decays to zero as $n\rightarrow\infty$.
Related, if the communication rate for a sequence of codes used over such a channel is strictly greater than zero, then it must be the case that the
error in communication is greater than or equal to $1-1/k$, which is a higher jump than discussed in the previous case. An example of a channel for which this effect occurs is a quantum erasure channel with erasure probability $1-1/k$.

Another example of a channel for which the bound in Corollary~\ref{cor:pretty-strong-k} holds is the universal cloning machine channel (a $1 \to k$ universal quantum cloner followed by a partial trace over $k-1$ of the clones) \cite{RevModPhys.77.1225}. When the dimension of the channel input is $M$, the bound in Corollary~\ref{cor:pretty-strong-k} is in fact saturated, as observed in the proof of \cite[Theorem~III.8]{JV13}.

\section{Conclusions}

\label{sec:conclusion}

In this paper, we obtained tight non-asymptotic bounds on the rates of entanglement transmission of a channel assisted by a $k$-extendible channel. To obtain these tight bounds, we developed the resource theory of unextendibility. The free states in this resource theory are $k$-extendible states, which have been studied previously for quantifying the entanglement present in a quantum state. We define $k$-extendible channels, and prove that these are free channels in the resource theory of $k$-unextendibility.  We then obtain non-asymptotic upper bounds on the rate at which qubits can be transmitted over a finite number of uses of a given quantum channel, by utilizing the monotones introduced for the resource theory of unextendibility. We show that these bounds are significantly tighter than those in \cite{TBR15,WFD17} for depolarizing and erasure channels.  

An interesting research direction would be to further explore the resource theory of unextendibility. One plausible direction would be to use this resource theory to obtain non-asymptotic converse bounds on the entanglement distillation rate of bipartite quantum interactions and compare with the bounds obtained in \cite{DBW17}. Another direction is to analyze the bounds in Theorem~\ref{thm:k-ext-q-cap-bnd} for other noise models that are practically relevant. Finally, it remains open to link the bounds developed here with the open problem of finding a strong converse for the quantum capacity of degradable channels \cite{MW13}. To solve that problem, recall that one contribution of \cite{MW13} was to reduce the question of the strong converse of degradable channels to that of establishing the strong converse for symmetric channels.

\begin{acknowledgments}
We thank Sumeet Khatri, Vishal Katariya, Felix Leditzky, and Stefano Mancini for insightful discussions. SD acknowledges support from the LSU Graduate School Economic Development Assistantship. EK and MMW acknowledge support from the US Office of Naval Research and the National Science Foundation under Grant No.~1350397. Andreas Winter acknowledges support from the ERC
Advanced Grant IRQUAT, the Spanish MINECO, projects FIS2013-40627-P and FIS2016-86681-P, with
the support of FEDER funds, and the Generalitat de Catalunya, CIRIT project 2014-SGR-966.
\end{acknowledgments}

\iffalse
\begin{figure}
\centering
\begin{tikzpicture}\label{fig-erasure1}
	\begin{axis}[
	scale = 1.3,
    ylabel=Number of qubits transmitted,
	xlabel=Number of channel uses ~$(n)$,
    xtick={5,10,15,20,25},
	xmin = 2,
	xmax = 20,
	ymin = 0,
	ymax = 8,
	tick label style={/pgf/number format/fixed, /pgf/number format/precision=3},
	legend style = {at = {(0.05,0.8)},anchor = north west}, 
	legend cell align = left,]
	\addplot[thick,dashed,red] table[x=number1,y=raten,col sep=comma] {p.35_erasure_channel_final.txt};
	\addplot[thick,color=dgreen, dashdotted] table[x=number1,y=tbr_raten,col sep=comma] {p.35_erasure_channel_final.txt};
	
	\legend{\small{KDWW},\small{TBR}, \small{WFD}};
	\fi

\bibliographystyle{alpha}
\bibliography{kex-prl}

\newcommand{\etalchar}[1]{$^{#1}$}
\begin{thebibliography}{BDGDMW17}

\bibitem[BBFS18]{BBFS18}
Mario Berta, Francesco Borderi, Omar Fawzi, and Volkher Scholz.
\newblock Semidefinite programming hierarchies for quantum error correction.
\newblock October 2018.
\newblock arXiv:1810.12197.

\bibitem[BD10]{BD10}
Francesco Buscemi and Nilanjana Datta.
\newblock The quantum capacity of channels with arbitrarily correlated noise.
\newblock {\em IEEE Transactions on Information Theory}, 56(3):1447--1460,
  March 2010.
\newblock arXiv:0902.0158.

\bibitem[BDF{\etalchar{+}}99]{BDFMRSSW99}
Charles~H. Bennett, David~P. DiVincenzo, Christopher~A. Fuchs, Tal Mor, Eric
  Rains, Peter~W. Shor, John~A. Smolin, and William~K. Wootters.
\newblock Quantum nonlocality without entanglement.
\newblock {\em Physical Review A}, 59(2):1070--1091, February 1999.
\newblock arXiv:quant-ph/9804053.

\bibitem[BDGDMW17]{DDMW17}
Khaled Ben~Dana, Mar\'{\i}a Garc\'{\i}a~D\'{\i}az, Mohamed Mejatty, and Andreas
  Winter.
\newblock Resource theory of coherence: Beyond states.
\newblock {\em Physical Review A}, 95(6):062327, June 2017.
\newblock arXiv:1704.03710.

\bibitem[BDS97]{PhysRevLett.78.3217}
Charles~H. Bennett, David~P. DiVincenzo, and John~A. Smolin.
\newblock Capacities of quantum erasure channels.
\newblock {\em Physical Review Letters}, 78(16):3217--3220, April 1997.
\newblock arXiv:quant-ph/9701015.

\bibitem[BDSW96]{BDSW96}
Charles~H. Bennett, David~P. DiVincenzo, John~A. Smolin, and William~K.
  Wootters.
\newblock Mixed-state entanglement and quantum error correction.
\newblock {\em Physical Review A}, 54(5):3824--3851, November 1996.
\newblock arXiv:quant-ph/9604024.

\bibitem[Bei13]{Bei13}
Salman Beigi.
\newblock Sandwiched {R\'enyi} divergence satisfies data processing inequality.
\newblock {\em Journal of Mathematical Physics}, 54(12):122202, December 2013.
\newblock arXiv:1306.5920.

\bibitem[BG15]{BG15}
Fernando G. S.~L. Brand{\~{a}}o and Gilad Gour.
\newblock Reversible framework for quantum resource theories.
\newblock {\em Physical Review Letters}, 115(7):070503, August 2015.
\newblock arXiv:1502.03149.

\bibitem[BH17]{BH17}
Fernando~{G.~S.~L.} Brand{\~a}o and Aram~W. Harrow.
\newblock Quantum de {Finetti} theorems under local measurements with
  applications.
\newblock {\em Communications in Mathematical Physics}, 353(2):469--506, July
  2017.
\newblock arXiv:1210.6367.

\bibitem[Bha97]{Bha97}
Rajendra Bhatia.
\newblock {\em Matrix Analysis}.
\newblock Springer New York, 1997.

\bibitem[BHLS03]{BHLS03}
Charles~H. Bennett, Aram~W. Harrow, Debbie~W. Leung, and John~A. Smolin.
\newblock On the capacities of bipartite {Hamiltonians} and unitary gates.
\newblock {\em IEEE Transactions on Information Theory}, 49(8):1895--1911,
  August 2003.
\newblock arXiv:quant-ph/0205057.

\bibitem[BKN00]{BKN98}
Howard Barnum, Emanuel Knill, and Michael~A. Nielsen.
\newblock On quantum fidelities and channel capacities.
\newblock {\em IEEE Transactions on Information Theory}, 46(4):1317, July 2000.
\newblock arXiv:quant-ph/9809010.

\bibitem[BW18]{BW17}
Mario Berta and Mark~M. Wilde.
\newblock Amortization does not enhance the max-{Rains} information of a
  quantum channel.
\newblock {\em New Journal of Physics}, 20:053044, May 2018.
\newblock arXiv:1709.04907.

\bibitem[CDP09]{CDP09}
Giulio Chiribella, Giacomo~Mauro D'Ariano, and Paolo Perinotti.
\newblock Realization schemes for quantum instruments in finite dimensions.
\newblock {\em Journal of Mathematical Physics}, 50(4):042101, April 2009.
\newblock arXiv:0810.3211.

\bibitem[CG06]{CG06}
Filippo Caruso and Vittorio Giovannetti.
\newblock Degradability of bosonic {Gaussian} channels.
\newblock {\em Physical Review A}, 74(6):062307, December 2006.
\newblock arXiv:quant-ph/0603257.

\bibitem[CG19]{CG19}
Eric Chitambar and Gilad Gour.
\newblock Quantum resource theories.
\newblock {\em Reviews of Modern Physics}, 91(2):025001, April 2019.
\newblock arXiv:1806.06107.

\bibitem[CJYZ16]{CJYZ16}
Jianxin Chen, Zhengfeng Ji, Nengkun Yu, and Bei Zeng.
\newblock Detecting consistency of overlapping quantum marginals by
  separability.
\newblock {\em Physical Review A}, 93(3):032105, Mar 2016.

\bibitem[CKMR07]{CKMR07}
Matthias Christandl, Robert K{\"o}nig, Graeme Mitchison, and Renato Renner.
\newblock One-and-a-half quantum de {Finetti} theorems.
\newblock {\em Communications in Mathematical Physics}, 273(2):473--498, July
  2007.
\newblock arXiv:quant-ph/0602130.

\bibitem[CLM{\etalchar{+}}14]{CLM+14}
Eric Chitambar, Debbie Leung, Laura Man{\v{c}}inska, Maris Ozols, and Andreas
  Winter.
\newblock Everything you always wanted to know about {LOCC} (but were afraid to
  ask).
\newblock {\em Communications in Mathematical Physics}, 328(1):303--326, May
  2014.
\newblock arXiv:1210.4583.

\bibitem[CMH17]{CM17}
Matthias Christandl and Alexander M\"{u}ller-Hermes.
\newblock Relative entropy bounds on quantum, private and repeater capacities.
\newblock {\em Communications in Mathematical Physics}, 353(2):821--852, July
  2017.
\newblock arXiv:1604.03448.

\bibitem[CMW16]{CMW14}
Tom Cooney, Milan Mosonyi, and Mark~M. Wilde.
\newblock Strong converse exponents for a quantum channel discrimination
  problem and quantum-feedback-assisted communication.
\newblock {\em Communications in Mathematical Physics}, 344(3):797--829, June
  2016.
\newblock arXiv:1408.3373.

\bibitem[Dat09a]{Dat09}
Nilanjana Datta.
\newblock Max-relative entropy of entanglement, alias log robustness.
\newblock {\em International Journal of Quantum Information}, 7(02):475--491,
  January 2009.
\newblock arXiv:0807.2536.

\bibitem[Dat09b]{D09}
Nilanjana Datta.
\newblock Min- and max-relative entropies and a new entanglement monotone.
\newblock {\em IEEE Transactions on Information Theory}, 55(6):2816--2826, June
  2009.
\newblock arXiv:0803.2770.

\bibitem[DBW20]{DBW17}
Siddhartha Das, Stefan B{\"a}uml, and Mark~M. Wilde.
\newblock Entanglement and secret-key-agreement capacities of bipartite quantum
  interactions and read-only memory devices.
\newblock {\em Physical Review A}, 101(1):012344, January 2020.
\newblock arXiv:1712.00827.

\bibitem[DKF{\etalchar{+}}13]{DKFRR13}
Frederic Dupuis, Lea Kraemer, Philippe Faist, Joseph~M. Renes, and Renato
  Renner.
\newblock Generalized entropies.
\newblock {\em XVIIth International Congress on Mathematical Physics}, pages
  134--153, 2013.
\newblock arXiv:1211.3141.

\bibitem[DLT02]{DLT02}
David~P. DiVincenzo, Debbie~W. Leung, and Barbara~M. Terhal.
\newblock Quantum data hiding.
\newblock {\em IEEE Transactions on Information Theory}, 48(3):580--598, March
  2002.
\newblock arXiv:quant-ph/0103098.

\bibitem[DPS02]{DPS02}
Andrew~C. Doherty, Pablo~A. Parrilo, and Federico~M. Spedalieri.
\newblock Distinguishing separable and entangled states.
\newblock {\em Physical Review Letters}, 88(18):187904, April 2002.
\newblock arXiv:quant-ph/0112007.

\bibitem[DPS04]{DPS04}
Andrew~C. Doherty, Pablo~A. Parrilo, and Federico~M. Spedalieri.
\newblock Complete family of separability criteria.
\newblock {\em Physical Review A}, 69(2):022308, February 2004.
\newblock arXiv:quant-ph/0308032.

\bibitem[DTW16]{DTW14}
Nilanjana Datta, Marco Tomamichel, and Mark~M. Wilde.
\newblock On the second-order asymptotics for entanglement-assisted
  communication.
\newblock {\em Quantum Information Processing}, 15(6):2569--2591, June 2016.
\newblock arXiv:1405.1797.

\bibitem[DW16]{DW16}
Runyao Duan and Andreas Winter.
\newblock No-signalling-assisted zero-error capacity of quantum channels and an
  information theoretic interpretation of the {Lov{\'{a}}sz} number.
\newblock {\em {IEEE} Transactions on Information Theory}, 62(2):891--914,
  February 2016.
\newblock arXiv:1409.3426.

\bibitem[FL13]{FL13}
Rupert~L. Frank and Elliott~H. Lieb.
\newblock Monotonicity of a relative {R\'enyi} entropy.
\newblock {\em Journal of Mathematical Physics}, 54(12):122201, December 2013.
\newblock arXiv:1306.5358.

\bibitem[GBP97]{GBP97}
Markus Grassl, Thomas Beth, and Thomas Pellizzari.
\newblock Codes for the quantum erasure channel.
\newblock {\em Physical Review A}, 56(1):33, July 1997.
\newblock arXiv:quant-ph/9610042.

\bibitem[Gha10]{Gharibian10}
Sevag Gharibian.
\newblock Strong {NP}-hardness of the quantum separability problem.
\newblock {\em Quantum Information {\&} Computation}, 10(3{\&}4):343--360,
  2010.

\bibitem[Gur03]{G03}
Leonid Gurvits.
\newblock Classical deterministic complexity of {Edmonds'} problem and quantum
  entanglement.
\newblock In {\em Proceedings of the Thirty-Fifth Annual ACM Symposium on
  Theory of Computing}, pages 10--19, San Diego, California, USA, June 2003.
\newblock arXiv:quant-ph/0303055.

\bibitem[HH99]{HH99}
Micha\l{} Horodecki and Pawe\l{} Horodecki.
\newblock Reduction criterion of separability and limits for a class of
  distillation protocols.
\newblock {\em Physical Review A}, 59(6):4206--4216, June 1999.
\newblock arXiv:quant-ph/9708015.

\bibitem[HHH99]{HHH99}
Micha\l{} Horodecki, Pawe\l{} Horodecki, and Ryszard Horodecki.
\newblock General teleportation channel, singlet fraction, and
  quasidistillation.
\newblock {\em Physical Review A}, 60(3):1888--1898, September 1999.
\newblock arXiv:quant-ph/9807091.

\bibitem[HHHH09]{HHHH09}
Ryszard Horodecki, Pawe\l{} Horodecki, Micha\l{} Horodecki, and Karol
  Horodecki.
\newblock Quantum entanglement.
\newblock {\em Review of Modern Physics}, 81(2):865--942, June 2009.
\newblock arXiv:quant-ph/0702225.

\bibitem[Hol02]{Hol02}
Alexander~S. Holevo.
\newblock Remarks on the classical capacity of quantum channel.
\newblock December 2002.
\newblock quant-ph/0212025.

\bibitem[Hol07]{Hol07}
Alexander~S. Holevo.
\newblock Complementary channels and the additivity problem.
\newblock {\em Theory of Probability \& Its Applications}, 51(1):92--100, 2007.

\bibitem[Hol08]{Holevo2008}
Alexander~S. Holevo.
\newblock Entanglement-breaking channels in infinite dimensions.
\newblock {\em Problems of Information Transmission}, 44(3):171--184, September
  2008.
\newblock arXiv:0802.0235.

\bibitem[Hol13]{H13book}
Alexander~S. Holevo.
\newblock {\em Quantum systems, channels, information: A mathematical
  introduction}, volume~16.
\newblock Walter de Gruyter, 2013.

\bibitem[JV13]{JV13}
Peter~D. Johnson and Lorenza Viola.
\newblock Compatible quantum correlations: Extension problems for {Werner} and
  isotropic states.
\newblock {\em Physical Review A}, 88(3):032323, September 2013.
\newblock arXiv:1305.1342.

\bibitem[KDWW19]{KDWW19}
Eneet Kaur, Siddhartha Das, Mark~M. Wilde, and Andreas Winter.
\newblock Extendibility limits the performance of quantum processors.
\newblock {\em Physical Review Letters}, 123(7):070502, August 2019.
\newblock arXiv:1803.10710.

\bibitem[Kit97]{Kit97}
Alexei Kitaev.
\newblock Quantum computations: algorithms and error correction.
\newblock {\em Russian Mathematical Surveys}, 52(6):1191--1249, December 1997.

\bibitem[KW18]{KW17a}
Eneet Kaur and Mark~M. Wilde.
\newblock Amortized entanglement of a quantum channel and approximately
  teleportation-simulable channels.
\newblock {\em Journal of Physics A}, 51(3):035303, January 2018.
\newblock arXiv:1707.07721.

\bibitem[LHL03]{LHL03}
Mathew~S. Leifer, Leah Henderson, and Noah Linden.
\newblock Optimal entanglement generation from quantum operations.
\newblock {\em Physical Review A}, 67(1):012306, January 2003.
\newblock arXiv:quant-ph/0205055.

\bibitem[MH12]{Mul12}
Alexander M{\"u}ller-Hermes.
\newblock Transposition in quantum information theory.
\newblock Master's thesis, Technical University of Munich, September 2012.

\bibitem[MHR17]{MR15}
Alexander Mueller-Hermes and David Reeb.
\newblock Monotonicity of the quantum relative entropy under positive maps.
\newblock {\em Annales Henri Poincar{\'{e}}}, 18(5):1777--1788, January 2017.
\newblock arXiv:1512.06117.

\bibitem[MLDS{\etalchar{+}}13]{MDSFT13}
Martin {M\"uller}-Lennert, Fr\'ed\'eric Dupuis, Oleg Szehr, Serge Fehr, and
  Marco Tomamichel.
\newblock On quantum {R\'enyi} entropies: a new definition and some properties.
\newblock {\em Journal of Mathematical Physics}, 54(12):122203, December 2013.
\newblock arXiv:1306.3142.

\bibitem[MW14]{MW13}
Ciara Morgan and Andreas Winter.
\newblock {``Pretty strong''} converse for the quantum capacity of degradable
  channels.
\newblock {\em IEEE Transactions on Information Theory}, 60(1):317--333,
  January 2014.
\newblock arXiv:1301.4927.

\bibitem[Myh10]{M10}
Geir~Ove Myhr.
\newblock {\em Symmetric extension of bipartite quantum states and its use in
  quantum key distribution with two-way postprocessing}.
\newblock PhD thesis, Universit\"at Erlangen-N\"urnberg, 2010.
\newblock arXiv:1103.0766.

\bibitem[OP93]{OP93}
Masanori Ohya and Denes Petz.
\newblock {\em Quantum Entropy and Its Use}.
\newblock Springer, 1993.

\bibitem[PBHS13]{PBHS13}
Lukasz Pankowski, Fernando~G.S.L. Brandao, Michal Horodecki, and Graeme Smith.
\newblock Entanglement distillation by extendible maps.
\newblock {\em Quantum Information and Computation}, 13(9--10):751--770,
  September 2013.
\newblock arXiv:1109.1779.

\bibitem[PHHH06]{PHH2006}
Marco Piani, Michal Horodecki, Pawel Horodecki, and Ryszard Horodecki.
\newblock Properties of quantum nonsignaling boxes.
\newblock {\em Physical Review A}, 74(1):012305, July 2006.
\newblock arXiv:quant-ph/0505110.

\bibitem[PLOB17]{PLOB17}
Stefano Pirandola, Riccardo Laurenza, Carlo Ottaviani, and Leonardo Banchi.
\newblock Fundamental limits of repeaterless quantum communications.
\newblock {\em Nature Communications}, 8:15043, February 2017.

\bibitem[PPV10]{PPV10}
Yury Polyanskiy, H.~Vincent Poor, and Sergio Verd\'{u}.
\newblock Channel coding rate in the finite blocklength regime.
\newblock {\em IEEE Transactions on Information Theory}, 56(5):2307--2359, May
  2010.

\bibitem[PV10]{PV10}
Yury Polyanskiy and Sergio Verd\'u.
\newblock Arimoto channel coding converse and {R\'enyi} divergence.
\newblock In {\em Proceedings of the 48th Annual Allerton Conference on
  Communication, Control, and Computation}, pages 1327--1333, September 2010.

\bibitem[Rai99]{Rai99}
Eric~M. Rains.
\newblock Bound on distillable entanglement.
\newblock {\em Physical Review A}, 60(1):179--184, July 1999.
\newblock arXiv:quant-ph/9809082.

\bibitem[Rai01]{Rai01}
Eric~M. Rains.
\newblock A semidefinite program for distillable entanglement.
\newblock {\em IEEE Transactions on Information Theory}, 47(7):2921--2933,
  November 2001.
\newblock arXiv:quant-ph/0008047.

\bibitem[RKB{\etalchar{+}}18]{RKB+17}
Luca Rigovacca, Go~Kato, Stefan Baeuml, M.~S. Kim, W.~J. Munro, and Koji Azuma.
\newblock Versatile relative entropy bounds for quantum networks.
\newblock {\em New Journal of Physics}, 20:013033, January 2018.
\newblock arXiv:1707.05543.

\bibitem[RST{\etalchar{+}}18]{RSTEDW18}
Filip {Rozpedek}, Thomas Schiet, Le~Phuc Thinh, David Elkouss, Andrew~C.
  Doherty, and Stephanie Wehner.
\newblock Optimizing practical entanglement distillation.
\newblock {\em Physical Review A}, 97(6):062333, June 2018.
\newblock arXiv:1803.10111.

\bibitem[SIGA05]{RevModPhys.77.1225}
Valerio Scarani, Sofyan Iblisdir, Nicolas Gisin, and Antonio Ac\'{\i}n.
\newblock Quantum cloning.
\newblock {\em Reviews of Modern Physics}, 77(4):1225--1256, November 2005.
\newblock arXiv:quant-ph/0511088.

\bibitem[SW12]{SW12}
Naresh Sharma and Naqueeb~Ahmad Warsi.
\newblock On the strong converses for the quantum channel capacity theorems.
\newblock 2012.
\newblock arXiv:1205.1712.

\bibitem[TBR16]{TBR15}
Marco Tomamichel, Mario Berta, and Joseph~M. Renes.
\newblock Quantum coding with finite resources.
\newblock {\em Nature Communications}, 7:11419, May 2016.
\newblock arXiv:1504.04617.

\bibitem[TGW14]{TGW14}
Masahiro Takeoka, Saikat Guha, and Mark~M. Wilde.
\newblock The squashed entanglement of a quantum channel.
\newblock {\em IEEE Transactions on Information Theory}, 60(8):4987--4998,
  August 2014.
\newblock arXiv:1310.0129.

\bibitem[TSW17]{PhysRevLett.119.150501}
Masahiro Takeoka, Kaushik~P. Seshadreesan, and Mark~M. Wilde.
\newblock Unconstrained capacities of quantum key distribution and entanglement
  distillation for pure-loss bosonic broadcast channels.
\newblock {\em Physical Review Letters}, 119(15):150501, October 2017.
\newblock arXiv:1706.06746.

\bibitem[TWW17]{TWW17}
Marco Tomamichel, Mark~M. Wilde, and Andreas Winter.
\newblock Strong converse rates for quantum communication.
\newblock {\em {IEEE} Transactions on Information Theory}, 63(1):715--727,
  January 2017.
\newblock arXiv:1406.2946.

\bibitem[Uhl76]{U76}
Armin Uhlmann.
\newblock The ``transition probability'' in the state space of a *-algebra.
\newblock {\em Reports on Mathematical Physics}, 9(2):273--279, April 1976.

\bibitem[Ume62]{Ume62}
Hisaharu Umegaki.
\newblock Conditional expectations in an operator algebra, {IV} (entropy and
  information).
\newblock {\em Kodai Mathematical Seminar Reports}, 14(2):59--85, June 1962.

\bibitem[vLB00]{vLB00}
Peter van Loock and Samuel~L. Braunstein.
\newblock Multipartite entanglement for continuous variables: A quantum
  teleportation network.
\newblock {\em Physical Review Letters}, 84(15):3482--3485, April 2000.
\newblock arXiv:quant-ph/9906021.

\bibitem[vN32]{Neu32}
Johann von Neumann.
\newblock {\em Mathematische grundlagen der quantenmechanik}.
\newblock Verlag von Julius Springer Berlin, 1932.

\bibitem[VP98]{VP98}
Vlatko Vedral and Martin~B. Plenio.
\newblock Entanglement measures and purification procedures.
\newblock {\em Physical Review A}, 57(3):1619, 1998.

\bibitem[Wer89a]{W89a}
Reinhard~F. Werner.
\newblock An application of {Bell's} inequalities to a quantum state extension
  problem.
\newblock {\em Letters in Mathematical Physics}, 17(4):359--363, May 1989.

\bibitem[Wer89b]{Wer89}
Reinhard~F. Werner.
\newblock Quantum states with {Einstein-Podolsky-Rosen} correlations admitting
  a hidden-variable model.
\newblock {\em Physical Review A}, 40(8):4277--4281, October 1989.

\bibitem[Wer01]{Wer01}
Reinhard~F. Werner.
\newblock All teleportation and dense coding schemes.
\newblock {\em Journal of Physics A: Mathematical and General}, 34(35):7081,
  August 2001.
\newblock arXiv:quant-ph/0003070.

\bibitem[WF14]{WF14}
Joel~J. Wallman and Steven~T. Flammia.
\newblock Randomized benchmarking with confidence.
\newblock {\em New Journal of Physics}, 16(10):103032, October 2014.
\newblock arXiv:1404.6025.

\bibitem[WFD19]{WFD17}
Xin Wang, Kun Fang, and Runyao Duan.
\newblock Semidefinite programming converse bounds for quantum communication.
\newblock {\em IEEE Transactions on Information Theory}, 65(4):2583--2592,
  April 2019.
\newblock arXiv:1709.00200.

\bibitem[Win16]{Win16}
Andreas Winter.
\newblock Tight uniform continuity bounds for quantum entropies: conditional
  entropy, relative entropy distance and energy constraints.
\newblock {\em Communications in Mathematical Physics}, 347(1):291--313,
  October 2016.
\newblock arXiv:1507.07775.

\bibitem[WR12]{WR12}
Ligong Wang and Renato Renner.
\newblock One-shot classical-quantum capacity and hypothesis testing.
\newblock {\em Physical Review Letters}, 108(20):200501, May 2012.
\newblock arXiv:1007.5456.

\bibitem[WTB17]{WTB17}
Mark~M. Wilde, Marco Tomamichel, and Mario Berta.
\newblock Converse bounds for private communication over quantum channels.
\newblock {\em IEEE Transactions on Information Theory}, 63(3):1792--1817,
  March 2017.
\newblock arXiv:1602.08898.

\bibitem[WWY14]{WWY14}
Mark~M. Wilde, Andreas Winter, and Dong Yang.
\newblock Strong converse for the classical capacity of entanglement-breaking
  and {Hadamard} channels via a sandwiched {R\'enyi} relative entropy.
\newblock {\em Communications in Mathematical Physics}, 331(2):593--622,
  October 2014.
\newblock arXiv:1306.1586.

\end{thebibliography}

\vspace{0.5in}
\appendix
\section{Class of $k$-extendible channels}\label{sec:subclass-channels}

\iffalse
 Before stating , recall that the diamond norm of the difference of two channels $\mc{N}$ and $\mc{M}$ is given by
\begin{equation} 
\left \Vert \mc{N} - \mc{M} \right \Vert_{\diamond} \coloneqq  \max_{\psi_{RA}}
\left \Vert \id_R \otimes (\mc{N} - \mc{M})(\psi_{RA})\right \Vert_1,
\end{equation} 
where the optimization is with respect to pure-state inputs $\psi_{RA}$, with $R$ a reference system isomorphic to the channel input system $A$.
\fi
Before stating the proposition, we state an alternate representation of $1$W-LOCC channels, which is of relavance in the proof. $1$W-LOCC channels can also be represented as 
% \begin{equation}
% \mc{D}_{C'B\to B'}  \circ\Tr_{C'^{k-1}}\circ\mc{P}_{\bar{C}\to C_1'C_2'\cdots C_k'} \circ \mc{M}_{C\to \bar{C}}\circ\mc{E}_{A\to A'C},
% \end{equation}
\begin{equation}\label{eq:1W-LOCC}
\mc{D}_{C'B\to B'}\circ\mc{P}_{\bar{C}\to C'} \circ \mc{M}_{C\to \bar{C}}\circ\mc{E}_{A\to A'C},
\end{equation}
where $\mc{E}_{A\to A'C}$ is an arbitrary channel, $\mc{M}_{C\to \bar{C}}$ is a measurement channel, $\mc{P}_{\bar{C}\to C'}$ is a preparation channel, such that $\bar{C}$ is a classical system, and $\mc{D}_{C'B\to B'}$ is an arbitrary channel. %A measurement channel followed by a preparation channel realizes an entanglement breaking (EB) channel \cite{HSR03}.

\begin{proposition}
The diamond distance of the channel $\mathcal{K}_{AB\rightarrow A'B'}^k$ in \eqref{eq:subclass_def} to a $1$W-LOCC channel is bounded from above as 
\begin{multline}
\inf_{\mc{L}_{AB \to A'B'}\in1\mathrm{W-LOCC}}\left\|\mathcal{K}_{AB\rightarrow A'B'}^k - \mc{L}_{AB \to A'B'}\right\|_{\diamond}\\\leq |C|\frac{2|C|^2}{|C|^2+k},
\end{multline}
where $|C|=|ABA'B'|$, and 1W-LOCC denotes the set of all 1W-LOCC channels acting on input systems $AB$ and with output systems $A'B'$.
\end{proposition}
\begin{proof}

Letting $\mc{S}^k_{C\to C_1'C_2'\cdots C_k'}$ denote an extension channel for $\mc{A}^k_{C\to C'}$,  observe that
\begin{align}
{}&\inf_{\mc{L}_{AB \to A'B'}\in1\mathrm{W-LOCC}}\left\|\mathcal{K}_{AB\rightarrow A'B'}^k\right. -\left. \mc{L}_{AB \to A'B'}\right\|_{\diamond} \notag \\
\begin{split}
    {}&\leq {}\inf_{\mc{P}\circ\mc{M}} \left\|\Tr_{C^{k-1}}\circ\mc{S}^k_{C\to C_1'C_2'\cdots C_k'} \circ\mc{E}_{A\to A'C}\right.\\&\qquad\left.-\mc{P}_{\bar{C}\to C'} \circ \mc{M}_{C\to \bar{C}}\circ\mc{E}_{A\to A'C}\right\|_{\diamond}
\end{split}\\
\begin{split}
   {}&=\inf_{\mc{P}\circ\mc{M}}\max_{\psi_{RA}}\left\|\Tr_{C^{k-1}}\circ\mc{S}^k_{C\to C_1'C_2'\cdots C_k'} \circ\mc{E}_{A\to A'C}(\psi_{RA})\right.\\&\qquad\left.-\mc{P}_{\bar{C}\to C'} \circ \mc{M}_{C\to \bar{C}}\circ\mc{E}_{A\to A'C}\(\psi_{RA}\)\right\|_1 
\end{split}\\
\quad{}&\leq \inf_{\mc{P}\circ\mc{M}}\norm{\Tr_{C^{k-1}}\circ\mc{S}^k_{C\to C_1'C_2'\cdots C_k'} - \mc{P}_{\bar{C}\to C'} \circ \mc{M}_{C\to \bar{C}}}_{\diamond}.
\end{align}
The first inequality follows from \eqref{eq:subclass_def}, by choosing a particular $1$W-LOCC and from the monotonicity of trace norm with respect to quantum channels. The first equality follows from the definition of diamond distance. The second inequality follows from the definition of diamond distance, which has an implicit maximization over all the input states. 
We now observe that 
\begin{align}
&\inf_{\mc{P}\circ\mc{M}}\left\|\Tr_{C'^{k-1}}\circ\mc{S}^k_{C\to C_1'C_2'\cdots C_k'} - \mc{P}_{\bar{C}\to C'} \circ \mc{M}_{C\to \bar{C}}\right\|_{\diamond} \notag \\
&\quad\leq |C| \inf_{\Gamma^{EB}_{R'C'}}\left\|\Gamma^{k,\mathcal{S}}_{R'C'}/|C|-\Gamma^{EB}_{R'C'}/|C|\right\|_1\\
&\quad\leq |C|\frac{2|C'|^2}{|C'|^2+k},
\end{align}
where
\begin{align}
\Gamma^{k,\mathcal{S}}_{R'C'}/|C| & = \Tr_{C'^{k-1}}\circ\mc{S}^k_{C\to C_1'C_2'\cdots C_k'}\(\Phi_{RC}\)\\
&\in\kex(R\!:\!C'), \\
\Gamma^{EB}_{R'C'}/|C| & =\mc{P}_{\bar{C}\to C'}\circ\mc{M}_{C\to\bar{C}}\(\Phi_{RC}\)\\
&\in\SEP(R\!:\!C').
\end{align}
The first inequality follows from bounding the diamond distance between the two channels by the trace norm between the corresponding Choi operators (see, e.g., \cite[Lemma~7]{WF14}). The last inequality follows from \cite[Eq.~(11)]{CJYZ16}, which in turn built on the developments in \cite{CKMR07}.  %since a channel is entanglement breaking if and only if its Choi state is seprabale. 
\end{proof}

%--------------------------------------------------------------
%--------------------------------------------------------------

\section{Amortization does not enhance the max-$k$-unextendibility of a channel}

\label{app:amort-collapse}

The amortized entanglement $E_A(\mc{N})$ of a channel $\mc{N}_{A\to B}$ is defined as the following optimization~\cite{KW17a} (see also \cite{LHL03,BHLS03,CM17,DDMW17,RKB+17}):
\begin{equation}\label{eq:ent-arm}
E_A(\mc{N})\coloneqq \sup_{\rho_{R_AAR_B}} \left[E(R_A;BR_B)_{\tau}-E(R_AA;R_B)_{\rho}\right],
\end{equation}
where $E$ is an entanglement measure, $\tau_{R_ABR_B}=\mc{N}_{A\to B}(\rho_{R_AAR_B})$ for a state $\rho_{R_AAR_B}$ and $R_A, R_B$ are reference systems associated with the systems $A,B$, respectively. The supremum is with respect to all input states $\rho_{R_AAR_B}$ and the systems $R_A,R_B$ are finite-dimensional but could be arbitrarily large. Thus, in general, $E_A(\mc{N})$ need not be computable. The amortized entanglement quantifies the net amount of entanglement that can be generated by using the channel $\mc{N}_{A\to B}$, if the sender and the receiver are allowed to begin with some initial entanglement in the form of the state $\rho_{R_AAR_B}$. That is, $E(R_AA;R_B)_\rho$ quantifies the entanglement of the initial state $\rho_{R_AAR_B}$, and $E(R_A;BR_B)_{\tau}$ quantifies the entanglement of the final state produced after the action of the channel. 

The purpose of this appendix is to prove that the unextendible max-relative entropy of a quantum channel does not increase under amortization. Similar results are known for the squashed entanglement of a channel \cite{TGW14}, a channel's  max-relative entropy of entanglement \cite{CM17}, and the max-Rains information of a quantum channel \cite{BW17}. Our proof of this result is strongly based on the approach given in \cite{BW17}, which in turn made use of the developments in \cite{WFD17}.
 
We begin by establishing equivalent forms for the unextendible max-relative entropy of a state and a channel. Let $\overrightarrow{\EXT}_k(A;B)$ denote the cone of all $k$-extendible operators. This set is defined in the same way as the set of $k$-extendible states, but there is no requirement for a $k$-extendible operator to have trace equal to one. Then we have the following alternative expression for the max-relative entropy of unextendibility:

\begin{lemma}\label{lemma-6}
Let $\rho_{AB}\in\mc{D}(\mc{H}_A\otimes\mc{H}_B)$. Then 
\begin{equation}
E^{\max}_{k}(A;B)_{\rho}=\log_2 W_{k}(A;B)_{\rho},
\end{equation}
where 
\begin{equation}
W_k(A;B)_{\rho}\coloneqq \inf_{X_{AB}\in\overrightarrow{\EXT}_k(A;B)}\{\Tr\{X_{AB}\}:\rho_{AB}\leq X_{AB}\}.
\end{equation}
\end{lemma}
\begin{proof}
Employing the definition of $k$-unextendible max-relative entropy, consider that
\begin{align}
&E_{k}^{\max}(A;B)_{\rho}\nonumber \\=& \inf_{\sigma_{AB}\in\kex(A:B)}D_{\max}(\rho_{AB}\Vert \sigma_{AB})\\
=&\log_2 \inf_{\mu,\sigma_{AB}}\{\mu:\rho_{AB}\leq \mu\sigma_{AB},\sigma_{AB}\in\kex(A\!:\!B)\}\\
=& \log_2 \inf_{X_{AB}}\{\Tr\{X_{AB}\}:\rho_{AB}\leq X_{AB}, \notag \\
& \qquad\qquad\qquad\qquad\qquad X_{AB}\in\overrightarrow{\EXT}_k(A;B)\}.
\end{align}
This concludes the proof.
\end{proof}

Let $E^{\max}_{k}(\mc{N})$ denote the unextendible max-relative entropy of a channel $\mc{N}$, as defined in \eqref{eq:unex-channel-def}, but with the generalized divergence $\tf{D}$ replaced
by the max-relative entropy $D_{\max}$.
We can write $E^{\max}_k(\mc{N})$ in an alternate way, by employing similar reasoning as given in the proof of \cite[Lemma 6]{CMW14}:
\begin{multline}\label{equation-67}
E_k^{\max}(\mc{N})\\=\max_{\rho_S\in\mc{D}(\mc{H}_S)}\inf_{\sigma_{SB}\in\kex(S;B)}D_{\max}(\rho_S^{1/2}\Gamma^{\mc{N}}_{SB}\rho^{1/2}_{S}\Vert \sigma_{SB}),
\end{multline}
where $\Gamma^{\mc{N}}_{SB}$ is the Choi operator for the channel $\mc{N}$.

An alternative expression for the unextendible max-relative entropy $E^{\max}_k(\mc{N})$ of the channel $\mc{N}$ is given by the following lemma:

\begin{lemma}\label{lemma-7}
For any quantum channel $\mc{N}_{A\to B}$,
\begin{equation}
E^{\max}_k(\mc{N})=\log_2 \Sigma_k(\mc{N}),
\end{equation}
where 
\begin{equation}
\Sigma_k(\mc{N})=\inf_{Y_{SB}\in\overrightarrow{\EXT}_k(S;B)}\{\norm{\Tr_B\{Y_{SB}\}}_{\infty}:\Gamma^{\mc{N}}_{SB}\leq Y_{SB}\} ,
\end{equation}
and $\Gamma^{\mc{N}}_{SB}$ is the Choi operator for the channel $\mc{N}_{A \to B}$.
\end{lemma}

\begin{proof}
The proof follows by employing \eqref{equation-67} and Lemma~\ref{lemma-6}, and following arguments similar to those needed to prove \cite[Lemma 7]{BW17}, given that $\overrightarrow{\EXT}_k$ is also a cone.
\end{proof}

\begin{proposition}[Amortization inequality]\label{prop:amort}
Let $\rho_{R_AAR_B}$ be a state, and let $\mc{N}_{A\to B}$ be an arbitrary quantum channel. 
Then the following inequality holds for the $k$-unextendible max-relative-entropy of a channel $\mc{N}$:
\begin{equation}
E^{\max}_k(R_A;BR_B)_{\omega} \leq E^{\max}_k(R_AA;R_B)_{\rho}+ E^{\max}_{k}(\mc{N}),
\end{equation}
where $\omega_{R_ABR_B}\coloneqq \mc{N}_{A\to B}(\rho_{R_AAR_B})$.   
\end{proposition}

\begin{proof}
We adapt the proof steps of \cite[Proposition 8]{BW17} to show that amortization does not enhance the unextendible max-relative entropy of an arbitrary channel. 

By removing logarithms and applying Lemmas~\ref{lemma-6} and \ref{lemma-7}, the desired inequality is equivalent to the following one:
\begin{equation}\label{eq:w-omega-ineq}
W_k(R_A;B R_B)_{\omega}\leq W_k(R_A A; R_B)_{\rho}\cdot \Sigma_k(\mc{N}),
\end{equation}
and so we aim to prove this one. Exploiting the identity in Lemma~\ref{lemma-6}, we find that 
\begin{equation}
W_k(R_AA; R_B)_{\rho}=\inf \Tr\{C_{R_AAR_B}\},
\end{equation}
subject to the constraints 
\begin{align}
C_{R_AAR_B} &\in \overrightarrow{\EXT}_k(R_AA;B),\\
C_{R_AAR_B}&\geq \rho_{R_AAR_B},
\end{align}
while the identity in Lemma~\ref{lemma-7} gives that 
\begin{equation}
\Sigma_k(\mc{N})=\inf \norm{\Tr_{B}\{Y_{SB}\}}_\infty,
\end{equation}
subject to the constraints 
\begin{align}
Y_{SB} &\in \overrightarrow{\EXT}_k(S;B),\\
 Y_{SB}&\geq \Gamma^\mc{N}_{SB}.\label{eq:choi-bi-b}
\end{align}
The identity in Lemma~\ref{lemma-6} implies that the left-hand side of \eqref{eq:w-omega-ineq} is equal to 
\begin{equation}
W_k(R_A;BR_B)_\omega=\inf \Tr\{E_{R_ABR_B}\},
\end{equation}
subject to the constraints
\begin{align}
E_{R_ABR_B}&\in \overrightarrow{\EXT}_k(R_A;BR_B),\label{eq:rains-sdp-ef}\\
E_{R_ABR_B}&\geq \mc{N}_{A\to B}(\rho_{R_AAR_B}) \label{eq:rains-sdp-channel-ef}.
\end{align}

Once we have these optimizations, we can now show that the inequality in \eqref{eq:w-omega-ineq} holds by making an appropriate choice for $E_{R_ABR_B}$.  Let $C_{R_AAR_B}$ be optimal for $W_k(R_AA;R_B)_\rho$, and let $Y_{R_ABR_B}$ be optimal for $\Sigma(\mc{N})$. Let $| \Gamma \rangle_{SA}$ be the maximally entangled vector. Choose
\begin{align}
E_{R_ABR_B}&=\langle \Gamma |_{SA}C_{R_AAR_B}\otimes Y_{SB}| \Gamma \rangle_{SA}.
\end{align}
We need to prove that $E_{R_ABR_B}$ is feasible for $W_{k}(R_A;BR_B)_{\omega}$. To this end, we have
\begin{align}
 &\langle \Gamma |_{SA}C_{R_AAR_B}\otimes Y_{SB})| \Gamma \rangle_{SA}\nonumber\\&\quad\geq  \langle \Gamma |_{SA}\rho_{R_AAR_B}\otimes \Gamma^{\mc{N}}_{SB})| \Gamma \rangle_{SA}\nonumber\\&\quad=\mc{N}_{A\to B}(\rho_{R_AAR_B}).
\end{align}

Now, since $C_{R_AAR_B}\in\overrightarrow{\EXT}_k(R_AA;R_B)$
and $Y_{SB}\in\overrightarrow{\EXT}_k(S;B)$, it immediately follows that $\langle \Gamma |_{SA}C_{R_AAR_B}\otimes Y_{SB})| \Gamma \rangle_{SA} 
\in\overrightarrow{\EXT}_k(R_A;R_B B)
$.

Consider that
\begin{align}
\Tr\{E_{R_ABR_B}\} & = \Tr\{\langle \Gamma |_{SA}(C_{R_AAR_B}\otimes Y_{SB})| \Gamma \rangle_{SA}\} \notag \\
& = \Tr\{C_{R_AAR_B} T_{A}( Y_{AB})\} \notag \\
& = \Tr\{C_{R_AAR_B} T_{A}(\Tr_{B}\{ Y_{AB})\}\} \notag \\
& \leq \Tr\{C_{R_AAR_B}\}\norm{ T_{A}(\Tr_{B}\{ Y_{AB}\})}_\infty \notag \\
& = \Tr\{C_{R_AAR_B}\}\norm{ \Tr_{B}\{ Y_{AB}\}}_\infty \notag \\
&=W_k(R_AA;R_B)_\rho\cdot \Sigma(\mc{N}).
\end{align}
The inequality is a consequence of H\"{o}lder's inequality \cite{Bha97}. The final equality follows because the spectrum of a positive semi-definite operator is invariant under the action of a full transpose (note, in this case, $\T_{A}$ is the full transpose as it acts on reduced positive semi-definite operator $Y_{A}$).

Therefore, we can infer that our choice of $E_{R_ABR_B}$ is feasible for $W_k(R_A;BR_B)_\omega$. Since $W_k(R_A;BR_B)_\omega$ involves a minimization over all  $E_{R_ABR_B}$ satisfying \eqref{eq:rains-sdp-ef} and \eqref{eq:rains-sdp-channel-ef}, this concludes our proof of \eqref{eq:w-omega-ineq}.
\end{proof}

\begin{remark}\label{rem:cov-sim}
We briefly remark here that if a channel $\mc{N}_{A\to B}$ can be simulated by the action of a $k$-extendible channel $\mc{K}_{ARB'\to B}$ on the channel input $\rho_A$ as well as a resource state $\omega_{RB'}$ (i.e., $\mc{N}_{A\to B}(\rho_A) = \mc{K}_{ARB'\to B}(\rho_A \otimes \omega_{RB'})$),  then the $k$-unextendible divergence of that channel does not increase under amortization, for divergences that are subadditive with respect to tensor-product states. This is a special case of the more general observation put forward in \cite[Section~7]{KW17a} for general resource theories.
\end{remark}
\onecolumngrid

\section{Exploiting symmetries} \label{appendix:exploiting-symmetries}

In this appendix, we provide the following Lemma~\ref{lemma:covariant_channels}, similar to Proposition~2 of \cite{TWW17}, which is helpful in determining the form of the state that optimizes the unextendible generalized channel divergence of a quantum channel that has some symmetry. %For completeness, we give a proof in Appendix~\ref{appendix:exploiting-symmetries}.
Its proof is identical to that given for \cite[Proposition~2]{TWW17}, but we give it here for completeness.

\begin{lemma}\label{lemma:covariant_channels}
Let $\mc{N}_{A\to B}$ be a covariant channel with respect to a group $\mc{G}$. Let $\rho_A\in\mc{D}(\mc{H}_A)$, and let $\psi^{\rho}_{RA}$ be a purification for it. Define $\rho_{RB}\coloneqq \mc{N}_{A\to B}(\psi^{\rho}_{RA})$ and $\bar{\rho}_{A}\coloneqq \frac{1}{|\mc{G}|}\sum_{g\in\mc{G}}U_A(g)\rho_A U^\dag_A(g)$. Let $\phi^{\bar{\rho}}_{RA}$ be a purification of $\bar{\rho}_A$ and $\bar{\rho}_{RB}\coloneqq \mc{N}_{A\to B}(\phi^{\bar{\rho}}_{RA})$. Then
\begin{equation}
\tf{E}_k(R;B)_{\bar{\rho}}\geq \tf{E}_k(R;B)_{\rho}.
\end{equation}
\end{lemma}

%The following lemma is helpful in determining the form of a state that optimizes the unextendible generalized channel divergence of a quantum channel that has some symmetry. 

%\begin{lemma}\label{lemma:covariant_channels1}
%Let $\mc{N}_{A\to B}$ be a covariant channel with respect to a group $\mc{G}$. Let $\rho_A\in\mc{D}(\mc{H}_A)$, and let $\psi^{\rho}_{RA}$ be a purification for it. Define $\rho_{RB}\coloneqq \mc{N}_{A\to B}(\psi^{\rho}_{RA})$ and $\bar{\rho}_{A}\coloneqq \frac{1}{|\mc{G}|}\sum_{g\in\mc{G}}U_A(g)\rho_A U^\dag_A(g)$. Let $\phi^{\bar{\rho}}_{RA}$ be a purification of $\bar{\rho}_A$ and $\bar{\rho}_{RB}\coloneqq \mc{N}_{A\to B}(\phi^{\bar{\rho}}_{RA})$. Then
%\begin{equation}
%\tf{E}_k(R;B)_{\bar{\rho}}\geq \tf{E}_k(R;B)_{\rho}.
%\end{equation}
%\end{lemma}

\begin{proof}
Define
\begin{equation}
\ket{\phi}_{PRA}\coloneqq \frac{1}{\sqrt{|\mc{G}|}}\sum_{g\in\mc{G}}\ket{g}_P[I_R\otimes U_A(g)]\ket{\psi}_{RA},
\end{equation}
so that $\phi_{PRA}$ is a purification of $\bar{\rho}_{A}$. 
Let $\tau_{PRB}\in\kex(PR\!:\!B)$, and, given that a local channel is a $k$-extendible channel, observe that
\begin{align}
\sum_{g\in\mc{G}}| g\rangle \!\langle g |_P\tau_{PRB}| g\rangle \!\langle g |_P=\sum_{g\in\mc{G}}p(g)| g\rangle \!\langle g |_P\otimes\tau^g_{RB}\in\kex(PR\!:\!B),
\end{align}
where $\tau^g_{RB}=\frac{1}{p(g)}\bra{g}\!\tau_{PRB}\ket{g}_P$ and $p(g)=\Tr\{\bra{g}\tau_{PRB}\ket{g}_P\}$. Then
\begingroup
\allowdisplaybreaks
\begin{align}
&\tf{D}\!\(\mc{N}_{A\to B}(\phi_{PRA})\left\Vert \tau_{PRB}\)\right.\nonumber\\
&= \tf{D}\!\left.\left(\mc{N}_{A\to B}\!\left(\sum_{g, g'\in\mc{G}}\frac{1}{|\mc{G}|}| g\rangle \!\langle g^\prime |_P\otimes [I_R\otimes U_A(g)]\psi^\rho_{RA}[I_R\otimes U_A^\dag(g^\prime)]\right)\right\Vert \tau_{PRB}\right)\\
&\geq \tf{D}\!\left.\left(\sum_{g\in\mc{G}}\frac{1}{|\mc{G}|}| g\rangle \!\langle g |_P\otimes\mc{N}_{A\to B}
\!\left(U_A(g)\psi^\rho_{RA}U^\dag_A(g)\right)\right\Vert \sum_{g\in\mc{G}}p(g)| g\rangle \!\langle g |_P\otimes\tau^g_{RB}\right)\\
&=\tf{D}\!\left.\left(\sum_{g\in\mc{G}}\frac{1}{|\mc{G}|}| g\rangle \!\langle g |_P\otimes V_B(g)\mc{N}_{A\to B}
\!\(\psi^\rho_{RA}\right)V^\dag_B(g)\right\Vert \sum_{g\in\mc{G}}p(g)| g\rangle \!\langle g |_P\otimes\tau^g_{RB}\right)\\
&=\tf{D}\!\left.\left(\sum_{g\in\mc{G}}\frac{1}{|\mc{G}|}| g\rangle \!\langle g |_P\otimes \mc{N}_{A\to B}\!\(\psi^\rho_{RA}\right)\right\Vert \sum_{g\in\mc{G}}p(g)| g\rangle \!\langle g |_P\otimes V_B^\dag(g)\tau^g_{RB}V_B(g)\right)\\
&\geq \tf{D}\!\left(\mc{N}_{A\to B}(\psi^\rho_{RA})\left\Vert \sum_{g\in\mc{G}}p(g)V_B^\dag(g)\tau^g_{RB}V_B(g)\right)\right.\\
&\geq \inf_{\tau^\prime_{RB}\in\kex(R;B)}\tf{D}(\mc{N}_{A\to B}(\psi_{RA})\Vert \tau^\prime_{RB})\\
&=\tf{E}_k(R;B)_{\rho}.
\end{align}
\endgroup
The first inequality follows because any general divergence is monotonically non-increasing under the action of a quantum channel, which in this case is the completely dephasing channel $(\cdot) \to \sum_{g\in\mc{G}}| g\rangle \!\langle g |_P(\cdot)| g\rangle \!\langle g |_P$. The second equality follows because the channel $\mc{N}$ is covariant. To arrive at the third equality, we use the fact that any generalized divergence is invariant under the action of isometries. To get the second inequality, we apply the partial trace over the classical register $P$, which is a quantum channel. 
The last inequality follows because the state $\sum_{g\in\mc{G}}p(g)V_B^\dag(g)\tau^g_{RB}V_B(g)$ is $k$-extendible, given that it arises from the action of a 1W-LOCC channel on the $k$-extendible state $\tau_{PRB}$.
Noticing that the chain of inequalities holds for arbitrary $\tau_{PRB}\in\kex(PR;B)$, we can then take an infimum over all possible $\tau_{PRB}\in\kex(PR;B)$, and we arrive at the following inequality:
\begin{equation}
\tf{E}_k(PR;B)_{\mc{N}(\phi)} \geq \tf{E}_k(R;B)_{\rho}
\end{equation}
The desired inequality in the statement of the lemma then follows because all purifications of a given state are related by an isometry acting on the purifying system, and the unextendible generalized divergence is invariant under the action of a local isometry.
\end{proof}

\end{document}